\documentclass{article}

\usepackage{graphics}
\usepackage{graphicx}
\usepackage{geometry}
\usepackage{url}
\usepackage{algorithm}
\usepackage{mcaption}
\usepackage{amsmath}
\usepackage{amssymb}
\usepackage{bookmark}
\usepackage[dvipsnames]{color}
\usepackage{changebar}
\usepackage{soul}
\definecolor{mygreen}{RGB}{59, 168, 59}

\usepackage{caption}
\usepackage{tikz}
\usepackage{subcaption} 
\usetikzlibrary{positioning}
\usetikzlibrary{decorations.pathmorphing}
\tikzset{
simpleNodes/.style={
  circle,
  inner sep=0pt,
  text width=5mm,
  text height=3mm,
  align=center,
  draw=black,
  fill=black!5
  },
  simpleNodesBigger/.style={
  circle,
  inner sep=0pt,
  text width=6mm,
  align=center,
  draw=black,
  fill=black!5
  },
  simpleBlankNodes/.style={
  circle,
  inner sep=0pt,
  text width=5mm,
  align=center,
  draw=white,
  fill=white!5
  },
  tinyDot/.style={
circle,
  inner sep=0pt,
  text width=1mm,
  text height=1mm,
  align=center,
  draw=black,
  fill=black!100
  },
thinLine/.style={black, thin},
standardLine/.style={black, semithick},
standardLineDashed/.style={black, dashed, semithick},
dashedLine/.style={black, dashed, thin},
thickLine/.style={black, ultra thick},
zigzag1/.style={decorate, decoration = {coil, aspect=0, segment length=2.5mm, amplitude = 1mm}},
thinLineArrow/.style={->, black, thin}

\usepackage{cancel}
}

\usepackage[colorinlistoftodos]{todonotes}

\usetikzlibrary{patterns}
\usepackage{pgfgantt}
\usepackage{color}
\definecolor{mygray}{gray}{0.6}
\usepackage{amsthm}
\usepackage[numbers]{natbib}
\usepackage{multirow, bigdelim}
\usepackage{pdflscape}
\graphicspath{ {./images/} }

\usepackage{placeins}
\usepackage[shortlabels]{enumitem}
\usepackage{mdframed}
\usepackage{siunitx}
\usepackage{amssymb,amsmath,amsthm}
\usepackage{upgreek}

\usepackage{booktabs}
\usepackage{siunitx}
\usepackage{algpseudocode}

\newtheorem{theorem}{Theorem}[section]
\newtheorem{definition}[theorem]{Definition}
\newtheorem{lemma}[theorem]{Lemma}
\newtheorem{corollary}[theorem]{Corollary}
\newtheorem{prop}[theorem]{Proposition}

\newdimen\origiwstr
\newdimen\origiwshr
\makeatletter
\newcommand{\fixedspaceword}[1]{%
  \origiwstr=\fontdimen3\font
  \origiwshr=\fontdimen4\font
  \fontdimen3\font=\z@
  \fontdimen4\font=\z@
  #1%
  \fontdimen3\font=\origiwstr
  \fontdimen4\font=\origiwshr
}
\mdfdefinestyle{mdf1}{innerleftmargin=1cm,innerrightmargin=1cm, leftmargin=0.5cm, rightmargin=0.5cm}
\mdfdefinestyle{mdf2}{linewidth=0pt, innerleftmargin=1cm,innerrightmargin=1cm, leftmargin=0.5cm, rightmargin=0.5cm}

\newcommand{\graphxStd}{1.5}

\usepackage{hyperref}
\usepackage[nomain, acronym]{glossaries}

\newcolumntype{L}[1]{>{\raggedright\let\newline\\\arraybackslash\hspace{0pt}}m{#1}}
\newcolumntype{C}[1]{>{\centering\let\newline\\\arraybackslash\hspace{0pt}}m{#1}}
\newcolumntype{R}[1]{>{\raggedleft\let\newline\\\arraybackslash\hspace{0pt}}m{#1}}

\usepackage{complexity}
\makeglossaries
\newacronym{rmsr}{\sc rmsr}{The problem of finding a rank-maximal stable matching in \acrshort{sr}}
\newacronym{gensr}{\sc gensr}{The problem of finding a generous stable matching in \acrshort{sr}}
\newacronym{sr}{\sc sr}{Stable Roommates problem}
\newacronym{ip}{\sc IP}{Integer Programming}
\newacronym{smti}{\sc SMTI}{Stable Marriage problem with Ties and Incomplete lists}
\newacronym{bctcs}{\sc bcb}{British Colloquium for Theoretical Computer Science}
\newacronym{cp}{\sc cp}{Constraint Programming}
\newacronym{sm}{\sc sm}{Stable Marriage problem}
\newacronym{smi}{\sc smi}{Stable Marriage problem with Incomplete lists}
\newacronym{hr}{\sc hr}{Hospitals/Residents problem}
\newacronym{sat}{\sc sat}{Satisfiability Problems}
\newacronym{cbc}{CBC}{COIN-OR Branch and Cut}
\newacronym{nrmp}{NRMP}{National Resident Matching Program}
\newacronym{ha}{\sc ha}{House Allocation problem}
\newacronym{hat}{\sc hat}{House Allocation problem with Ties}
\newacronym{cha}{\sc cha}{Capacitated House Allocation problem}
\newacronym{chat}{\sc chat}{Capacitated House Allocation problem with Ties}

\usepackage{cancel}

\usepackage{authblk}
\title{Two-sided profile-based optimality in the stable marriage problem}
\author[1]{Frances Cooper\footnote{Supported by an Engineering and Physical Sciences Research Council Doctoral Training Account}}
\author[2]{David Manlove\footnote{Supported by Engineering and Physical Sciences Research Council grant EP/P028306/01}}
\affil[1]{University of Glasgow, Scotland, \href{mailto:f.cooper.1@research.gla.ac.uk}{f.cooper.1@research.gla.ac.uk}}
\affil[2]{University of Glasgow, Scotland, \href{mailto:david.manlove@glasgow.ac.uk}{david.manlove@glasgow.ac.uk}}
\date{}                     
\begin{document}


\maketitle










\pagenumbering{arabic}


\abstract{We study the problem of finding ``fair'' stable matchings in the Stable Marriage problem with Incomplete lists ({\sc smi}).  In particular, we seek stable matchings that are optimal with respect to \emph{profile}, which is a vector that indicates the number of agents who have their first-, second-, third-choice partner, etc.  In a \emph{rank maximal stable matching}, the maximum number of agents have their first-choice partner, and subject to this, the maximum number of agents have their second-choice partner, etc., whilst in a \emph{generous stable matching} $M$, the minimum number of agents have their $d$th-choice partner, and subject to this, the minimum number of agents have their $(d-1)$th-choice partner, etc., where $d$ is the maximum rank of an agent's partner in $M$.  Irving et al. \cite{ILG87} presented an $O(nm^2\log n)$ algorithm for finding a rank-maximal stable matching, which can be adapted easily to the generous stable matching case, where $n$ is the number of men / women and $m$ is the number of acceptable man-woman pairs.  An $O(n^{0.5}m^{1.5})$ algorithm for the rank-maximal stable matching problem was later given by Feder \cite{Fed89}.  However these approaches involve the use of weights that are in general exponential in $n$, potentially leading to overflow or inaccuracies upon implementation.  In this paper we present an $O(nm^2\log n)$ algorithm for finding a rank-maximal stable matching using a \emph{vector-based} approach that involves weights that are polynomially-bounded in $n$. We show how this approach has a far reduced memory requirement (an estimated $100$-fold improvement for instances with $100, 000$ men or women) when compared to Irving et al.'s algorithm above.  Additionally, we show how to adapt our algorithm for the generous case to run in $O(\min\{m, nd\}^2 d \log n)$ time, where $d$ is the degree of a minimum regret stable matching.  We conduct an empirical evaluation to compare rank-maximal and generous stable matchings over a range of measures.  In particular, we observe that a generous stable matching is typically considerably closer than a rank-maximal stable matching in terms of the egaliatarian and sex-equality optimality criteria. In addition to investigating profile-based optimality in \acrshort{smi}, we also examine the complexity of the problem of finding profile-based optimal stable matchings in the \acrlong{sr} (\acrshort{sr}).}

\vspace{-2mm}
\section{Introduction}
\label{sm_rm_intro}

\subsection{Background}
The \acrlong{sm} (\acrshort{sm}) was first introduced in \citeauthor{GS62}'s seminal paper ``College Admission and the Stability of Marriage". In an instance of \acrshort{sm} we have two sets of agents,  men and women (of equal number, henceforth $n$), such that each man ranks every woman in strict preference order, and vice versa. An extension to \acrshort{sm}, known as the \acrlong{smi} (\acrshort{smi}) allows each man (woman) to rank a \emph{subset} of women (men). Let $m$ denote the total length of all preference lists.

Generalisations of \acrshort{smi} in which one or both sets of agents may be multiply assigned have been extensively applied in the real-world. The \emph{\acrlong{nrmp}} (\acrshort{nrmp}) is a long standing matching scheme in the US (beginning in 1952) which assigns graduating medical students to hospitals \citep{PR95}. Other examples include the assignment of children to schools in Boston \citep{APRS05} and the allocation of high-school students to university places in China \citep{Zha11}.

\citet{GS62} described linear time algorithms to find a stable matching in an instance of \acrshort{smi}. These classical algorithms find either a man-optimal (or woman-optimal) stable matching in which every man (woman) is assigned to their best partner in any stable matching and every woman (man) is assigned to their worst partner in any stable matching. Favouring one set of agents over the other is often undesirable and so we look at the notion of a ``fair" matching in which the happiness of both sets of agents is taken into account.

There may be many stable matchings in any given instance of \acrshort{smi}, and there are several different criteria that may be used to describe an optimal or ``fair" stable matching. The \emph{rank} of an agent $a$ in a stable matching $M$ is the position $a$'s partner on $a$'s preference list, while the \emph{degree} $d$ of $M$ is the highest rank of any agent in $M$. We might wish to limit the number of agents with large rank. A \emph{minimum regret} stable matching is a stable matching that minimises $d$ and can be found in $O(m)$ time \cite{Gus87}. Another type of optimality criteria, uses an arbitrary weight function to find a \emph{minimum (maximum) weight} stable matching, which is a stable matching that has minimum (maximum) weight among the set of all stable matchings. A special case of this is known as the \emph{egalitarian} stable matching which minimises the sum of ranks of all agents. \citet{ILG87} gave an algorithm to find an egalitarian stable matching in $O(m^2)$ time and discussed how to generalise their method to the minimum (and maximum) weight stable marriage problem. Let $K$ be the weight of an optimal stable matching. \citet{Fed89} later improved on the time complexities detailed above showing that any minimum weight stable matching (including an egalitarian stable matching), where $K \leq m$, may be found in $O(m^{1.5})$ time using weighted \acrshort{sat}, where $m$ is the number of possible man-woman pairs. In the general case where $m < K$ and $K = O(m^{c})$, and $c$ is a constant value, this rises to $O(nm \log m)$ time.  A \emph{sex-equal} stable matching seeks to minimise the difference in the sum of ranks between men and women. Finding a sex-equal stable matching was shown to be NP-hard \cite{Kat93}. A \emph{median} stable matching, defined formally in Section \ref{sec_formaldefs}, describes a stable matching in which each agent gains their median partner (if the partners of an agent for all stable matchings were lined up in order of preference) \cite{TS98}. Computing the set of median stable matchings is \#$\P$-hard \cite{Che08}.

Other notions of fairness involve the \emph{profile} of a matching which is a vector representing the number of agents assigned to their first, second, third choices etc., in the matching. A \emph{rank-maximal} stable matching $M$ is a stable matching whose profile is lexicographically maximum, ie. $M$ maximises the number of agents assigned to their first choice and, subject to that their second choice, and so on. Meanwhile, a \emph{generous} stable matching $M$ is a stable matching whose reverse profile is lexicographically minimum, ie. $M$ minimises the number of agents with rank $d$, and subject to that, rank $d-1$, and so on. Profile-based optimality such as rank-maximality or the generous criteria provide guarantees that do not exist with other optimality criteria giving a distinct advantage to these approaches in certain scenarios. 

\citet{ILG87} describe the use of weights that are exponential in $n$ (henceforth \emph{exponential weights}) in order to find a rank-maximal stable matching using a maximum weight approach. This requires an additional factor of $O(n)$ time complexity to take into account calculations over exponential weights, giving an overall time complexity of $O(nm^2 \log n)$ \footnote{\citet{ILG87} actually state a time complexity of $O(nm^2 \log n \log n)$, however, we believe that this time complexity bound is somewhat pessimistic and that a bound of $O(nm^2 \log n)$ applies to this approach.}. 
\citeauthor{ILG87}'s approach (described in more detail in Section \ref{sec_ran_max_exp_weights}) requires a Max Flow algorithm to be used. \citet{ILG87} stated that the strongly polynomial $O(m^2 \log n)$ Sleator-Tarjan algorithm \citep{ST83} was the best option (at the time of writing). The Sleator-Tarjan algorithm \citep{ST83} is an adapted version of \citeauthor{Din70}'s algorithm \citep{Din70} and finds a maximum flow in a network in $O(|V||E| \log |V|)$ time. Since $|V| \leq m$, $|E| \leq m$ and $O(\log m)$ is equivalent to $O(\log n)$ \citep{GI89}, this translates to $O(m^2 \log n)$ for the maximum weight stable matching problem and an overall time complexity of $O(nm^2 \log n)$ for the rank-maximal stable matching problem. However in \citeyear{Orl13} \citet{Orl13} described an improved strongly polynomial Max Flow algorithm with an $O(|V||E|)$ (translating to $O(m^2)$) time complexity, giving a total overall time complexity for finding a rank-maximal stable matching of $O(nm^2)$. \citeauthor{Fed92}'s weighted \acrshort{sat} approach \cite{Fed92} has an overall $O(n^{0.5}m^{1.5})$ time complexity for finding a rank-maximal stable matching.
Neither \citet{ILG87} nor \citet{Fed92} considered generous stable matchings, however, a generous stable matching may be found in a similar way to a rank-maximal stable matching with the use of exponential weights.

\subsection{Motivation} 
\label{sm_rm_motivation_sec}
For the rank-maximal stable matching problem, \citet{ILG87} suggest a weight of $n^{n-i}$ for each agent assigned to their $i$th choice and a similar approach can be taken to find a generous stable matching as we demonstrate later in this paper. In both the rank-maximal and generous cases, the use of exponential weights introduces the possibility of overflow and accuracy errors upon implementation. 
This may occur as a consequence of limitations of data types: in Java for example, the \emph{int} and \emph{long} primitive types restrict the number of integers that can be represented, and the \emph{double} primitive type may introduce inaccuracies when the number of significant figures is greater than $15$. Using a weight of $n^{n-i}$ for each agent assigned to their $i$th choice as above, it may be that we need to distribute $n$ capacities of size $n^{n-1}$ across the network \cite{ILG87}. As a theoretical example the \emph{long} data type has a maximum possible value of $2^{63}-1 < 10^{19}$ \cite{Ora19}. Since $16^{15}< 10^{19} < 17^{16}$, when we are dealing with flows or capacities of order $n^{n-1}$, the largest $n$ possible without risking errors is $16$. However, alternative data structures such as Java's \emph{BigInteger} do allow an arbitrary limit on integer size \cite{Ora18}, by storing each number as an array of \emph{int}s to the base $2^{31} - 1$ (the maximum \emph{int} value), 
meaning we are more likely to be dependent on the size of computer memory than any data type limits.

When looking for a rank-maximal or generous stable matching, we describe an alternative approach to finding a maximum flow through a network that does not require exponential weights. This approach is based on using polynomially-bounded weight vectors (henceforth \emph{vector-based weights}) for edge capacities rather than exponential weights. On the surface, performing operations over vector-based weights rather than over equivalent exponential weights, would appear not to improve the time or space complexity of the algorithm, since an exponential weight may naturally be stored as an equivalent array of integers in memory, as in the \emph{BigInteger} case above. However, vector-based weights allow us to explore vector compression that is unavailable in the exponential case. One form of lossless vector compression saves the index and value of each non-zero vector element. This type of vector compression is used in our experiments in Section \ref{sm_rm_sec_exps} to show that for randomly-generated instances of size $n=1000$, we are able to store a network with vector-based weights using approximately $10$ times less space than one stored with the equivalent exponential weights. Indeed extrapolating to $n=100,000$ we achieve an approximate  factor of $100$ improvement using vector-based weights, with the space required to store a network using exponential weights nearing $1$GB. We also show that for a specific instance of size $n=100,000$, the space required to store exponential weights of a network was over $10$GB, whereas the vector-based weights were over $100,000$ times less costly at $0.64$MB.

Combining these space requirement calculations with the fact that the time complexity of \citeauthor{ILG87}'s \cite{ILG87} $O(nm^2\log n)$ algorithm to find a rank-maximal stable matching is dominated substantially by the maximum flow algorithm (no other part taking more than $O(m)$ time), it is arguably important to ensure that the network is as small as possible and fits comfortably in RAM.

\subsection{Contribution}
\label{smi_rm_sec_contribution}

In this paper we present an $O(nm^2 \log n)$ algorithm to find a rank-maximal stable matching in an instance of \acrshort{smi} using a vector-based weight approach rather than using exponential weights. We also show that a similar process can be used to find a generous stable matching in $O(\min\{m, nd\}^2 d \log n)$ time, where $d$ is the degree of a minimum regret stable matching. Finally, we show that the problems of finding rank-maximal and generous stable matchings in \acrshort{sr} are $\NP$-hard. In addition to theoretical contributions we also run experiments using randomly-generated \acrshort{sm} instances. In these experiments we compare rank-maximal and generous stable matchings over a range of measures (cost, sex-equal score, degree, number of agents obtaining their first choice and number of agents who obtain a partner in the lower a\% of their preference list). An example of this final measure is as follows. If $n = 200$ and $a = 50$ then we record the number of agents who obtain a partner between their $101$st and $200$th choice inclusive. We additionally compare these profile-based optimal stable matchings with median stable matchings. The median criterion is somewhat unique in that its definition is not based on cost, degree or profile. We were interested in determining whether, in practice, a median stable matching more closely approximates a rank-maximal or a generous stable matching. In these experiments, we find that a generous stable matching typically outperforms both a rank-maximal and a median stable matching when considering cost and sex-equal score measures, and that a median stable matching more closely approximates a generous stable matching in practice.

\subsection{Structure of the paper}
Section \ref{sec_formaldefs} gives formal definitions of \acrshort{smi} and \acrshort{sr}, and defines various types of optimal stable matchings. Section \ref{sec_ran_max_exp_weights} gives a description of \citeauthor{ILG87}'s \citep{ILG87} method for finding a rank-maximal stable matching using exponential weights. Sections \ref{sm_rm_combinatoricApproach} and \ref{sec_gen} describe the new approach to find a rank-maximal stable matching and a generous stable matching respectively, without the use of exponential weights. Complexity results for rank-maximal and generous stable matchings in \acrshort{sr} are presented in Section \ref{sec_rm_sr}. Our experimental evaluation is presented in Section \ref{sm_rm_sec_exps}, whilst future work is discussed in Section \ref{sec:sm_rm_fw}.

\subsection{Related work}  
Work undertaken by \citet{CMS08} describes the happiness of an agent $a$ in a stable matching $M$, defined $s(a,M)$, as a map from all agents over a given matching to $\mathbb{R}$. The map $s(a,M)$ is said to have the \emph{independence property} if it is only reliant upon information contained in $M(a)$. The \acrlong{hr} (\acrshort{hr}) is a more general case of \acrshort{smi} in which women may be assigned more than one man.  \citet{CMS08} provide an algorithm for the family of variants of \acrshort{hr} incorporating happiness functions that exhibit the independence property, to calculate egalitarian and minimum regret stable matchings. For the case that we are given an instance of \acrshort{smi}, this algorithm has a time complexity of $O(n^2f(c) + n^4)$ where $f(c)$ is the time it takes to calculate the weight of a matching. It is worth noting that the $n^4$ term of this time complexity is due to \citeauthor{ILG87}'s \citep{ILG87} method of finding a minimum weight stable matching. This method also requires the use of exponential weights which would be problematic for the reasons outlined above.

The \acrlong{ha} (\acrshort{ha}) is an extension of \acrshort{smi} in which women do not have preferences over men. The \acrlong{chat} (\acrshort{chat}) is an extension of \acrshort{ha} in which women may be assigned more than one man and men may be indifferent between one or more women on their preference list. With one-sided preferences the notion of stability does not exist. Rank-maximality however, may be described in an analogous way to \acrshort{smi}, and there is an $O(\min(n+d,d \sqrt(n))m)$ algorithm to find the rank-maximal matching in an instance of \acrshort{chat} \citep{Sng08}, where $m$ is the total length of men's preference lists and $d$ is the degree of the matching. We may also seek to find a generous maximum matching in which the most number of men are assigned as possible and then subject to that we use a generous criteria analogous to the \acrshort{smi} case. There is an $O(dn^2m)$ algorithm to find the generous maximum matching in \acrshort{chat} \citep{Sng08}.

\section{Preliminary definitions and results}
\label{sec_formaldefs}

\subsection{Formal definition of SMI}
The \acrlong{smi} (\acrshort{smi}) comprises a set of men $U$ and a set of women $W$. Each man ranks a subset of women in preference order and vice versa. A man $m_i$, finds a woman $w_j$ \emph{acceptable} if $w_j$ appears on $m_i$'s preference list and vice versa. A \emph{matching} $M$ in this context is an assignment of men to women such that no man or woman is assigned to more than one person, and if $(m_i,w_j) \in M$, then $m_i$ finds $w_j$ acceptable and $w_j$ finds $m_i$ acceptable. An example \acrshort{smi} instance $I_0$ with $8$ men and women is taken from \citeauthor{GI89}'s book \citep[p. 69]{GI89} and is given as Figure \ref{sm_rm:ex_preflists}. Let $M(m_i)$ denote $m_i$'s assigned partner in $M$, and similarly, let $M(w_j)$ denote $w_j$'s assigned partner in $M$. A matching $M$ is \emph{stable} if there is no man-woman pair $(m_i,w_j)$ who would rather be assigned to each other than to their assigned partners in $M$ (if any). Denote by $\mathcal{M}_S$, the set of all stable matchings. By the \emph{``Rural Hospitals" Theorem} \citep{GS85}, the same set of men and women are assigned in all stable matchings of $\mathcal{M}_S$. For the remainder of this paper, we assume \acrshort{smi} instances have been pre-processed to remove men and women unassigned in any stable matching. Thus we can assume that the number of men and women is equal and we denote this number by $n$.

\begin{figure}[]
\centering
   \begin{subfigure}[t]{0.4\textwidth}
  Men's preferences:\\
$m_1$: $w_5$ $w_7$ $w_1$ $w_2$ $w_6$ $w_8$ $w_4$ $w_3$\\
$m_2$: $w_2$ $w_3$ $w_7$ $w_5$ $w_4$ $w_1$ $w_8$ $w_6$\\
$m_3$: $w_8$ $w_5$ $w_1$ $w_4$ $w_6$ $w_2$ $w_3$ $w_7$\\
$m_4$: $w_3$ $w_2$ $w_7$ $w_4$ $w_1$ $w_6$ $w_8$ $w_5$\\
$m_5$: $w_7$ $w_2$ $w_5$ $w_1$ $w_3$ $w_6$ $w_8$ $w_4$\\
$m_6$: $w_1$ $w_6$ $w_7$ $w_5$ $w_8$ $w_4$ $w_2$ $w_3$\\
$m_7$: $w_2$ $w_5$ $w_7$ $w_6$ $w_3$ $w_4$ $w_8$ $w_1$\\
$m_8$: $w_3$ $w_8$ $w_4$ $w_5$ $w_7$ $w_2$ $w_6$ $w_1$\\
  \end{subfigure}
  \hspace{1cm}
   \begin{subfigure}[t]{0.4\textwidth}
  Women's preferences:\\
$w_1$: $m_5$ $m_3$ $m_7$ $m_6$ $m_1$ $m_2$ $m_8$ $m_4$\\
$w_2$: $m_8$ $m_6$ $m_3$ $m_5$ $m_7$ $m_2$ $m_1$ $m_4$\\
$w_3$: $m_1$ $m_5$ $m_6$ $m_2$ $m_4$ $m_8$ $m_7$ $m_3$\\
$w_4$: $m_8$ $m_7$ $m_3$ $m_2$ $m_4$ $m_1$ $m_5$ $m_6$\\
$w_5$: $m_6$ $m_4$ $m_7$ $m_3$ $m_8$ $m_1$ $m_2$ $m_5$\\
$w_6$: $m_2$ $m_8$ $m_5$ $m_3$ $m_4$ $m_6$ $m_7$ $m_1$\\
$w_7$: $m_7$ $m_5$ $m_2$ $m_1$ $m_8$ $m_6$ $m_4$ $m_3$\\
$w_8$: $m_7$ $m_4$ $m_1$ $m_5$ $m_2$ $m_3$ $m_6$ $m_8$\\
  \end{subfigure}
  \caption{\acrshort{smi} instance $I_0$ \citep[p. 69]{GI89}}
  \label{sm_rm:ex_preflists}
    \end{figure}

It is well known that a stable matching in \acrshort{smi} can be found in $O(m)$ time via the \emph{Gale-Shapley} algorithm \citep{GS62}, where $m$ is the total length of all agents preference lists. This algorithm requires either men or women to be the \emph{proposers} and those of the opposite gender are \emph{receivers}. However, this procedure naturally produces a \emph{proposer-optimal} stable matching where members of the proposer group will be assigned to their best possible partner in any stable matching. Unfortunately, this also ensures a \emph{receiver-pessimal} stable matching in which members of the receiver group will be assigned their worst assignees in any stable matching.

It is natural therefore to want to find some notion of optimality which provides a sense of equality between men and women in a stable matching. This problem has been researched widely and and a summary of the literature is now given.

\subsection{Optimality in SMI}
\label{sec_optimality}
Let $\text{rank}(m_i, w_j)$ be the rank of woman $w_j$ on man $m_i$'s list with an analogous definition for the rank of man on a woman's list. Let $M$ be a stable matching in an instance $I$ of \acrshort{smi}. The rank of man $m_i$, with respect to $M$, is given by $\text{rank}(m_i, M(m_i))$. Similarly, the rank of woman $w_j$, with respect to $M$, is given by $\text{rank}(w_j, M(w_j))$. Then, the \emph{man-cost} $c_U(M)$ is defined as the sum of ranks of the set of all men, that is, $c_U(M) = \sum_{m_i \in U}\text{rank}(m_i, M(m_i))$. Similarly, the \emph{woman-cost} $c_W(M)$ is given by $c_W(M) = \sum_{w_j \in W}\text{rank}(w_j, M(w_j))$. We may then define the \emph{cost} of $M$, which is given by $c(M) = c_U(M) + c_W(M)$. The \emph{man-degree} $d_U(M)$ is the maximum rank of a man with respect to $M$, that is, $d_U(M) = \max \{\text{rank}((m_i, M(m_i)): m_i \in U\}$. Analogously, the \emph{woman-degree} $d_W(M)$ is given by $d_W(M) = \max \{\text{rank}((w_j, M(w_j)): w_j \in W\}$. Finally, the \emph{degree} $d(M)$ of stable matching $M$ is given by $d(M)=\max\{d_U(M), d_W(M)\}$.

Let $I$ be an instance of \acrshort{smi} with set of stable matchings $\mathcal{M}_S$. We consider both cost- and degree-based optimality.

First, we look at cost-based optimality. An \emph{egalitarian stable matching} is a stable matching $M$ such that $c(M)$ is minimised taken over all stable matchings in $\mathcal{M}_S$. Let $w(M)$ define some arbitrary weight function of stable matching $M$. A matching $M$ is \emph{minimum (maximum) weight} if $w(M)$ is minimum (maximum) taken over all stable matchings in $\mathcal{M}_S$. Thus the minimum weight function $w(M)$ is a generalisation of the cost function $c(M)$. \citet{ILG87} showed that an egalitarian stable matching can be found in $O(m^2)$ time and a minimum weight stable matching in $O(m^2 \log n)$ time.  Recall that $K$ is the weight of an optimal stable matching. \citet{Fed92} improved on the methods above, giving an $O(m^{1.5})$ algorithm for finding a minimum weight stable matching when $K\leq m$. In the general case when $m < K$ and $K = O(m^{c})$ the time complexity for this algorithm rises to $O(nm \log m)$, where $c$ is a constant value. A stable matching $M$ is \emph{balanced} if $\max\{c_U(M),c_W(M)\}$ is minimised, taken over all stable matchings in $\mathcal{M}_S$. The problem of finding a balanced stable matching was shown to be $\NP$-hard \cite{Fed90}. A \emph{sex-equal} stable matching is a stable matching $M$ such that the difference $|c_U(M)-c_W(M)|$ is minimum, taken over all stable matchings in $\mathcal{M}_S$. \citet{Kat93} showed that the problem of finding a sex-equal stable matching is also NP-hard. 

Next we consider degree-based optimality. A \emph{minimum regret} stable matching $M$ is a stable matching such that $d(M)$ is minimised over all stable matchings in $\mathcal{M}_S$, and can be found in $O(m)$ time \cite{Gus87}. A stable matching $M$ is \emph{regret-equal} if $|d_U(M)-d_W(M)|$ is minimised over all stable matchings. A regret-equal stable matching may be found in $O(d_0nm)$ time where $d_0 = |d_U(M_0) - d_W(M_0)|$, and $M_0$ is the man-optimal stable matching \cite{CM_20}. Finally, a stable matching $M$ is \emph{min-regret sum} if $\max\{d_U(M), d_W(M)\}$ is minimised over all stable matchings in $\mathcal{M}_S$. It is possible to find a min-regret sum stable matching in $O(d_sm)$ time where $d_s = d_U(M_0) - d_U(M_z)$ and $M_z$ is the woman-optimal stable matching \cite{CM_20}.

A \emph{median} stable matching may be described in the following way. Let $l_i$ denote the multiset of all women who are assigned to man $m_i$ in the stable matchings in $\mathcal{M}_S$ (in general $l_i$ is a multiset as $m_i$ may have the same partner in more than one stable matching). Assume that $l_i$ is sorted according to $m_i$'s preference order (there may be repeated values) and let $l_{i,j}$ represent the $j$th element of this list. For each $j$ ($1\leq j\leq |\mathcal{M}_S|$), let $M_j$ denote the set of pairs obtained by assigning $m_i$ to $l_{i,j}$. \citeauthor{TS98} \citep{TS98} showed the surprising result that $M_j$ is a stable matching for every $j$ such that $1 \leq j \leq |\mathcal{M}_S|$. If $|\mathcal{M}_S|$ is odd then the unique \emph{median} stable matching is found when  $j = \left\lceil \frac{|\mathcal{M}_S|}{2} \right\rceil$. However, if $|\mathcal{M}_S|$ is even, then the \emph{set of median} stable matchings are the stable matchings such that each man (woman) does no better (worse) than his (her) partner when $j = \frac{|\mathcal{M}_S|}{2}$ and no worse (better) than his (her) partner when $j = \frac{|\mathcal{M}_S|}{2} + 1$. For the purposes of this paper, in particular the  experimentation section, we define the \emph{median} stable matching as the stable matching found when $j = \left\lceil \frac{|\mathcal{M}_S|}{2} \right\rceil$. Computing the set of median stable matchings is \#\P-hard \cite{Che08}.

Define a \emph{rank-maximal} stable matching $M$ in \acrshort{smi} to be a stable matching in which the largest number of agents gain their first choice, then subject to that, their second choice and so on. More formally we define a \emph{profile} as a finite vector of integers (positive or negative) and the \emph{profile} of a stable matching as follows. Given a stable matching $M$, let the profile of $M$ be given by the vector $p(M) = \langle  p_1, p_2, ..., p_n \rangle$ where
$p_k = |\{(m_i,w_j)\in M : \text{rank}(m_i,w_j) = k\}| + |\{(m_i,w_j)\in M : \text{rank}(w_j, m_i) = k\}|$ for some $k : (1 \leq k \leq n)$. Thus we define a stable matching $M$ in an instance $I$ of \acrshort{smi} to be \emph{rank-maximal} if $p(M)$ is lexicographically maximum, taken over all stable matchings in $I$. We define the \emph{reverse profile} $p_r(M)$ to be the vector $p_r(M)=\langle p_k, p_{k-1}, ..., p_1 \rangle$. A stable matching $M$ in an instance $I$ of \acrshort{smi} is \emph{generous} if  $p_r(M)$ is lexicographically minimum, taken over all stable matchings in $I$.

\subsection{Graphical structures}

\citet{ILG87} define a \emph{rotation} $\rho = \{(m_1, w_1), (m_2, w_2), ..., (m_k, w_k)\}$ as a list of man-woman pairs in a stable matching $M$, such that when their assignments are permuted (each man $m_i$ moving from $w_i$ to $w_{i+1}$, where $i$ is incremented modulo $k$), we obtain another stable matching. Applying this permutation is known as \emph{eliminating} a rotation. A rotation $\rho$ is \emph{exposed} in $M$ if all of the pairs in $\rho$ are in $M$. A list of rotations of instance $I_0$ is given in Figure \ref{sm_rm:ex_rotationsList}.

\begin{figure}[]
\centering
   \mbox{\parbox{6cm}{
   $\rho_0$: $(m_1, w_5)$ $(m_3, w_8)$\\
$\rho_1$: $(m_1, w_8)$ $(m_2, w_3)$ $(m_4, w_6)$\\
$\rho_2$: $(m_3, w_5)$ $(m_6, w_1)$\\
$\rho_3$: $(m_7, w_2)$ $(m_5, w_7)$\\
$\rho_4$: $(m_3, w_1)$ $(m_5, w_2)$\\
    }}
    \vspace{-0.4cm}
  \caption{Rotations for instance $I_0$.}
  \label{sm_rm:ex_rotationsList}
    \end{figure}

In order to describe \emph{profiles} of rotations we must first describe arithmetic over profiles. {\color{black}Addition over profiles may be defined in the following way. Let $p=\langle p_1, p_2, ..., p_n \rangle$ and $p' = \langle p'_1, p'_2, ..., p'_n \rangle$ be profiles of length $n$. Then the addition of $p'$ to $p$ is taken pointwise over elements from $1 ... n$. That is, $p + p'=\langle p_1 + p'_1, p_2 + p'_2, ..., p_n + p'_n \rangle$. We define $p = p'$ if $p_i = p'_i$ for $1 \leq i \leq n$. Now suppose $p \neq p'$. Let $k$ be the first point at which these profiles differ, that is, suppose $p_k \neq p'_k$ and $p_i = p'_i$ for $1 \leq i < k$. Then we define $p \prec p'$ if $p_k < p'_k$. We say $p \preceq p'$ if either $p_k < p'_k$ or $p = p'$. Finally, we define a profile $p$ as \emph{maximum} (\emph{minimum}) among a set of profiles $P$ if for any other profile $p' \in P$, $p \succeq p'$ ($p \preceq p'$). It is trivial to show that an addition or comparison of two profiles would take $O(n)$ time in the worst case (since the length of any profile is bounded by $n$). Let $p'' = \langle p_1, p_2, ..., p_i, 0, ..., 0 \rangle$ be a profile, where $i \leq n$. Then for ease of description we may shorten this profile to $p'' = \langle p_1, p_2, ..., p_i \rangle$.  
}

{\color{black}
Suppose we have a rotation $\rho$ that, when eliminated, takes us from stable matching $M$ to stable matching $M'$, where $M$ and $M'$ have profiles $p(M)=\langle p_1, p_2, ..., p_n \rangle$ and $p(M')=\langle p_1', p_2', ..., p_n' \rangle$ respectively. Then the \emph{profile} of $\rho$ is defined as the net change in profile between $M$ and $M'$, that is, $p(\rho)=\langle p_1' - p_1, p_2' - p_2, ..., p_n' - p_n \rangle$. Hence, $p(M')=p(M) + p(\rho)$. It is easy to see that a particular rotation will give the same net change in profile regardless of which stable matching it is eliminated from. For a set of rotations $R=\{\rho_1, \rho_2, ..., \rho_r\}$, we define the \emph{profile} of $R$ as $p(R)=p(\rho_1) + p(\rho_2) + ... + p(\rho_r)$.}

A \emph{rotation poset} may be constructed as a directed graph which indicates the order in which rotations may be eliminated. Informally, if one rotation $\rho$ precedes another, $\tau$, in the rotation poset then $\tau$ is not exposed until $\rho$ has been eliminated. A \emph{closed subset} of the rotation poset may be defined as a set of rotations $P = \{\rho_1, \rho_2,..., \rho_r\}$ such that for every $\rho_i$ in $P$, all of $\rho_i$'s predecessors are also in $P$. It has been shown that there is a $1$-$1$ correspondence between the closed subsets of the rotation poset and the set of all stable matchings \citep[Theorem 3.1]{ILG87}.  The rotation poset for $I_0$, denoted $R_p(I_0)$, is shown in Figure \ref{sm_rm:rotationPoset}.

\begin{figure}[]
\centering
   \begin{subfigure}[t]{0.4\textwidth}
    \centering
    \begin{tikzpicture}
		\node [simpleNodes](r0) at (1*\graphxStd,4) {$\rho_0$};
		\node [simpleNodes](r1) at (0,2) {$\rho_1$};
		\node [simpleNodes](r2) at (2*\graphxStd,2.75) {$\rho_2$};
		\node [simpleNodes](r3) at (2*\graphxStd,1.25) {$\rho_3$};
		\node [simpleNodes](r4) at (1*\graphxStd,0) {$\rho_4$};

		\draw[->, >=latex, standardLine] (r0.225) -- (r1.45);
		\draw[->, >=latex, standardLine] (r0.315) -- (r2.135);
		\draw[->, >=latex, standardLine] (r1.315) -- (r4.135);
		\draw[->, >=latex, standardLine] (r2.south) -- (r3.north);
		\draw[->, >=latex, standardLine] (r3.225) -- (r4.45);
		\end{tikzpicture}
  \caption{Rotation poset $R_p(I_0)$.}
  \label{sm_rm:rotationPoset}
  \end{subfigure}
   \begin{subfigure}[t]{0.4\textwidth}
   \centering
    \begin{tikzpicture}
    
		\node [simpleNodes](r0) at (1*\graphxStd,4) {$\rho_0$};
		\node [simpleNodes](r1) at (0,2) {$\rho_1$};
		\node [simpleNodes](r2) at (2*\graphxStd,2.75) {$\rho_2$};
		\node [simpleNodes](r3) at (2*\graphxStd,1.25) {$\rho_3$};
		\node [simpleNodes](r4) at (1*\graphxStd,0) {$\rho_4$};

		\draw[->, >=latex, standardLine] (r0.225) -- (r1.45);
		\draw[->, >=latex, standardLine] (r0.315) -- (r2.135);
		\draw[->, >=latex, standardLine] (r1.315) -- (r4.135);
		\draw[->, >=latex, standardLine] (r2.south) -- (r3.north);
		\draw[->, >=latex, standardLine] (r3.225) -- (r4.45);
		\draw[->, >=latex, standardLine] (r2.225) -- (r4.north);
		
		\node at (3.3,2) {$2$};
		\node at (2,1.8) {$1$};
		\node at (2.5,0.5) {$1$};
		\node at (2.4,3.6) {$1$};
		\node at (0.5,0.9) {$2$};
		\node at (0.4,3.3) {$1,2$};
		
		\end{tikzpicture}
  \caption{Rotation digraph $R_d(I_0)$.}
  \label{sm_rm:rotationDigraph}
  \end{subfigure}
    \caption{Rotation poset and digraph of $I_0$.}
  \label{sm_rm:rotationPosetAndDigraph}
    \end{figure}
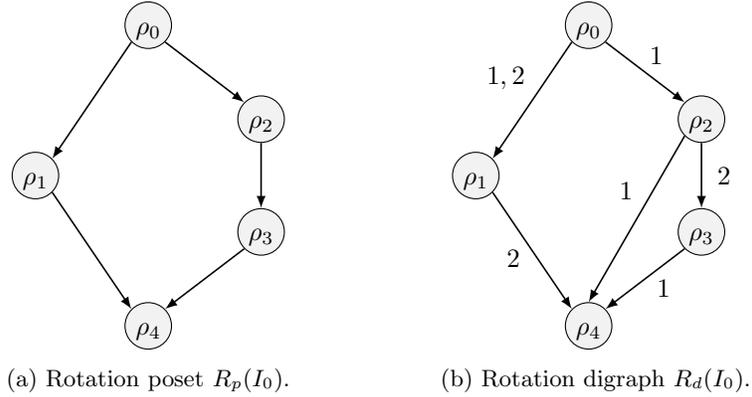

A description of the creation of a \emph{rotation digraph} now follows. First, retain each rotation from the rotation poset as a node. There are two types of predecessor relationships to consider.
\begin{enumerate}
	\item Suppose pair $(m_i,w_j) \in \rho$. We have a directed edge in our digraph from $\rho'$ to $\rho$ if $\rho'$ is the unique rotation that moves $m_i$ to $w_j$. In this case we say that $\rho'$ is a \emph{type $1$ predecessor} of $\rho$.
	\item Let $\rho$ be the rotation that moves $m_i$ below $w_j$ and $\rho' \neq \rho$ be the rotation that moves $w_j$ above $m_i$. Then we add a directed edge from $\rho'$ to $\rho$ and say $\rho'$ is a \emph{type $2$ predecessor} of $\rho$. 
\end{enumerate}
  The rotation digraph for instance $I_0$, denoted $R_d(I_0)$, is shown in Figure \ref{sm_rm:rotationDigraph}.

Using the rotation digraph structure, \citeauthor{GI89} \citep{GI89} were able to enumerate all stable matchings in $O(m + n|\mathcal{M}|)$, where $\mathcal{M}$ is the set of all stable matchings. All stable matchings of instance $I_0$ are listed in Figure \ref{sm_rm:ex_stableMatchingsList}.

\begin{figure}[]
\centering
$M_0 = \{(m_1,w_5), (m_2,w_3), (m_3,w_8), (m_4,w_6), (m_5,w_7), (m_6,w_1), (m_7,w_2), (m_8,w_4)\}$\\
$M_1 = \{(m_1,w_8), (m_2,w_3), (m_3,w_5), (m_4,w_6), (m_5,w_7), (m_6,w_1), (m_7,w_2), (m_8,w_4)\}$\\
$M_2 = \{(m_1,w_3), (m_2,w_6), (m_3,w_5), (m_4,w_8), (m_5,w_7), (m_6,w_1), (m_7,w_2), (m_8,w_4)\}$\\
$M_3 = \{(m_1,w_8), (m_2,w_3), (m_3,w_1), (m_4,w_6), (m_5,w_7), (m_6,w_5), (m_7,w_2), (m_8,w_4)\}$\\
$M_4 = \{(m_1,w_3), (m_2,w_6), (m_3,w_1), (m_4,w_8), (m_5,w_7), (m_6,w_5), (m_7,w_2), (m_8,w_4)\}$\\
$M_5 = \{(m_1,w_8), (m_2,w_3), (m_3,w_1), (m_4,w_6), (m_5,w_2), (m_6,w_5), (m_7,w_7), (m_8,w_4)\}$\\
$M_6 = \{(m_1,w_3), (m_2,w_6), (m_3,w_1), (m_4,w_8), (m_5,w_2), (m_6,w_5), (m_7,w_7), (m_8,w_4)\}$\\
$M_7 = \{(m_1,w_3), (m_2,w_6), (m_3,w_2), (m_4,w_8), (m_5,w_1), (m_6,w_5), (m_7,w_7), (m_8,w_4)\}$\\
  \caption{Stable matchings for instance $I_0$.}
  \label{sm_rm:ex_stableMatchingsList}
    \end{figure}
    
\subsection{Formal definition of SR}
\label{sec_formal_sr}

The \emph{\acrlong{sr}} (\acrshort{sr}) is a non-bipartite generalisation of \acrshort{sm}. An instance of \acrshort{sr} consists of a single set of $n$ agents (roommates), $A=\{a_1, a_2, ..., a_n\}$, each of whom ranks other members of the set in strict order of preference. A matching in this context is an assignment of pairs of agents such that each agent is assigned exactly once. Let $m$ be the total length of all preference lists. A \emph{matching} $M$ in \acrshort{sr} is then an assignment of acceptable pairs of agents such that each agent is assigned at most once. If $a_i$ is assigned in a matching $M$, we let $M(a_i)$ denote $a_i$'s assigned partner.  

The notion of stability also exists in this setting. As in the \acrshort{smi} case, a matching $M$ is \emph{stable} if there is no pair of agents who would rather be assigned to each other than to their assigned partners in $M$ (if any).
Using a counterexample, \citet{GS62} showed that a stable matching in \acrshort{sr} need not exist in all instances. \citet{Irv85} gave an $O(n^2)$ algorithm to find a stable matching in \acrshort{sr} or report that no such matching exists.

Let $M$ be a stable matching in \acrshort{sr}. For any two agents $a_i$ and $a_j$, we denote by $\text{rank}(a_i,a_j)$ the position of $a_j$ in $a_i$'s preference list and define the \emph{rank} of $a_i$ with respect to matching $M$ as $\text{rank}(a_i,M(a_i))$. The \emph{degree} of $M$ is the largest rank of all agents in $M$. A stable matching $M$ is \emph{minimum regret} if the degree of $M$ is minimum over all stable matchings. \citet{GI89} describe an algorithm to find a minimum regret stable matching, in a solvable instance $I$ of \acrshort{sr}, in $O(n^2)$ time. As in the \acrshort{smi} case, the \emph{profile} of $M$ is given by the vector $p(M) = \langle  p_1, p_2, ..., p_n \rangle$ where
$p_k = |\{a_i\in A : \text{rank}(a_i,M(a_i))=k\}|$ for each $k$ $(1 \leq k \leq n)$. A stable matching $M$ is then \emph{rank-maximal} if $p(M)$ is lexicographically maximum, taken over all stable matchings. Finally, a stable matching $M$ is \emph{generous} if the \emph{reverse profile} $p_r(M) = \langle  p_n, p_{n-1}, ..., p_1 \rangle$ is lexicographically minimum, taken over all stable matchings.

\section{Finding a rank-maximal stable matching using exponential weights} 

\label{sec_ran_max_exp_weights}

In this section we will describe how \citeauthor{ILG87}'s \citep{ILG87} maximum weight stable matching algorithm works and how it can be used to find a rank-maximal stable matching using exponential weights.

{\color{red}

}

\subsection{Exponential weight network} 
\label{smi_rm_exp_weight_network_sec}

\citeauthor{ILG87}'s \citep{ILG87} method for finding a maximum weight stable matching involves finding a maximum weight closed subset of the rotation poset.
In order to find a maximum weight closed subset of the rotation poset, a network is built and a maximum flow is found over this network.

The rotation digraph is converted to a network $R_n(I)$ as follows. First we add two extra vertices; a source vertex $s$ and a sink vertex $t$. An edge of capacity $\infty$ replaces each original edge in the digraph. 
Since we are finding a rank-maximal stable matching, capacities on other edges of $R_n(I)$ are calculated by converting each profile of a rotation to a single exponential weight. We decide on a weight function of $(2n+1)^{n-i}$ for each person assigned to their $i$th choice. From this point onwards we refer to the use of this weight function as the \emph{high-weight} scenario, and denote it as $w$.

\begin{definition}{}
\label{def_sm_rm_highweightfunction}
Given a profile $p = \langle p_1, p_2, ..., p_a \rangle$ such that $|p_1| + |p_2| + ... + |p_a| \leq 2n$ and $1 \leq a \leq n$, define the high-weight function $w$ as,

$$w(p) = p_1(2n+1)^{n-1} + p_2(2n+1)^{n-2} + ... + p_a(2n+1)^{n-a}.$$

\end{definition}



Lemma \ref{thm_sm_rm_highweightrankmax} shows that when the above function $w$ is used, a matching of maximum weight will be a rank-maximal matching.

\begin{prop}
\label{thm_sm_rm_iand1contributionToW}
Let $p = \langle p_1, p_2, ..., p_n \rangle$ and $p' = \langle p'_1, p'_2, ..., p'_n \rangle$ be profiles such that $|p_1| + |p_2| + ... + |p_n| \leq 2n$ and $|p'_1| + |p'_2| + ... + |p'_n| \leq 2n$. Let $w_i(p) = p_i(2n+1)^{n-i}$ denote the $i$th term of $w(p)$ and let $w_i^+(p)= \sum_{j=i}^{n}p_j(2n+1)^{n-j}$ denote the sum of $w(p')$ terms for all $j$ such that $i \leq j \leq n$. If $p_i>p'_i$, then $w_i(p) > w_i^+(p')$. Additionally, if $i$ is the first point at which $p$ and $p'$ differ, then $w(p) > w(p')$.
\end{prop}
\begin{proof}
	Assume $p_i > p'_i$. Then $p_i$ must be at least $1$ larger than $p_i'$ since each profile element is an integer by definition. A value of $1$ for $p_i$ will contribute $(2n+1)^{n-i}$ to $w_i(p)$ and so it follows that $w_i(p)\geq w_i(p') + (2n+1)^{n-i}$. 
	
	Since $(2n+1)^{n-k}$ decreases as $k$ increases and $|p'_1| + |p'_2| + ... + |p'_n| \leq 2n$, the maximum weight contribution that $p'_{i+1}, p'_{i+2}, ..., p'_n$ can make to $w_i^+(p')$ is when $p'_{i+1} =2n$.

Through the following series of inequalities,

\begin{equation} \label{eq:sm_onelargerhalflast}
\begin{split}
w_i^+(p') & \leq w_i(p') + 2n(2n+1)^{n-(i+1)}\\
& \leq w_i(p) - (2n+1)^{n-i} + \frac{2n}{2n+1}(2n+1)^{n-i} \\
& \leq w_i(p) + \left(\frac{2n}{2n+1} - 1\right)(2n+1)^{n-i} \\
 & < w_i(p) \\
\end{split}
\end{equation}

it follows that $w_i(p) > w_i^+(p')$ as required. If $i$ is the first point at which $p$ and $p'$ differ then it follows that $w(p) > w(p')$.
\end{proof}

\begin{lemma}
\label{thm_sm_rm_highweightrankmax}
Let $I$ be an instance of \acrshort{smi} and let $M$ be a stable matching in $I$. If $w(p(M))$ is maximum amongst all stable matchings of $I$, where $p(M)$ is the profile of $M$, then $M$ is a rank-maximal stable matching.
\end{lemma}
\begin{proof}
Suppose $w(p(M))$ is maximum amongst all stable matchings of $I$. Now, assume for contradiction that $M$ is not rank-maximal. Then, there exists some stable matching $M'$ in $I$ such that $M'$ lexicographically larger than $M$. Let $i$ be the first point at which $p(M)=\langle p_1, p_2, ..., p_n \rangle$ and $p(M')=\langle p'_1, p'_2, ..., p'_n \rangle$ differ. Since $M'$ is lexicographically larger than $M$ we know that $p'_i > p_i$ and by Proposition \ref{thm_sm_rm_iand1contributionToW} it follows that $w(p(M')) > w(p(M))$.

But this contradicts the fact that $w(p(M))$ is maximum over all stable matchings of $I$. Therefore our assumption that $M$ is not rank-maximal is false, as required.
\end{proof}

We now continue describing \citeauthor{ILG87}'s technique for finding a maximum weight closed subset of the rotation poset. The rotations are divided into positive and negative vertices as follows. A rotation $\rho$ is positive if $w(p(\rho)) > 0$ and negative if $w(p(\rho)) < 0$. A directed edge is added from the source to each negative vertex and is given a capacity equal to $|w(p(\rho))|$. A directed edge is also added between each positive vertex and $t$ with capacity $w(p(\rho))$. The high-weight network of instance $I_0$ is denoted $R_n(I_0)$ and is shown in Figure \ref{sm_rm:networkForRanks}. In this figure, each edge $e$ has a pair of associated integers $e_1/e_2$ where $e_1$ is the flow over $e$ and $e_2$ is the capacity of $e$.

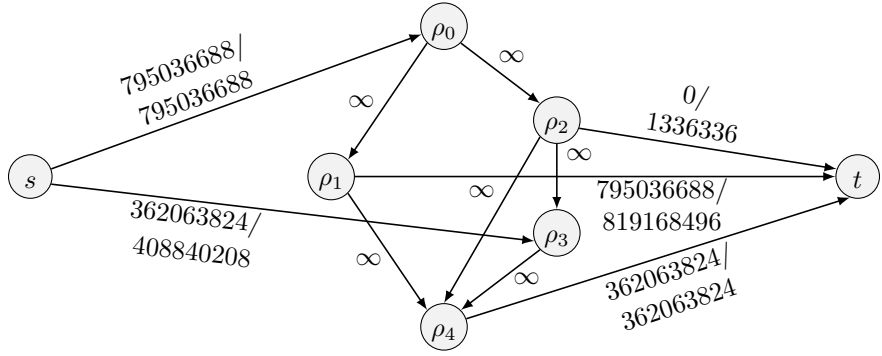
\begin{figure}[]
\centering
    \begin{tikzpicture}
		\node [simpleNodes](r0) at (1*\graphxStd,4) {$\rho_0$};
		\node [simpleNodes](r1) at (0,2) {$\rho_1$};
		\node [simpleNodes](r2) at (2*\graphxStd,2.75) {$\rho_2$};
		\node [simpleNodes](r3) at (2*\graphxStd,1.25) {$\rho_3$};
		\node [simpleNodes](r4) at (1*\graphxStd,0) {$\rho_4$};
		
		\node [simpleNodes](s) at (-4,2) {$s$};
		\node [simpleNodes](t) at (7,2) {$t$};

		\draw[->, >=latex, standardLine] (r0.225) -- (r1.45);
		\draw[->, >=latex, standardLine] (r0.315) -- (r2.135);
		\draw[->, >=latex, standardLine] (r1.315) -- (r4.135);
		\draw[->, >=latex, standardLine] (r2.south) -- (r3.north);
		\draw[->, >=latex, standardLine] (r3.225) -- (r4.45);
		\draw[->, >=latex, standardLine] (r2.225) -- (r4.north);
		
		\draw[->, >=latex, standardLine] (r1.east) -- (t.west);
		\draw[->, >=latex, standardLine] (s.20) -- (r0.200);
		\draw[->, >=latex, standardLine] (s.340) -- (r3.200);
		\draw[->, >=latex, standardLine] (r4.20) -- (t.250);
		\draw[->, >=latex, standardLine] (r2.340) -- (t.160);

		\node at (3.3,2.3) {$\infty$};
		\node at (2,1.8) {$\infty$};
		\node at (2.6,0.65) {$\infty$};
		\node at (2.4,3.6) {$\infty$};
		\node at (0.5,0.9) {$\infty$};
		\node at (0.4,3) {$\infty$};

		\node[rotate=-9] at (4.85,3.05) {$0/$};
		\node[rotate=-9] at (4.8,2.65) {$1336336$};

		\node at (4.4,1.8) {$795036688/$};
		\node at (4.4,1.4) {$819168496$};

		\node[rotate=18] at (4.5,0.7) {$362063824/$};
		\node[rotate=18] at (4.62,0.3) {$362063824$};

		\node[rotate=20] at (-1.95,3.45) {$795036688/$};
		\node[rotate=20] at (-1.8,3.05) {$795036688$};

		\node[rotate=-6.7] at (-1.8,1.45) {$362063824/$};
		\node[rotate=-6.7] at (-1.85,1) {$408840208$};

		\end{tikzpicture}
    \caption{The high-weight network $R_n(I_0)$.}
  \label{sm_rm:networkForRanks}
    \end{figure}

\subsection{Maximum weight closed subset of $R_p(I)$} 
\label{sec_min_cut_stable} 

We wish to find a minimum cut in $R_n(I)$ as this will allow us to calculate the maximum weight closed subset of $R_p(I)$. 
By the Max Flow-Min Cut Theorem \cite{FF56}, we need only find a maximum flow through $R_n(I)$ in order to find a minimum cut in $R_n(I)$. \citet{ILG87} used the Sleator-Tarjan algorithm \cite{ST83} to find a maximum flow. This algorithm completes when no augmenting path exists in $R_n(I)$ with respect to the final flow. Several analogous definitions used in the Sleator-Tarjan algorithm are required when we move to a combinatorial approach and so are described below.

In $R_n(I)$, we denote the flow over a node or edge as $f$ and an $s$-$t$ cut as $c_T$ with capacity given by $c(c_T)$. In order to search for augmenting paths we construct a new network known as the \emph{residual graph}. Given a network $R_n(I)$ and a flow $f$ in $R_n(I)$, the \emph{residual graph} relative to $R_n(I)$ and $f$, denoted $R_{res}(I,f)$, is defined as follows. The vertex set of $R_{res}(I,f)$ is equal to the vertex set of $R_n(I)$. An edge $(u,v)$, known as a forward edge, is added to $R_{res}(I, f)$ with capacity $c(u,v) - f(u,v)$ if $(u,v) \in E$ and $f(u,v) < c(u,v)$. Similarly an edge $(u,v)$, known as a backwards edge, is added to $R_{res}(I, f)$ with capacity $f(v,u)$ if $(v,u) \in E$ and $f(v,u) > 0$. Using a breadth-first search in $R_{res}(I,f)$ we may find an augmenting path or determine that none exists in $O(|E|)$ time. Once an augmenting path $P$ is found we \emph{augment} $R_n(I)$ in the following way:

\begin{itemize}
	\item The \emph{residual capacity} $c_a$ is the minimum of the capacities of the edges in $P$ in $R_{res}(I,f)$;
	\item For each edge $(u,v) \in P$, if $(u,v)$ is a forwards edge, the flow through $(u,v)$ is increased by $c_a$, whilst if $(u,v)$ is a backwards edge, the flow through $(v,u)$ is decreased by $c_a$.
\end{itemize}

\citet{FF56} showed that if no augmenting path in $R_n(I)$ can be found then the flow $f$ in $R_n(I)$ is maximum.

 In Figure \ref{sm_rm:networkForRanks} we show the high-weight network $R_n(I_0)$ with a maximum flow. There is one minimum cut, $c_T = \{(s,\rho_0),(\rho_4,t)\}$. Note that $c_T$ is a cut since removing these edges leaves no path from $s$ to $t$. By the Max Flow-Min Cut Theorem \cite{FF56}, it is also a minimum cut since the capacity of $c_T$ ($1157100512$) is equal to the value of a maximum flow.
 For this cut $c_T$ we list every rotation $\rho$ such that $(\rho,t)$ is an edge in $R_n(I)$ and $(\rho,t)\notin c_T$. Then a maximum weight closed subset of the rotation poset is given by this set of rotations and their predecessors. $c_T$ has associated closed subset of $\{\rho_0,\rho_1,\rho_2\}$ which is precisely a maximum weight closed subset of $R_p(I_0)$. The man-optimal stable matching of $I_0$ is $$M = \{(m_1,w_5), (m_2,w_3), (m_3,w_8), (m_4,w_6), (m_5,w_7), (m_6,w_1), (m_7,w_2), (m_8,w_4)\}.$$ By eliminating rotations $\{\rho_0,\rho_1,\rho_2\}$ from the man-optimal stable matching, we find the rank-maximal stable matching $$M' = \{(m_1,w_3), (m_2,w_6), (m_3,w_1), (m_4,w_8), (m_5,w_7), (m_6,w_5), (m_7,w_2), (m_8,w_4)\}.$$

The following Theorem summarises the work in this section.

\begin{theorem}[\citep{ILG87}]
\label{thm_sm_rm_highweightrankmax_overall}
Let $I$ be an instance of \acrshort{smi}. A rank-maximal stable matching $M$ of $I$ can be found in $O(nm^2 \log n)$ time using weights that are exponential in $n$.
\end{theorem}

An alternative to high-weight values when looking for a rank-maximal stable matching, is to use a new approach, involving polynomially-bounded weight vectors, to find a maximum weight closed subset of rotations. This is the focus of the rest of this paper.

\section[Finding a rank-maximal stable matching using polynomially-bounded weight vectors]{\texorpdfstring{Finding a rank-maximal stable matching using \\polynomially-bounded weight vectors}{Finding a rank-maximal stable matching using polynomially-bounded weight vectors}}
\label{sm_rm_combinatoricApproach}


\subsection{Strategy}
\label{sm_prof_strategy}

Following a similar strategy to \citet{ILG87}, we aim to show that we can return a rank-maximal stable matching in $O(nm^2 \log n)$ time without the use of exponential weights. The process we follow is described in the steps below.

\begin{enumerate}
	\item Calculate man-optimal and woman-optimal stable matchings using the Extended Gale-Shapley Algorithm -- $O(m)$ time;
	\item Find all rotations using the minimal differences algorithm -- $O(m)$ time;
	\item Build the rotation digraph and network -- $O(m)$ time; \label{SM_RM_steptochangeA}
	\item Find a minimum cut of the network in $O(nm^2 \log n)$ time without reverting to high weights; \label{SM_RM_steptochangeB}
	\item Use this cut to find a maximum profile closed subset $S$ of the rotation poset -- $O(m)$ time; 
	\item Eliminate the rotations of $S$ from the man-optimal matching to find the rank-maximal stable matching. \label{SM_RM_steptochangeMCSORP} 
\end{enumerate}

In the next section we discuss required adaptions to the high-weight procedure. 

\subsection{Vb-networks and vb-flows}
\label{sm_prof_comb_defns}

In this section we look at steps in the strategy to find a rank-maximal stable matching without the use of exponential weights (Section \ref{sm_prof_strategy}) which either require adaptations or further explanation.

In Step \ref{SM_RM_steptochangeMCSORP} of our strategy we eliminate the rotations of a maximum profile closed subset of the rotation poset from the man-optimal stable matching. We now present Lemma \ref{sm_rm_eliminatingCSORP}, an analogue of \cite[Corollary 3.6.1]{GI89}, which shows that eliminating a maximum profile closed subset of the rotation poset from the man-optimal stable matching results in a rank-maximal stable matching. 

\begin{lemma}
	\label{sm_rm_eliminatingCSORP}
	Let $I$ be an instance of \acrshort{smi} and let $M_0$ be the man-optimal stable matching in $I$. A rank-maximal stable matching $M$ may be obtained by eliminating a maximum profile closed subset $S$ of the rotation poset from $M_0$.
\end{lemma}
\begin{proof}
Let $R_p(I)$ be the rotation poset of $I$. By \citet[Theorem 2.5.7]{GI89}, there is a $1$-$1$ correspondence between closed subsets of $R_p(I)$ and the stable matchings of $I$.
Let $S$ be a maximum profile closed subset of the rotation poset $R_p(I)$ and let $M$ be the unique corresponding stable matching. Then, $p(M) = p(M_0) + \sum_{\rho_i \in S} p(\rho_i)$.
Suppose $M$ is not rank-maximal. Then there is a stable matching $M'$ such that $p(M') \succ p(M)$. As above, $M'$ corresponds to a unique closed subset $S'$ of the rotation poset. Also $p(M') = p(M_0) + \sum_{\rho_i \in S'} p(\rho_i)$. But $p(M') \succ p(M)$ and so $S$ cannot be a maximum profile closed subset of $R_p(I)$, a contradiction.
\end{proof}

Steps \ref{SM_RM_steptochangeA} and \ref{SM_RM_steptochangeB} of our strategy are the only places where we are required to check that it is possible to directly substitute an operation involving large weights taking $O(n)$ time with a comparable profile operation taking $O(n)$ time.

The first deviation from \citeauthor{GI89}'s method (described in Section \ref{smi_rm_exp_weight_network_sec}) is in the creation of a \emph{vector-based network} (abbreviated to \emph{vb-network}). For ease of description we denote this new vb-network as $R_n'(I)$ to distinguish it from the high-weight version $R_n(I)$. We now define a \emph{vb-capacity} in $R_n'(I)$ which is of similar notation to that of a profile.

\begin{definition}{}
In a vb-network $R_n'(I)$, the vector-based capacity (vb-capacity) of an edge $e$ is a vector $\boldsymbol{c}(e) = \langle c_1, c_2, ..., c_n \rangle$, where $n$ is the number of men or women in $I$ and $c_i \geq 0$ for $1 \leq i \leq n$.
\end{definition}

As before we add a source $s$ and sink $t$ vertex to the rotation digraph. We replace each original digraph edge with an edge with vb-capacity $\langle \infty, \infty, ..., \infty \rangle$ ($\infty$ repeated $n$ times). For convenience these edges are marked with `$\boldsymbol\infty$' in vb-network diagrams. The definition of a positive and negative rotation is also amended. Let $\rho$ have profile $p(\rho)=\langle p_1, p_2, ..., p_n\rangle $. Let $p_k$ be the first non-zero profile element where $1 \leq k \leq n$. We now define a positive rotation $\rho$ as a rotation where $p_k > 0$, and a negative rotation is one where $p_k < 0$. Define the absolute value operation, denoted $|p(\rho)|$, as follows. If $p_k > 0$, then leave all elements unchanged. If $p_k < 0$, then reverse the sign of all non-zero profile elements. Figure \ref{sm_rm:ex_rotationsProfsAndAbsProfs} shows the profile and absolute profile for each rotation of $I_0$. Then we add a directed edge to the vb-network from $s$ to each negative rotation vertex $\rho$ with a vb-capacity of $|p(\rho)|$ and a directed edge from each positive rotation vertex $\rho$ to $t$ with a vb-capacity of $p(\rho)$.

    \begin{figure}[]
\centering 
\begin{subfigure}[t]{0.35\textwidth}
$\rho_0$: $(m_1, w_5)$ $(m_3, w_8)$\\
$\rho_1$: $(m_1, w_8)$ $(m_2, w_3)$ $(m_4, w_6)$\\
$\rho_2$: $(m_3, w_5)$ $(m_6, w_1)$\\
$\rho_3$: $(m_7, w_2)$ $(m_5, w_7)$\\
$\rho_4$: $(m_3, w_1)$ $(m_5, w_2)$\\
  \caption{Rotations for instance $I_0$.}
  \label{sm_rm:ex_rotationsListA}
    \end{subfigure}
    
    \vspace{0.7cm}
    
   \begin{subfigure}[t]{0.45\textwidth}
   $p(\rho_0) = \langle -2, 1, 1, 1, 0, -1 \rangle$\\
$p(\rho_1) = \langle 2, 0, -1, -1, -1, -2, 1, 2\rangle$\\
$p(\rho_2) = \langle 0, 0, 1, -1\rangle$\\
$p(\rho_3) = \langle -1, 0, 1, 1, -1 \rangle$\\
$p(\rho_4) = \langle 1, -2, 0, 0, 0, 1\rangle$\\
  \caption{Rotation profiles.}
  \label{sm_rm:ex_rotationsprofiles}
    \end{subfigure}	
 \begin{subfigure}[t]{0.45\textwidth}
   $|p(\rho_0)| = \langle 2, -1, -1, -1, 0, 1 \rangle$\\
$|p(\rho_1)| = \langle 2, 0, -1, -1, -1, -2, 1, 2 \rangle$\\
$|p(\rho_2)| = \langle 0, 0, 1, -1 \rangle$\\
$|p(\rho_3)| = \langle 1, 0 -1, -1, 1 \rangle$\\
$|p(\rho_4)| = \langle 1, -2, 0, 0, 0, 1 \rangle$\\
  \caption{Absolute rotation profiles.}
  \label{sm_rm:ex_rotationsAbsoluteProfiles}
    \end{subfigure}	
	\caption{The profile and absolute profile for rotations of $I_0$.}
  \label{sm_rm:ex_rotationsProfsAndAbsProfs}
    \end{figure}

Let $e = (u, v)$ be an edge in $R_n'(I)$ with capacity $\boldsymbol{c}(e)$. We define a \emph{vb-flow} over $e$, denoted $\boldsymbol{f}(e)$, in the following way.

\begin{definition}{}
\label{smi_rm_defn_vb_flow}
In a vb-network $R_n'(I)$, the vector-based flow (vb-flow) over an edge $e$ is a vector $\boldsymbol{f}(e) = \langle f_1, f_2, ..., f_n \rangle$, where $\sum_{i=1}^{n}{|f_i|} \leq 2n$.
\end{definition}

 By Definition \ref{smi_rm_defn_vb_flow}, the sum of the absolute values of the elements of a flow over an edge cannot exceed $2n$, which implies that each flow element $f_i$ satisfies $-2n\leq f_i\leq 2n$. We will use this fact in the next section when proving the equivalence of the vector-based and high-weight approaches.
We also remark that it is possible for $f_i<0$ for some $i$ $(2 \leq i \leq n)$ and for $\boldsymbol{f}(e)$ to be a positive vb-flow. For example
$\langle 0, 0, 0\rangle \prec \langle 0, 1, -10\rangle$, and hence we are not bounding each individual element of a vb-flow over an edge by a minimum of zero.

We may now define a vb-flow over $R_n'(I)$ as follows.

\begin{definition}{}
\label{sm_rm_def_defofflow} In a vb-network $R_n'(I)=(V,E)$, a vector-based flow (vb-flow) is a function $f:E \rightarrow \mathbb{R}^n$ such that \footnote{\citeauthor{ST83}'s algorithm \cite{ST83} (used later in this section) may be used in a way that assumes integer value flows, and hence for the rest of this paper we assume vb-flow value elements will only ever be integers. That is, a vb-flow is a function $f:E \rightarrow \mathbb{N}^n$.},
\begin{enumerate}[i)]
	\item (vb-capacity) $\boldsymbol{f}(e) \succeq \boldsymbol{0}$ and $\boldsymbol{f}(e) \preceq \boldsymbol{c}(e)$ for all $e \in E$;
	\item (vb-conservation) $\sum\limits_{(u,v) \in E} \boldsymbol{f}(u,v) = \sum\limits_{(v,w) \in E} \boldsymbol{f}(v,w)$ for all $v \in V \backslash \{s,t\}$.
\end{enumerate}

Vb-flows are non-negative by the vb-capacity constraint.
	
\end{definition}

In addition we define the following notation and terminology for vb-flows. Let $\boldsymbol{f}$ and be a vb-flow in $R_n'(I) = (V,E)$. Define $\text{val}(\boldsymbol{f}) = \sum\{f(s,v): v \in V \wedge (s,v) \in E \}$. We define a maximum vb-flow $\boldsymbol{f}$ to be a vb-flow such that there is no other vb-flow $\boldsymbol{f'}$ where $\text{val}(\boldsymbol{f'}) > \text{val}(\boldsymbol{f})$.


The vb-network $R_n'(I_0)$ and the corresponding high-weight version $R_n(I_0)$ are shown in Figure \ref{sm_rm:networkForProfiles_HCandPB}. In order to translate \emph{vector-based values} (vb-values) to high-weight values we use the same formula as for profiles. That is, $(2n+1)^{n-i}$ for each man or woman assigned to their $i$th choice \cite{ILG87}. As an example, the vb-value $\langle 0, 0, 1, -1\rangle$ for rotation $\rho_2$ translates to a high-weight value of $0 * 17^7 + 0 * 17^6 + 1 * 17^5 - 1 * 17^4 = 1336336$.

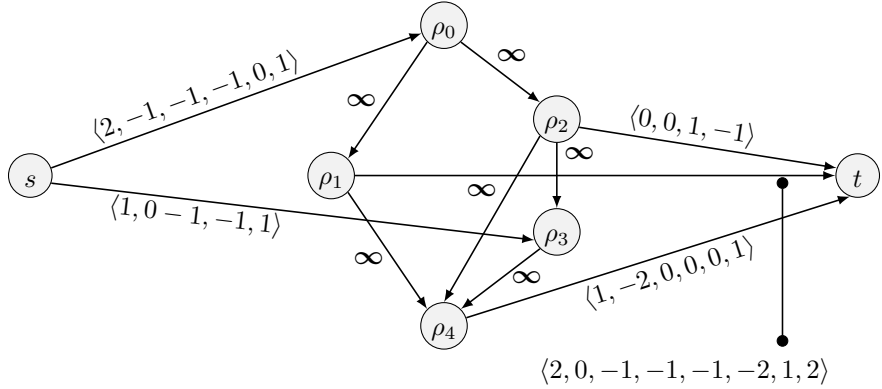
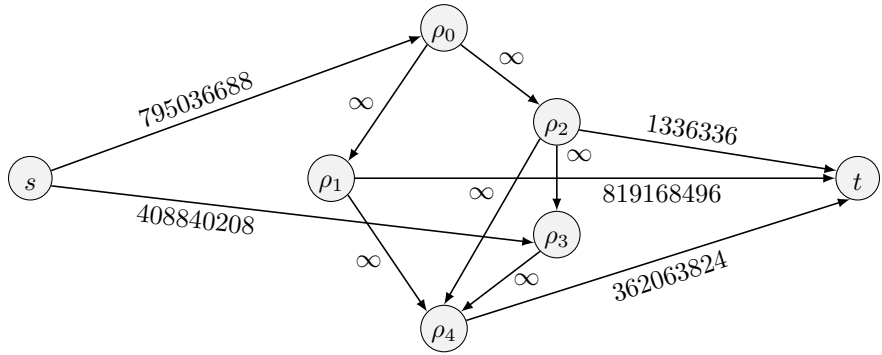
\begin{figure}[]
\centering
  
  \begin{subfigure}[t]{0.8\textwidth}
\centering
    \begin{tikzpicture}
		\node [simpleNodes](r0) at (1*\graphxStd,4) {$\rho_0$};
		\node [simpleNodes](r1) at (0,2) {$\rho_1$};
		\node [simpleNodes](r2) at (2*\graphxStd,2.75) {$\rho_2$};
		\node [simpleNodes](r3) at (2*\graphxStd,1.25) {$\rho_3$};
		\node [simpleNodes](r4) at (1*\graphxStd,0) {$\rho_4$};
		
		\node [simpleNodes](s) at (-4,2) {$s$};
		\node [simpleNodes](t) at (7,2) {$t$};

		\draw[->, >=latex, standardLine] (r0.225) -- (r1.45);
		\draw[->, >=latex, standardLine] (r0.315) -- (r2.135);
		\draw[->, >=latex, standardLine] (r1.315) -- (r4.135);
		\draw[->, >=latex, standardLine] (r2.south) -- (r3.north);
		\draw[->, >=latex, standardLine] (r3.225) -- (r4.45);
		\draw[->, >=latex, standardLine] (r2.225) -- (r4.north);
		
		\draw[->, >=latex, standardLine] (r1.east) -- (t.west);
		\draw[->, >=latex, standardLine] (s.20) -- (r0.200);
		\draw[->, >=latex, standardLine] (s.340) -- (r3.200);
		\draw[->, >=latex, standardLine] (r4.20) -- (t.250);
		\draw[->, >=latex, standardLine] (r2.340) -- (t.160);

		\node at (3.3,2.3) {$\boldsymbol\infty$};
		\node at (2,1.8) {$\boldsymbol\infty$};
		\node at (2.6,0.65) {$\boldsymbol\infty$};
		\node at (2.4,3.6) {$\boldsymbol\infty$};
		\node at (0.5,0.9) {$\boldsymbol\infty$};
		\node at (0.4,3) {$\boldsymbol\infty$};
		
		\node[rotate=-9] at (4.8,2.65) {$\langle 0, 0, 1, -1\rangle$};

		\node[rotate=18] at (4.5,0.7) {$\langle 1, -2, 0, 0, 0, 1\rangle$};
		\node[rotate=20] at (-1.8,3.05) {$\langle 2, -1, -1, -1, 0, 1\rangle$};
		\node[rotate=-6.7] at (-1.8,1.45) {$\langle 1, 0 -1, -1, 1\rangle$};
	
		
		\node[tinyDot](top) at (6,1.9) {};
		\node[tinyDot](bottom) at (6,-0.2) {};
		\draw[-, >=latex, standardLine] (top.south) -- (bottom.north);
	\node at (4.7,-0.6) {$\langle 2, 0, -1, -1, -1, -2, 1, 2\rangle$};
		
		\end{tikzpicture}
    \caption{Vector-based network $R_n'(I_0)$.}
  \label{sm_rm:networkForProfiles_profileBased}
  \end{subfigure}
  
  \vspace{0.7cm}
  
  \begin{subfigure}[t]{0.8\textwidth}
\centering
    \begin{tikzpicture}
		\node [simpleNodes](r0) at (1*\graphxStd,4) {$\rho_0$};
		\node [simpleNodes](r1) at (0,2) {$\rho_1$};
		\node [simpleNodes](r2) at (2*\graphxStd,2.75) {$\rho_2$};
		\node [simpleNodes](r3) at (2*\graphxStd,1.25) {$\rho_3$};
		\node [simpleNodes](r4) at (1*\graphxStd,0) {$\rho_4$};
		
		\node [simpleNodes](s) at (-4,2) {$s$};
		\node [simpleNodes](t) at (7,2) {$t$};

		\draw[->, >=latex, standardLine] (r0.225) -- (r1.45);
		\draw[->, >=latex, standardLine] (r0.315) -- (r2.135);
		\draw[->, >=latex, standardLine] (r1.315) -- (r4.135);
		\draw[->, >=latex, standardLine] (r2.south) -- (r3.north);
		\draw[->, >=latex, standardLine] (r3.225) -- (r4.45);
		\draw[->, >=latex, standardLine] (r2.225) -- (r4.north);
		
		\draw[->, >=latex, standardLine] (r1.east) -- (t.west);
		\draw[->, >=latex, standardLine] (s.20) -- (r0.200);
		\draw[->, >=latex, standardLine] (s.340) -- (r3.200);
		\draw[->, >=latex, standardLine] (r4.20) -- (t.250);
		\draw[->, >=latex, standardLine] (r2.340) -- (t.160);

		\node at (3.3,2.3) {$\infty$};
		\node at (2,1.8) {$\infty$};
		\node at (2.6,0.65) {$\infty$};
		\node at (2.4,3.6) {$\infty$};
		\node at (0.5,0.9) {$\infty$};
		\node at (0.4,3) {$\infty$};
		
		\node[rotate=-9] at (4.8,2.65) {$1336336$};
		\node at (4.4,1.8) {$819168496$};
		\node[rotate=18] at (4.5,0.7) {$362063824$};
		\node[rotate=20] at (-1.8,3.05) {$795036688$};
		\node[rotate=-6.7] at (-1.8,1.45) {$408840208$};
	
		
		\end{tikzpicture}
    \caption{High-weight network $R_n(I_0)$.}
  \label{sm_rm:networkForProfiles_highweight}
  \end{subfigure}

  \caption[Vector-based network $R_n'(I_0)$ and network $R_n(I_0)$.]{Vector-based network $R_n'(I_0)$ and network $R_n(I_0)$ with both vector-based and high-weight capacities respectively.}
  \label{sm_rm:networkForProfiles_HCandPB}
    \end{figure}

    An augmenting path in $R_n'(I)$ has an analogous definition to the standard definition of an augmenting path. The \emph{vector-based residual network} (vb-residual network) $R_{res}'(I, \boldsymbol{f})$ of $R_n'(I)$ with vb-capacities is created in the same way as the residual network $R_{res}(I,f)$ of $R_n(I)$. A cut in $R_n'(I)$, denoted $c_T'$, is defined in a similar way to a cut in $R_n(I)$ and has capacity $\boldsymbol{c}(c_T') = \sum{\boldsymbol{c}(e)}$ where $e$ is an edge in $c_T'$. Where the flow in $R_n(I)$ is equivalent to the vb-flow in $R_n'(I)$, we want to show that a maximum flow in $R_n(I)$ is also equivalent to a maximum vb-flow in $R_n'(I)$, and that the Max Flow-Min Cut Theorem holds for vb-networks.


\subsection{Rank-maximal stable matchings}

In this section, we show how we are able to use our vb-network to find a rank-maximal stable matching. First, we show that the Max Flow-Min Cut Theorem \cite{FF56} can be extended to a vb-network. Next, we prove that, analogous to the exponential weight case, a maximum profile closed subset of the rotation poset may be found by obtaining a minimum cut of the vb-network. Finally, we show that \citeauthor{ST83}'s Max Flow algorithm \citep{ST83} may be adapted to work with vb-networks, and that we are able to find a rank-maximal stable matching in $O(nm^2 \log n)$ time using polynomially-bounded weight vectors.

In Lemma \ref{lemma:sm_proveprofile_consistent_highweight} we show that vb-flows in $R_n'(I)$ correspond to high-weight flows in $R_n(I)$. Let $\boldsymbol{f}$ be a vb-flow in a vb-network, where $\text{val}(\boldsymbol{f}) = \langle f_1, f_2, ..., f_n \rangle$ and let $c_T'$ be a cut where $\boldsymbol{c}(c_T') = \langle c_{T_1}', c_{T_2}', ..., c_{T_n}' \rangle$.


\begin{prop}
	\label{thm_sm_rm_iand1contributionToW_flows}
Let $\boldsymbol{f}$ and $\boldsymbol{f'}$ be vb-flows. Let $w_i(\text{val}(\boldsymbol{f}))$ denote the $i$th term of $w(\text{val}(\boldsymbol{f}))$ and let $w_i^+(\text{val}(\boldsymbol{f'}))$ denote the sum of $w(\text{val}(\boldsymbol{f'}))$ terms for all $j$ such that $i \leq j \leq n$. If $f_i>f'_i$, then $w_i(\text{val}(\boldsymbol{f})) > w_i^+(\text{val}(\boldsymbol{f'}))$. Additionally, if $i$ is the first point at which $\boldsymbol{f}$ and $\boldsymbol{f'}$ differ, then $w(\text{val}(\boldsymbol{f})) > w(\text{val}(\boldsymbol{f'}))$.

Identical results hold for vb-capacities.
\end{prop}
\begin{proof}
Recall from the footnote of Definition \ref{sm_rm_def_defofflow} that vb-flow values must contain only integer elements.

The only difference between the structure of a profile $p$ and a vb-flow $\boldsymbol{f}$ is that each profile element must take a value between $0$ and $2n$ inclusive, whereas the lower bound of vb-flow elements is relaxed to $-2n$. This difference does not affect the validity of Proposition \ref{thm_sm_rm_iand1contributionToW} and so we may use identical reasoning to show that if $f_i>f'_i$, then $w_i(\text{val}(\boldsymbol{f})) > w_i^+(\text{val}(\boldsymbol{f'}))$, and if $i$ is the first point at which $\boldsymbol{f}$ and $\boldsymbol{f'}$ differ, then $w(\text{val}(\boldsymbol{f})) > w(\text{val}(\boldsymbol{f'}))$.

Since vb-capacities have an identical structure to vb-flows the all results also hold for the vb-capacity case.
 \end{proof}

\begin{lemma}
\label{lemma:sm_proveprofile_consistent_highweight}
	Let $\boldsymbol{f}$ and $\boldsymbol{f'}$ be vb-flows in $R_n'(I)$. Then $\text{val}(\boldsymbol{f}) \prec \text{val}(\boldsymbol{f'})$ if and only if $w(\text{val}(\boldsymbol{f})) < w(\text{val}(\boldsymbol{f'}))$. 
	
	Additionally, let $c_T'$ and $c_T''$ be cuts in $R_n'(I)$. Then $\boldsymbol{c}(c_T') \prec \boldsymbol{c}(c_T'')$ if and only if $w(\boldsymbol{c}(c_T')) < w(\boldsymbol{c}(c_T''))$.
\end{lemma}

\begin{proof}
Suppose that $\text{val}(\boldsymbol{f}) \prec \text{val}(\boldsymbol{f'})$. We know $\text{val}(\boldsymbol{f}) \neq \text{val}(\boldsymbol{f'})$, and at the first point $i$ at which $\text{val}(\boldsymbol{f})$ and $\text{val}(\boldsymbol{f'})$ differ $f_i<f'_i$. By Proposition \ref{thm_sm_rm_iand1contributionToW_flows}, $w(\text{val}(\boldsymbol{f})) < w(\text{val}(\boldsymbol{f'}))$ as required.

Now assume $w(\text{val}(\boldsymbol{f})) < w(\text{val}(\boldsymbol{f'}))$ and suppose for contradiction that $\text{val}(\boldsymbol{f}) \succeq \text{val}(\boldsymbol{f'})$. If $\text{val}(\boldsymbol{f}) = \text{val}(\boldsymbol{f'})$ then clearly $w(\text{val}(\boldsymbol{f})) = w(\text{val}(\boldsymbol{f'}))$ a contradiction. Therefore suppose $\text{val}(\boldsymbol{f}) \succ \text{val}(\boldsymbol{f'})$. Then, we can use identical arguments to the preceding paragraph to prove that $w(\text{val}(\boldsymbol{f})) > w(\text{val}(\boldsymbol{f'}))$. But this contradicts our original assumption that $w(\text{val}(\boldsymbol{f})) < w(\text{val}(\boldsymbol{f'}))$. Therefore, $\text{val}(\boldsymbol{f}) \prec \text{val}(\boldsymbol{f'})$.

Using identical reasoning to the vb-flow case for the  vb-capacity case, we can show that $\boldsymbol{c}(c_T') \prec \boldsymbol{c}(c_T'')$ if and only if $w(\boldsymbol{c}(c_T')) < w(\boldsymbol{c}(c_T''))$.
\end{proof}

%
%
%
%
%
%
%

Lemma \ref{thm:sm_rm_vb_noaug_maxflow} shows that if there is no augmenting path in a vb-network, then the vb-flow existing in this network is maximum.

\begin{lemma}
\label{thm:sm_rm_vb_noaug_maxflow}
	Let $I$ be an instance of \acrshort{smi} and let $R_n'(I)$ and $R_n(I)$ define the vb-network and network of $I$ respectively. For all vb-flows and vb-capacities we define a corresponding flow or capacity for $R_n(I)$ using the high-weight function $w$ (Definition \ref{def_sm_rm_highweightfunction}). Suppose $\boldsymbol{f}$ is a vb-flow in $R_n'(I)$ that admits no augmenting path. Then $\boldsymbol{f}$ is a maximum vb-flow in $R_n'(I)$.
\end{lemma}
\begin{proof}
	Let $f$ be the flow corresponding to $\boldsymbol{f}$ in $R_n(I)$. First, we show that $f$ is a maximum flow in $R_n(I)$. Suppose for contradiction that $f$ is not a maximum flow. Then there must exist an augmenting path relative to $f$ in $R_n(I)$. Let $E_P$ denote the edges involved in this augmenting path. Then, for each edge $(u,v) \in E_P$ of this augmenting path either,
	\begin{itemize}
		\item $f(u,v) < c(u, v)$, in which case the vb-flow $\boldsymbol{f}(u,v)$ through edge $(u,v) \in R_n'(I)$ may increase by $\langle 0, 0, ..., 1 \rangle$, or;
		\item $f(v,u) > 0$, and so the vb-flow $\boldsymbol{f}(v, u)$ through edge $(v,u) \in R_n'(I)$ may decrease by $\langle 0, 0, ..., 1 \rangle$.
	\end{itemize}
	
Therefore, there exists an augmenting path relative to $\boldsymbol{f}$ in $R_n'(I)$. But this contradicts the fact that $\boldsymbol{f}$ is a vb-flow in $R_n'(I)$ that admits no augmenting path. Hence our assumption that $f$ is not a maximum flow in $R_n(I)$ is false.

We now show that $\boldsymbol{f}$ is a maximum vb-flow in $R_n'(I)$. Suppose for contradiction that this is not the case. Then, there must exist a vb-flow $\boldsymbol{f'}$ such that $\text{val}(\boldsymbol{f'}) \succ \text{val}(\boldsymbol{f})$. By Lemma \ref{lemma:sm_proveprofile_consistent_highweight}, $w(\text{val}(\boldsymbol{f'})) > w(\text{val}(\boldsymbol{f}))$. Let $f'$ be the flow corresponding to $\boldsymbol{f'}$ in $R_n(I)$.
Then we have the following inequality: $$\text{val}(f') = w(\text{val}(\boldsymbol{f'})) > w(\text{val}(\boldsymbol{f})) = \text{val}(f)$$  
contradicting the fact that $f$ is a maximum flow in $R_n(I)$. Therefore $\boldsymbol{f}$ is a maximum vb-flow in $R_n'(I)$.
\end{proof}

This means if we use any Max Flow algorithm that terminates with no augmenting paths (such as the Ford-Fulkerson Algorithm \cite{FF56} adapted to work with vb-flows and vb-capacities) we have found a maximum flow in a vb-network.

We now show that the Max Flow-Min Cut Theorem can be extended to a vb-network.

\begin{theorem}
\label{thm:sm_rm_vb_mfmc}
Let $I$ be an instance of \acrshort{smi} and let $R_n'(I)=(V',E')$ and $R_n(I)$ define the vb-network and network of $I$ respectively. For all vb-flows and vb-capacities we define a corresponding flow or capacity for $R_n(I)$ using the high-weight function $w$ (Definition \ref{def_sm_rm_highweightfunction}).		 Let $\boldsymbol{f}$ be a maximum vb-flow through $R_n'(I)$ and $c_{T}'$ be a minimum cut of $R_n'(I)$. Then $\boldsymbol{c}(c_{T}') = \text{val}(\boldsymbol{f})$. 
\end{theorem}
\begin{proof}

Given $\boldsymbol{f}$ is a maximum vb-flow in $R_n'(I)$, we define a cut $c_T'$ in $R_n'(I)$ in the following way. A \emph{partial augmenting path} is an augmenting path from the source vertex $s$ to vertex $u \neq t$ in $V$ with respect to $\boldsymbol{f}$ in $R_n'(I)$. Note that, by Lemma \ref{thm:sm_rm_vb_noaug_maxflow} no augmenting path from $s$ to $t$ may exist at this point since $\boldsymbol{f}$ is a maximum vb-flow. Let $A$ be the set of reachable vertices along partial augmenting paths, and let $B = V \backslash A$. Then $s\in A$ and $t \in B$. Define $c_T' = \{(u,v) : u \in A, v \in B\}$. Then $c_T'$ is a cut in $R_n'(I)$. Since there is no partial augmenting path extending between vertices in $A$ and vertices in $B$, we know that for vertices $u \in A$ and $v \in B$,

\begin{itemize}
	\item if $(u,v) \in E'$ then $\boldsymbol{f}(u,v) = \boldsymbol{c}(u,v)$, and;
	\item if $(v,u) \in E'$ then $\boldsymbol{f}(v,u) = \boldsymbol{0}$.
\end{itemize}

Therefore $\text{val}(\boldsymbol{f}) = \boldsymbol{c}(c_T')$ (see, for example, \citet[p. 721, Lemma 26.4]{CLRS09} for a proof of this statement).

Let $f$ be the flow corresponding to $\boldsymbol{f}$ in $R_n(I)$. We now show that $c_T'$ is a minimum cut in $R_n'(I)$. Suppose for contradiction that this is not the case. Then there must exist a cut $c_T''$ in $R_n'(I)$ such that $\boldsymbol{c}(c_T'') \prec \boldsymbol{c}(c_T')$. By Lemma \ref{lemma:sm_proveprofile_consistent_highweight},
$w(\boldsymbol{c}(c_T'')) < w(\boldsymbol{c}(c_T'))$ and so we have the following inequality.

$$c(c_T'') = w(\boldsymbol{c}(c_T'')) < w(\boldsymbol{c}(c_T')) = w(\text{val}(\boldsymbol{f})) = \text{val}(f)$$

But then $c_T''$ is a cut with smaller capacity than $\text{val}(f)$ in $R_n(I)$ contradicting the Max Flow-Min Cut Theorem in $R_n(I)$. Hence $c_T'$ is a minimum cut in $R_n'(I)$.
\end{proof}

As an example, Figure \ref{sm_rm:networkForProfiles_profileBased_Flow} shows a maximum flow over the vb-network $R_n'(I_0)$, and Figure \ref{sm_rm:networkForProfiles_highweight_Flow} shows these vb-flows translated into the high-weight network $R_n(I_0)$. Similar to the high weight case, in $R_n'(I_0)$, each edge $e$ has a pair of associated integers $e_1/e_2$ where $e_1$ is the vb-flow over $e$ and $e_2$ is the vb-capacity of $e$. In Figure \ref{sm_rm:networkForProfiles_profileBased_Flow} each edge flow is positive, that is, the first non-zero element of each vb-flow is positive as required. The maximum vb-flow shown in this figure has saturated both edge $(s, \rho_0)$ leaving the source $s$ and edge $(\rho_4, t)$ entering the sink $t$. Let $c_T' = \{(s, \rho_0), (\rho_4, t)\}$. Clearly $c_T'$ constitutes a cut as removing these edges leaves no path from $s$ to $t$. Since the vb-capacity of $c_T'$ ($\langle 3, -3, -3, 1, -1, 0, 2\rangle$) is equal to the value of a maximum vb-flow of $R_n'(I_0)$, $c_T'$ is also a minimum cut, by Theorem \ref{thm:sm_rm_vb_mfmc}. The equivalent situation in the high-weight network is shown in Figure \ref{sm_rm:networkForProfiles_highweight_Flow}.

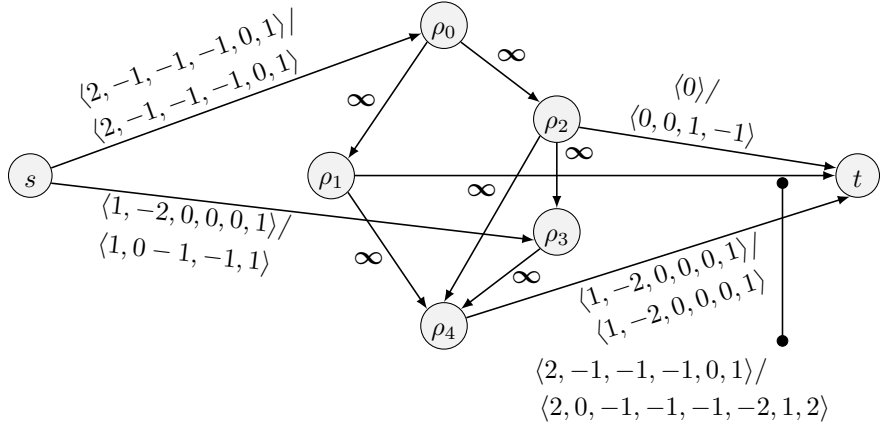
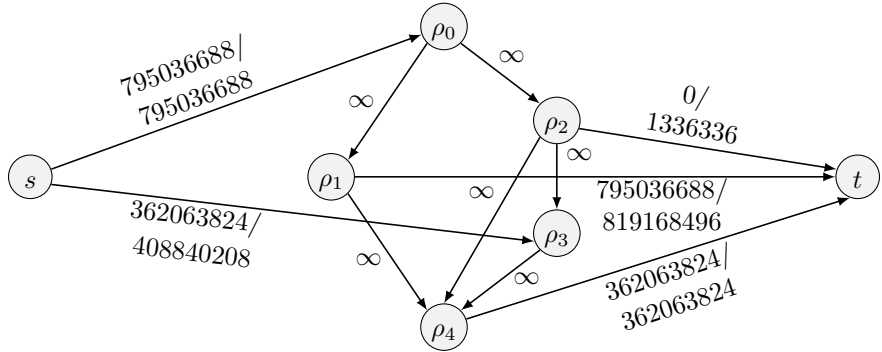
\begin{figure}[t]
\centering
  \begin{subfigure}[t]{0.8\textwidth}
\centering
    \begin{tikzpicture}
		\node [simpleNodes](r0) at (1*\graphxStd,4) {$\rho_0$};
		\node [simpleNodes](r1) at (0,2) {$\rho_1$};
		\node [simpleNodes](r2) at (2*\graphxStd,2.75) {$\rho_2$};
		\node [simpleNodes](r3) at (2*\graphxStd,1.25) {$\rho_3$};
		\node [simpleNodes](r4) at (1*\graphxStd,0) {$\rho_4$};
		
		\node [simpleNodes](s) at (-4,2) {$s$};
		\node [simpleNodes](t) at (7,2) {$t$};

		\draw[->, >=latex, standardLine] (r0.225) -- (r1.45);
		\draw[->, >=latex, standardLine] (r0.315) -- (r2.135);
		\draw[->, >=latex, standardLine] (r1.315) -- (r4.135);
		\draw[->, >=latex, standardLine] (r2.south) -- (r3.north);
		\draw[->, >=latex, standardLine] (r3.225) -- (r4.45);
		\draw[->, >=latex, standardLine] (r2.225) -- (r4.north);
		
		\draw[->, >=latex, standardLine] (r1.east) -- (t.west);
		\draw[->, >=latex, standardLine] (s.20) -- (r0.200);
		\draw[->, >=latex, standardLine] (s.340) -- (r3.200);
		\draw[->, >=latex, standardLine] (r4.20) -- (t.250);
		\draw[->, >=latex, standardLine] (r2.340) -- (t.160);

		\node at (3.3,2.3) {$\boldsymbol\infty$};
		\node at (2,1.8) {$\boldsymbol\infty$};
		\node at (2.6,0.65) {$\boldsymbol\infty$};
		\node at (2.4,3.6) {$\boldsymbol\infty$};
		\node at (0.5,0.9) {$\boldsymbol\infty$};
		\node at (0.4,3) {$\boldsymbol\infty$};
		
		\node[rotate=-9] at (4.85,3.15) {$\langle 0\rangle /$};
		\node[rotate=-9] at (4.8,2.65) {$\langle 0, 0, 1, -1\rangle$};

		\node[rotate=-6.7] at (-1.8,1.45) {$\langle 1, -2, 0, 0, 0, 1\rangle /$};
		\node[rotate=-6.7] at (-1.95,0.95) {$\langle 1, 0 -1, -1, 1\rangle$};
		
		\node[rotate=20] at (-1.9,3.55) {$\langle 2, -1, -1, -1, 0, 1\rangle /$};
		\node[rotate=20] at (-1.8,3.05) {$\langle 2, -1, -1, -1, 0, 1\rangle$};
		
		\node[rotate=18] at (4.5,0.7) {$\langle 1, -2, 0, 0, 0, 1\rangle /$};
		\node[rotate=18] at (4.65,0.25) {$\langle 1, -2, 0, 0, 0, 1\rangle$};

		
		\node[tinyDot](top) at (6,1.9) {};
		\node[tinyDot](bottom) at (6,-0.2) {};
		\draw[-, >=latex, standardLine] (top.south) -- (bottom.north);
	\node at (4.7,-1.1) {$\langle 2, 0, -1, -1, -1, -2, 1, 2\rangle$};
	\node at (4.25,-0.6) 
{$\langle 2, -1, -1, -1, 0, 1\rangle /$};
		
		\end{tikzpicture}
    \caption{Vector-based network $R_n'(I_0)$.}
  \label{sm_rm:networkForProfiles_profileBased_Flow}
  \end{subfigure}
  
    \vspace{0.7cm}
  
\begin{subfigure}[t]{0.8\textwidth}
\centering
    \begin{tikzpicture}
		\node [simpleNodes](r0) at (1*\graphxStd,4) {$\rho_0$};
		\node [simpleNodes](r1) at (0,2) {$\rho_1$};
		\node [simpleNodes](r2) at (2*\graphxStd,2.75) {$\rho_2$};
		\node [simpleNodes](r3) at (2*\graphxStd,1.25) {$\rho_3$};
		\node [simpleNodes](r4) at (1*\graphxStd,0) {$\rho_4$};
		
		\node [simpleNodes](s) at (-4,2) {$s$};
		\node [simpleNodes](t) at (7,2) {$t$};

		\draw[->, >=latex, standardLine] (r0.225) -- (r1.45);
		\draw[->, >=latex, standardLine] (r0.315) -- (r2.135);
		\draw[->, >=latex, standardLine] (r1.315) -- (r4.135);
		\draw[->, >=latex, standardLine] (r2.south) -- (r3.north);
		\draw[->, >=latex, standardLine] (r3.225) -- (r4.45);
		\draw[->, >=latex, standardLine] (r2.225) -- (r4.north);
		
		\draw[->, >=latex, standardLine] (r1.east) -- (t.west);
		\draw[->, >=latex, standardLine] (s.20) -- (r0.200);
		\draw[->, >=latex, standardLine] (s.340) -- (r3.200);
		\draw[->, >=latex, standardLine] (r4.20) -- (t.250);
		\draw[->, >=latex, standardLine] (r2.340) -- (t.160);

		\node at (3.3,2.3) {$\infty$};
		\node at (2,1.8) {$\infty$};
		\node at (2.6,0.65) {$\infty$};
		\node at (2.4,3.6) {$\infty$};
		\node at (0.5,0.9) {$\infty$};
		\node at (0.4,3) {$\infty$};

				\node[rotate=-9] at (4.85,3.05) {$0/$};
		\node[rotate=-9] at (4.8,2.65) {$1336336$};

		\node at (4.4,1.8) {$795036688/$};
		\node at (4.4,1.4) {$819168496$};

		\node[rotate=18] at (4.5,0.7) {$362063824/$};
		\node[rotate=18] at (4.62,0.3) {$362063824$};

		\node[rotate=20] at (-1.95,3.45) {$795036688/$};
		\node[rotate=20] at (-1.8,3.05) {$795036688$};

		\node[rotate=-6.7] at (-1.8,1.45) {$362063824/$};
		\node[rotate=-6.7] at (-1.85,1) {$408840208$};

		
		\end{tikzpicture}
    \caption{High-weight network $R_n(I_0)$.}
  \label{sm_rm:networkForProfiles_highweight_Flow}
  \end{subfigure}

  \caption[Maximum vb-flow and flow in the networks $R_n'(I_0)$ and $R_n(I_0)$.]{Maximum vb-flow and flow in the networks $R_n'(I_0)$ and $R_n(I_0)$ with both vector-based and high-weight capacities respectively.}
  \label{sm_rm:networkForProfiles_HC_PB_Flow}
    \end{figure}

In order to determine which rotations must be eliminated from the man-optimal stable matching, we must first determine a maximum profile closed subset of the rotation poset. As with the example in Section \ref{sec_ran_max_exp_weights}, we find the positive vertices which have edges into $t$ that are not in $c_T'$. These are $\rho_1$ and $\rho_2$. In Theorem \ref{th:sm_rm_proofRankMaximal} we will prove that a maximum profile closed subset of the vb-network comprises these vertices and their predecessors: $\{\rho_0, \rho_1, \rho_2\}$. It was shown in Section \ref{sec_ran_max_exp_weights} that this was indeed a maximum weight closed subset of $R_p(I_0)$.

The following theorem is a restatement of \citeauthor{GI89}'s theorem \citep[p. 130]{GI89} proving that a maximum profile closed subset of the rotation poset $R_p(I)$ can be found by finding a minimum $s$-$t$ cut of the vb-network $R_n'(I)$, when using a vector-based weight function.


%
%

\begin{theorem} 
	\label{th:sm_rm_proofRankMaximal}
	 Let $I$ be an instance of \acrshort{smi} and let $R_p(I)$, $R_d(I)$ and $R_n'(I)$ denote the rotation poset, rotation digraph and vb-network of $I$ respectively. Let $c_T'$ be a minimum $s$-$t$ cut in $R_n'(I)$, and let $\mathcal{P}_{c_T'}$ be the positive vertices of the network whose edges into $t$ are not in $c_T'$. Then the vertices $\mathcal{P}_{c_T'}$ and their predecessors define a maximum profile closed subset of $R_p(I)$. Further $\mathcal{P}_{c_T'}$ is exactly the set of positive vertices of this closed subset of the rotation poset.
	\end{theorem}
\begin{proof}

Let $S$ be an arbitrary set of rotations in $I$ and define $w(S) = \sum_{s_i \in S}{p(s_i)}$, that is, $w(S)$ is the total vector-based weight of these rotations.  Let $\mathcal{P}$ be the set of all positive rotations. For any set of rotations $S \subseteq \mathcal{P}$, let $N(S)$ be the set of all negative rotation predecessors of $S$ in the rotation digraph $R_d(I)$. Let $Q$ denote a maximum profile closed subset of $R_d(I)$. In order for a negative rotation vertex to exist in $Q$ it must precede at least one positive rotation vertex, otherwise $Q$ could not be of maximum weight. Hence $Q$ can be found by maximising $w(S) + w(N(S))$ over all subsets of $S \subseteq \mathcal{P}$.

We now show that $w(S) + w(N(S))$ is equivalent to $w(S) - |w(N(S))|$ according to our vector-based weight function. We know that $w(N(S))$ is negative (i.e. the first non-zero element of $w(N(S))$ is negative) and therefore taking the absolute value of $w(N(S))$ will reverse the signs of all non-zero elements. Taking the negative of $|w(N(S))|$ reverses the element signs once more and so we have $w(S) + w(N(S)) = w(S) - |w(N(S))|$.

Therefore we may say that $Q$ can be found by maximising $w(S) - |w(N(S))|$ over all subsets of $S \subseteq \mathcal{P}$. But by maximising this function, we also minimise $w(\mathcal{P}) - (w(S) - |w(N(S))|) = w(\mathcal{P} \backslash S) + |w(N(S))|$. That is, we are minimising the total weight of the positive rotations not in $S$ added to the absolute value of the negative rotations that are $S$'s predecessors. This becomes clearer when looking at the vb-network $R_n'(I)$. 


Let $\boldsymbol{c}({c_T'})$ denote the capacity of the minimum cut ${c_T'}$. We want to show that $\boldsymbol{c}({c_T'})$ is at least as small as $w(\mathcal{P} \backslash S) + |w(N(S))|$ for any $S \subseteq \mathcal{P}$. We can find an upper bound for $\boldsymbol{c}({c_T'})$ by doing the following. If we have a set of edges $c_T''$ that comprises \emph{(1)} all edges from $s$ to vertices in $N(S)$, and \emph{(2)} all edges from vertices in $\mathcal{P} \backslash S$ to $t$, then $c_T''$ is certainly a cut since there can be no flow through $R_n'(I)$. Moreover $\boldsymbol{c}(c_T'') = w(\mathcal{P} \backslash S) + |w(N(S))|$ and therefore, $\boldsymbol{c}({c_T'})\leq w(\mathcal{P} \backslash S) + |w(N(S))|$ for any $S \subseteq \mathcal{P}$. Now let $S^* \subseteq \mathcal{P}$ be the set of positive rotation vertices that have edges into $t$ that are not in ${c_T'}$. Then ${c_T'}$ must contain all edges from $\mathcal{P} \backslash S^*$ to $t$. Since ${c_T'}$ has finite capacity all the edges within it must also have finite capacity and consequently ${c_T'}$ must also contain all edges in $N(S^*)$ (since it cannot contain any edges with capacity $\boldsymbol\infty$). Therefore, $$\boldsymbol{c}({c_T'}) = w(\mathcal{P} \backslash S^*) + |w(N(S^*))| \leq w(\mathcal{P} \backslash S) + |w(N(S))|$$ 
for all $S \subseteq \mathcal{P}$. Hence, $\mathcal{P}_{c_T'} = S^*$ and, $\mathcal{P}_{c_T'}$ and their predecessors define a maximum profile closed subset of the rotation poset $R_p(I)$.
\end{proof}

It remains to show that we can adapt \citeauthor{ST83}'s $O(m^2 \log n)$ Max Flow algorithm to work with vb-networks. This is shown in Lemma \ref{sm_rm_sleator_tarjan_ok}.

\begin{lemma}
\label{sm_rm_sleator_tarjan_ok}
	Let $I$ be an instance of \acrshort{smi} and let $R_n'(I)$ be a vb-network. We can use a version of \citeauthor{ST83}'s Max Flow algorithm \citep{ST83} adapted to work with vb-networks in order to find a maximum flow $\boldsymbol{f}$ of  $R_n'(I)$ in $O(nm^2 \log n)$ time.
\end{lemma}
\begin{proof}
A \emph{blocking flow} in the high-weight setting is a flow such that each path through the network from $s$ to $t$ has a saturated edge. Note that this is different from a maximum flow, since a blocking flow may still allow extra flow to be pushed from $s$ to $t$ using backwards edges in the residual graph. The Sleator-Tarjan algorithm \citep{ST83} is an adapted version of \citeauthor{Din70}'s algorithm \citep{Din70} which improves the time complexity of finding a blocking flow. This is achieved by the introduction of a new \emph{dynamic tree structure}. 

The following operations are required from \citeauthor{ST83}'s \emph{dynamic tree structure} in the max flow setting \cite{ST83}: \emph{link}, \emph{capacity}, \emph{cut}, \emph{mincost}, \emph{parent}, \emph{update} and \emph{cost}. Each of these processes (not described here) consists of straightforward graph operations (such as adding a parent vertex, deleting an edge etc.) and comparisons, additions, subtractions and updating of edge capacities and flows. Since we have a vector-based interpretation of comparison, addition and subtraction operations, it is possible to adapt \citeauthor{ST83}'s Max Flow algorithm to work in the vector-based setting.

\citeauthor{ST83}'s algorithm \cite{ST83} terminates with a flow that admits no augmenting path. Let $\boldsymbol{f}$ be a vb-flow given at the termination of \citeauthor{ST83}'s algorithm, as applied to $R_n'(I)$. Since $\boldsymbol{f}$ admits no augmenting path, it follows, by Lemma \ref{thm:sm_rm_vb_noaug_maxflow} that $\boldsymbol{f}$ is a maximum vb-flow in $R_n'(I)$. \citeauthor{ST83}'s algorithm runs in $O(m^2 \log n)$ time assuming constant time operations for comparison, addition and subtraction. However, in the vb-flow setting, each of these operations takes $O(n)$ time in the worst case. Therefore, using the \citeauthor{ST83} algorithm, we have a total time complexity of $O(nm^2 \log n)$ to find a maximum flow of a vb-network.
\end{proof}

Finally, we now show that there is an $O(nm^2 \log n)$ algorithm for finding a rank-maximal stable matching in an instance of \acrshort{smi}, based on polynomially-bounded weight vectors.

\begin{theorem}
	Given an instance $I$ of \acrshort{smi} there is an $O(nm^2 \log n)$ algorithm to find a rank-maximal stable matching in $I$ that is based on polynomially-bounded weight vectors.
\end{theorem}
\begin{proof}
	We use the process described in Section \ref{sm_prof_strategy}. All operations from this are well under the required time complexity except number $4$. Let $R_n'(I) = (V',E')$ be a vb-network of $I$. Here, bounds on the number of edges and number of vertices are identical to the maximum weight case, that is $|E'| \leq m$ and $|V'| \leq m$. This is because, despite having alternative versions of flows and capacities, we have an identical graph structure to the high-weight case.
	
	By using the adaption of \citeauthor{ST83}'s Max Flow algorithm from Lemma \ref{sm_rm_sleator_tarjan_ok} we achieve an overall time complexity of $O(nm^2 \log n)$ to find a maximum vb-flow $\boldsymbol{f}$ in $R_n'(I)$. Let $c_T'$ denote a minimum cut in $R_n'(I)$. By Theorem \ref{thm:sm_rm_vb_mfmc}, $\boldsymbol{c}(c_{T}') = \text{val}(\boldsymbol{f})$. Therefore using the process described in Section \ref{sm_prof_strategy}, with vector-based adaptations, we can find a rank-maximal stable matching in $O(nm^2 \log n)$ time without reverting to high weights.
\end{proof}

Hence we have an $O(nm^2 \log n)$ algorithm for finding a rank-maximal stable matching, without reverting to high-weight operations.

\section{Generous stable matchings}
\label{sec_gen}

We now show how to adapt the techniques in Section \ref{sm_rm_combinatoricApproach} to the generous setting. Let $I$ be an instance of \acrshort{smi} and let $M$ be a matching in $I$ with profile $p(M)=\langle p_1, p_2, ..., p_k \rangle$. Recall the \emph{reverse profile} $p_r(M)$ is the vector $p_r(M)=\langle p_k, p_{k-1}, ..., p_1 \rangle$.

As with the rank-maximal case, we wish to use an approach to finding a generous stable matching that does not require exponential weights. Recall $n$ is the number of men in $I$. A simple $O(n)$ operation on a matching profile allows the rank-maximal approach described in the previous section to be used in the generous case. 

Let $M$ be a stable matching in $I$ with degree $k$ and profile
$p(M) = \langle p_1, p_2, ..., p_{k-1}, p_k \rangle$. Since we wish to minimise the reverse profile $p_r(M)=\langle p_k, p_{k-1}, ..., p_1 \rangle$ we can simply maximise the reverse profile where the value of each element is negated. A short proof of this is given in Proposition \ref{th:sm_gen_equivalentProfiles}. We denote this profile by $p_r'(M)$, where

\begin{equation}
	\label{sm_rm_generousProfile}
	p_r'(M) = \langle -p_k, -p_{k-1}, ..., -p_2, -p_1 \rangle.
\end{equation}

Thus in general, profile elements corresponding to $p_r'(M)$ can take negative values. All profile operations described in Section \ref{sec_ran_max_exp_weights} still apply to profiles of this type.

\begin{prop} 
	\label{th:sm_gen_equivalentProfiles}
	Let $M$ be a matching in an instance $I$ of \acrshort{smi}. Then, $$M \in \arg \min \{ p_r(M'): M' \text{ is a matching in } I\}$$ if and only if $$M \in \arg \max \{ p'_r(M'): M' \text{ is a matching in } I\}.$$
	\end{prop}
\begin{proof}
	Suppose $M$ is a matching in $I$ such that $p_r(M)$ is minimum taken over all matchings in $I$ and $p_r'(M)$ is not maximum taken over all matchings in $I$. Then, there is a matching $M'$ in $I$ such that $p_r'(M') \succ p_r'(M)$.
	
	Let $p_r(M)=\langle p_1, p_2, ..., p_k \rangle$ (note that we use indices from $1$ to $k$, despite $p_r(M)$ being a reverse profile) and therefore $p'_r(M)=\langle -p_1, -p_2, ..., -p_k \rangle$. Also let $p'_r(M') = \langle p_1', p_2', ..., p_l' \rangle$. Since $p_r'(M') \succ p_r'(M)$, there must exist some $i$ $(1 \leq i \leq \min\{k,l\})$ such that $p_i'> -p_i$ and $p_j'=-p_j$ for $1\leq j <i$.
	
	Then,
\begin{equation} 
\begin{split}
p_r(M') & = \langle -p_1', -p_2', ..., -p_l' \rangle \\
 & = \langle p_1, p_2, ..., p_{i-1}, -p'_i, ..., -p'_l \rangle \\
  & \prec \langle p_1, p_2, ...p_{i-1}, p_i, ..., p_k \rangle\\
  & = p_r(M). \\
\end{split}
\end{equation}
	Hence, $p_r(M)$ cannot be minimum taken over all matchings in $I$, a contradiction. Therefore, $M$ is a matching such that $p_r'(M)$ is maximum taken over all matchings in $I$.
	
	Now conversely, suppose that $M$ is a matching such that $p_r'(M)$ is maximum taken over all matchings in $I$, but $p_r(M)$ is not minimum taken over all matchings in $I$. Then there is a matching $M'$ in $I$ such that $p_r(M') \prec p_r(M)$. 
	
	Let $p_r(M)=\langle p_1, p_2, ..., p_k \rangle$ and therefore $p'_r(M)=\langle -p_1, -p_2, ..., -p_k \rangle$. Also let $p'_r(M') = \langle p_1', p_2', ..., p_l' \rangle$ and so $p_r(M') = \langle -p_1', -p_2', ..., -p_l' \rangle$. Since $p_r(M') \prec p_r(M)$, there must exist some $i$ $(1 \leq i \leq \min\{k,l\})$ such that $-p_i'<p_i$ and $-p_j'= p_j$ for $1\leq j <i$.
	
	Then,
	\begin{equation} 
\begin{split}
p'_r(M') & = \langle p_1', p_2', ..., p_l' \rangle \\
 & = \langle -p_1, -p_2, ..., -p_{i-1}, p'_i, ..., p'_l \rangle \\
  & \succ \langle -p_1, -p_2, ... -p_{i-1} , -p_i, ..., -p_k \rangle\\
  & = p'_r(M). \\
\end{split}
\end{equation}
	Hence, $p'_r(M)$ cannot be maximum taken over all matchings in $I$, a contradiction, meaning that $M$ is a matching such that $p_r(M)$ is minimum taken over all matchings in $I$.
\end{proof}

We now show that a generous stable matching may be found by eliminating a maximum profile closed subset of the rotation poset as in the rank-maximal case. Let $\rho$ be the rotation that takes us from stable matching $M$ to stable matching $M'$, where $M$ and $M'$ have profiles $p(M)=\langle p_1, p_2, ..., p_k \rangle$ and $p(M')=\langle p_1', p_2', ..., p_k' \rangle$ with each profile having length $k$ without loss of generality. Then $p(\rho)=\langle p_1' - p_1, p_2' - p_2, ..., p_k' - p_k \rangle$ and so $p(M')=p(M) + p(\rho)$. We also know that $p_r'(M)=\langle -p_k,-p_{k-1}, ..., -p_1 \rangle$ and $p_r'(M')=\langle -p_k',-p_{k-1}', ..., -p_1' \rangle$.

 Now, since $p_r'(\rho)=\langle p_k - p_k', p_{k-1} - p_{k-1}', ..., p_1 - p_1' \rangle$, in the generous case we have $p_r'(M')=p_r'(M) + p_r'(\rho)$. We next present Lemma \ref{sm_rm_eliminatingCSORPgen} which is an analogue of Lemma \ref{sm_rm_eliminatingCSORP} and shows that a generous stable matching may be found by eliminating a maximum profile closed subset of the rotation poset.

\begin{lemma}
	\label{sm_rm_eliminatingCSORPgen}
		Let $I$ be an instance of \acrshort{smi} and let $M_0$ be the man-optimal stable matching in $I$. A generous stable matching $M$ may be obtained by eliminating a maximum profile closed subset of the rotation poset $S$ from $M_0$.
\end{lemma}
\begin{proof}
Let $R_p(I)$ be the rotation poset of $I$. There is a $1$-$1$ correspondence between closed subsets of $R_p(I)$ and the stable matchings of $I$ \citep{GI89}.
Let $S$ be a maximum profile closed subset of the rotation poset $R_p(I)$, whose rotation profiles are reversed and negated, and let $M$ be the unique corresponding stable matching. Then by addition over reversed negated profiles described in the text above this lemma, $p_r'(M) = p_r'(M_0) + \sum_{\rho_i \in S} p_r'(\rho_i)$.
Suppose $M$ is not generous. Then there is a stable matching $M'$ such that $p_r'(M') \succ p_r'(M)$. $M'$ corresponds to a unique closed subset $S'$ of the rotation poset, such that $p_r'(M') = p_r'(M_0) + \sum_{\rho_i \in S'} p_r'(\rho_i)$. But $p_r'(M') \succ p_r'(M)$ and so $S$ cannot be a maximum profile closed subset of $R_p(I)$, a contradiction.

By Proposition \ref{th:sm_gen_equivalentProfiles}, since $M$ is a matching such that $p_r'(M)$ is maximum among all stable matchings, $M$ is also a matching such that $p_r(M)$ is minimum among all stable matchings. Therefore $M$ is a generous stable matching in $I$.\end{proof}

Recall that in Definition \ref{def_sm_rm_highweightfunction} we defined the high-weight function $w$ in order to show that a stable matching of maximum weight is a rank-maximal stable matching and then showed that vb-flows and vb-capacities correspond directly with this high-weight setting.  In the generous case, since we seek a matching $M$ that maximises $p_r'(M)$, the negation of the reverse profile, the constraints of Definition \ref{def_sm_rm_highweightfunction} still apply. By Proposition \ref{th:sm_gen_equivalentProfiles} and Lemma \ref{sm_rm_eliminatingCSORPgen}, all processes from the previous section to find a rank-maximal stable matching may now be used to find a generous stable matching in $O(nm^2 \log n)$ time. 

However, it is also possible to exploit the structure of a generous stable matching to bound some part of the overall time complexity by the generous stable matching degree rather than by the number of men or women $n$.

Let $I$ be an instance of \acrshort{smi}. First we find a minimum regret stable matching $M'$ of $I$ as described in Section \ref{sec_optimality} in $O(m)$ time. It must be the case that the degree $d(M)$ of a generous stable matching $M$ is the same as the degree of $M'$. Therefore, since no man or women can be assigned to a partner of rank higher than $d(M)$, it is possible to simply truncate all preference lists beyond rank $d(M)$, which has a positive effect on the overall time complexity of finding a generous stable matching. Theorem \ref{sm_rm_generous_tc} shows that a generous stable matching may be found in $O(\min\{m, nd\}^2 d \log n)$ time, where $d$ is the degree of a minimum regret stable matching.
 
\begin{theorem}
\label{sm_rm_generous_tc}
Given an instance $I$ of \acrshort{smi} there is an $O(\min\{m, nd\}^2 d \log n)$ algorithm to find a generous stable matching in $I$ using polynomially-bounded weight vectors, where $d$ is the degree of a minimum regret stable matching.
\end{theorem}
\begin{proof}
Each step required to find a generous stable matching in $I$ is outlined below along with its time complexity. 
	 \begin{enumerate}
	\item Calculate the degree $d$ of a minimum regret stable matching and truncate preference lists accordingly. A minimum regret stable matching may be found in $O(m)$ time \cite{Gus87}. We now assume all preference lists are truncated below rank $d$.
	\item Calculate the man-optimal and woman-optimal stable matchings. The Extended Gale-Shapley Algorithm takes $O(\min\{m, nd\})$ time since the number of acceptable pairs is now $\min\{m, nd\}$. \label{sm_rm_genlist_egs}
	\item Find all rotations using the Minimal Differences Algorithm \cite{GI89} in $O(\min\{m, nd\})$ time, since the number of acceptable pairs is now $\min\{m, nd\}$. \label{sm_rm_genlist_mindiff}
	\item Build the rotation digraph and vb-network, using the process described in Section \ref{sm_rm_combinatoricApproach}, but where rotation profiles are reversed and negated for the building of the vb-network. We know that no man-woman pair can appear in more than one rotation. This means that the number of vertices in each of the associated rotation poset, rotation digraph and vb-network is also $O(\min\{m, nd\})$. Identical reasoning (with adapted time complexities to suit the generous case) to that of \citet[p. 112]{GI89} may be used to obtain a bound of $O(\min\{m, nd\})$ on the number of edges. This is because we have a bound of $O(\min\{m, nd\})$ for the creation of both type $1$ and type $2$ edges of the rotation digraph. Therefore we may build the rotation digraph and vb-network in $O(\min\{m, nd\})$ time.
	\item Find a minimum cut of the vb-network using the process described in Section \ref{sm_rm_combinatoricApproach}. With $O(\min\{m, nd\})$ vertices and edges and a maximum length of $d$ for any preference list, the \citeauthor{ST83} algorithm \citep{ST83} has a time complexity of $O(\min\{m, nd\}^2 \log n)$, with an additional factor of $O(d)$ to perform operations over vectors. Hence this step takes a total of $O(\min\{m, nd\}^2 d \log n)$ time. \label{sm_gen_placeImproved}
	\item Use this cut to find a maximum profile closed subset $S$ of the rotations in $O(\min\{m, nd\})$ time, since the numbers of vertices and edges are bounded by $O(\min\{m, nd\})$.\label{sm_rm_genlist_mwcsrp}
	\item Eliminate the rotations of $S$ from the man-optimal matching to find the corresponding rank-maximal stable matching.
\end{enumerate}

 Therefore the operation that dominates the time complexity is still Step \ref{sm_gen_placeImproved} and the overall time complexity to find a generous stable matching for an instance $I$ of \acrshort{smi} is $O(\min\{m, nd\}^2 d \log n)$.
\end{proof}

  \section{Complexity of finding profile-based stable matchings in {\sc sr}}
\label{sec_rm_sr}

 \acrshort{sr}, a generalisation of \acrshort{sm}, was first introduced in Section \ref{sec_formal_sr}. In this section, we look at the complexity of finding rank-maximal and generous stable matchings in \acrshort{sr}. 

First let \acrshort{rmsr} be the problem of finding a rank-maximal stable matching in an instance $I$ of \acrshort{sr}, and let \acrshort{gensr} be the problem of finding a generous stable matching in an instance $I$ of \acrshort{sr}. We now define their respective decision problems.

\begin{definition}
	We define {\sc rmsr-d}, the decision problem of \acrshort{rmsr}, as follows. An instance $(I, \sigma)$ of {\sc rmsr-d} comprises an instance $I$ of \acrshort{sr} and a profile $\sigma$. The problem is to decide whether there exists a stable matching $M$ in $I$ such that $p(M) \succeq \sigma$.
\end{definition}

\begin{definition}
	We define {\sc gensr-d}, the decision problem of \acrshort{gensr}, as follows. An instance $(I, \sigma)$ of {\sc gensr-d} comprises an instance $I$ of \acrshort{sr} and a profile $\sigma$. The problem is to decide whether there exists a stable matching $M$ in $I$ such that $p_r(M) \preceq \sigma_r$, where $\sigma_r$ is the reverse profile of $\sigma$.
\end{definition}

We now show that both {\sc rmsr-d} and {\sc gensr-d} are $\NP$-complete.

\begin{theorem}
	\label{pb_opt:thm_rmsr_gensr}
    {\sc rmsr-d} is $\NP$-complete.
\end{theorem}

\begin{proof}

We first show that {\sc rmsr-d} is in $\NP$. Given a matching $M$ it is possible to check that $M$ is stable in $O(m)$ time by iterating through the preference lists of all men and women, comparing the rank of a given agent on a preference list with the rank of the assigned partner in $M$. It is also possible to calculate profile $p(M)$ and to test whether $p(M) \succeq \sigma$ in $O(n)$ time. Therefore {\sc rmsr-d} is in $\NP$.

	We reduce from {\sc vertex cover} in cubic graphs \cite{GJS76, MS77} (VC-3). An instance $(G,K)$ of VC-3 comprises a graph $G=(V,E)$, with vertex set $V=\{v_1, v_2, ..., v_n\}$ such that each vertex has degree $3$, and a positive integer $K$. The problem is to determine whether there exists a set of vertices $V' \subseteq V$ such that for every edge $e \in E$ at least one endpoint of $e$ is in $V'$, where $|V'| \leq K$.
	
	We now construct an instance $(I, \sigma)$ of {\sc rmsr-d} from $(G,K)$.
	
	Agents of the constructed instance $I$ of \acrshort{sr} comprise $A \cup B \cup V \cup W \cup X \cup Y$ where $A = \{a_1, a_2, a_3, a_4\}$, $B = \{b_1, b_2, b_3, b_4\}$, $V = \{v_1, v_2, ..., v_n\}$, $W = \{w_1, w_2, ..., w_n\}$, $X = \{x_1, x_2, ..., x_n\}$ and $Y = \{y_1, y_2, ..., y_n\}$.
	For each vertex $v_i \in V$ with neighbours $v_i^1, v_i^2, v_i^3$ in $G$, shown in Figure \ref{sm_rm:srcomplexity_a}, we create complete preference lists for instance $I$ of \acrshort{sr} as in Figure \ref{sm_rm:srcomplexity_b}, with members of the sets $A$, $B$, $V$, $W$, $X$ and $Y$ comprising new roommate agents of the constructed \acrshort{sr} instance. Note that relative order of $v_i^1, v_i^2$ and $v_i^3$ in $v_i$'s preference list is not important. Additional complete preference lists are created as shown in Figure \ref{sm_rm:srcomplexity_c}. In these figures, the symbol ``$...$" at the end of preference lists indicates that all other agents in $I$ who are not already specified in the preference list are then listed in any order. 
	
\begin{figure}
	\centering
   \begin{subfigure}[b]{0.4\textwidth}
    \centering
    \begin{tikzpicture}
 
 		\node [simpleNodes](vi) at (0, -2) {$v_i$};
 		\node [simpleNodes](vi1) at (-1.75, -1.25) {$v_i^1$};
 		\node [simpleNodes](vi2) at (1, -0.5) {$v_i^2$};
 		\node [simpleNodes](vi3) at (0.5, -3.75) {$v_i^3$};

		\draw[-, >=latex, standardLine] (vi) -- (vi1);
		\draw[-, >=latex, standardLine] (vi) -- (vi2);
		\draw[-, >=latex, standardLine] (vi) -- (vi3);
		\draw[-, >=latex, standardLineDashed] (vi1) -- (-1.75, -0.4);
		\draw[-, >=latex, standardLineDashed] (vi1) -- (-2.15, -2);
		\draw[-, >=latex, standardLineDashed] (vi2) -- (0.75, 0.3);
		\draw[-, >=latex, standardLineDashed] (vi2) -- (1.75, -1);
		\draw[-, >=latex, standardLineDashed] (vi3) -- (0, -4.45);
		\draw[-, >=latex, standardLineDashed] (vi3) -- (1.4, -4);
		\end{tikzpicture}
  \caption{The neighbourhood of $v_i$ in the graph $G$ of a VC-3 instance.}
  \label{sm_rm:srcomplexity_a}
  \end{subfigure}
	\begin{subfigure}[b]{0.4\textwidth}
	\centering
	\begin{tabular}{l}
For each $v_i \in V$: \\\\
$v_i$: $w_i$, $v_i^1$, $v_i^2$, $v_i^3$, $y_i$, $...$ \\
$w_i$: $x_i$, $a_1$, $a_2$, $a_3$, $a_4$, $v_i$, $...$\\
$x_i$: $a_1$, $y_i$, $w_i$, $...$\\
$y_i$: $v_i$, $x_i$, $...$\\\\
  \end{tabular}
\caption{Preference lists of $I$ for a given $v_i$.}
\label{sm_rm:srcomplexity_b}
\end{subfigure}
\vspace{1cm}

\begin{subfigure}[t]{0.6\textwidth}
	\centering
	\begin{tabular}{l}
$a_j$: $b_j$, $...$ \hspace{1cm} $(1 \leq j \leq 4)$\\
$b_j$: $a_j$, $...$ \hspace{1cm} $(1 \leq j \leq 4)$\\
  \end{tabular}
  \caption{Additional preference lists for instance $I$.}
  \label{sm_rm:srcomplexity_c}
\end{subfigure}
\caption[Creation of an instance $I$ of {\sc sr}.]{Creation of an instance $I$ of \acrshort{sr}.}
\label{sm_rm:srcomplexity}
\end{figure}

Finally, set the profile $\sigma = \langle 2n - K + 8, 2K, n-K, 0, n-K, K \rangle$.

We claim that $G$ has a vertex cover of size $\leq K$ if and only if $I$ has a stable matching $M$ such that $p(M) \succeq \sigma$.

Suppose that $G$ has a vertex cover $C$ such that $|C| = k \leq K$. We create a matching $M$ as follows.
\begin{itemize}
	\item Add, $\{a_j, b_j\}$ to $M$ for $1 \leq j \leq 4$ giving $8$ first choices. 
	\item If $v_i \in C$ then add $\{v_i,w_i\}$ and $\{x_i,y_i\}$ to $M$. This means $v_i$ is assigned their first choice in $I$, $w_i$ is assigned their sixth choice, and $x_i$ and $y_i$ both have their second choices.
	\item If $v_i \notin C$ then add $\{v_i,y_i\}$ and $\{w_i,x_i\}$ to $M$. Then $v_i$ is assigned their fifth choice in $I$, $y_i$ and $w_i$ are assigned their first choices, and $x_i$ is assigned their third choice.
\end{itemize}

We now show that $M$ is stable. Agents from $A$ and $B$ are always assigned to each other, and so cannot block $M$. We now look through the four addition types of pair of $M$. Pair $\{v_i,w_i\}$ cannot block $M$ since $w_i$ cannot prefer $v_i$ to their assigned partner. Pair $\{x_i,y_i\}$ cannot block $M$ since $y_i$ cannot prefer $x_i$ to their assigned partner. Pair $\{v_i,y_i\}$ cannot block $M$ since $v_i$ cannot prefer $y_i$ to their assigned partner. Finally, pair $\{w_i,x_i\}$ cannot block $M$ since $x_i$ cannot prefer $w_i$ to their assigned partner. The only remaining possibility is that $\{v_i, v_j\}$ blocks $M$. Assume for contradiction that this is the case. Then, it must be that $\{v_i,y_i\}$ and $\{v_j,y_j\}$ are in $M$. But, by construction this implies that neither $v_i$ nor $v_j$ are in $C$, meaning $C$ is not a valid vertex cover, a contradiction since $\{v_i, v_j\} \in E$.

We are now able to calculate the profile $p(M) = \langle 2n - k + 8, 2k, n-k, 0, n-k, k \rangle$ by totalling the choices for each rank described in the bullet points above. Since $k \leq K$, we have $p(M) \succeq \sigma$.
	
Conversely, suppose $M$ is a stable matching in $I$ such that $p(M) \succeq \sigma$. 

We first show that each agent may only be assigned to a subset of agents shown on their preference lists above. 
\begin{itemize}
	\item For $1 \leq j \leq 4$, since $a_j$ ranks $b_j$ as first choice and vice versa, $\{a_j, b_j\} \in M$; 
	\item $v_i$ can never be assigned lower on their preference list than $y_i$, or else $\{v_i, y_i\}$ would block $M$;
	\item $w_i$ can never be assigned lower on their preference list than $v_i$, or else $\{v_i, w_i\}$ would block $M$;
	\item $x_i$ can never be assigned lower on their preference list than $w_i$, or else $\{w_i, x_i\}$ would block $M$;
	\item $y_i$ cannot be assigned lower on their preference list than $x_i$, or else $\{x_i, y_i\}$ would block $M$, noting that $\{x_i, a_1\} \notin M$ by the first point above;
	\item Finally, suppose $\{v_i, v_j\} \in M$. Then one of $w_i, x_i, y_i$ is unassigned in $M$, a contradiction to the above. So, either $\{v_i, w_i\} \in M$ or $\{v_i, y_i\} \in M$.
\end{itemize}



If $\{v_i, w_i\} \in M$ then we add $v_i$ to $C$. It remains to prove that $C$ is a vertex cover of $G$. Suppose for a contradiction that it is not. Then there exists an edge $\{v_i, v_j\} \in E$ such that $v_i \notin C$ and $v_j \notin C$. By construction this means that $\{v_i, y_i\} \in M$ and $\{v_j, y_j\} \in M$. But then $\{v_i, v_j\}$ blocks $M$, a contradiction.

Finally, suppose for contradiction that $|C| = k > K$. Note that if $v_i \in C$ then $\{v_i, w_i\} \in M$, and in turn $\{x_i, y_i\} \in M$ and also if $v_i \notin C$ then $\{v_i, y_i\} \in M$, and in turn $\{w_i, x_i\} \in M$. Therefore, using similar logic to above we can calculate the profile of $M$ as $p(M) = \langle 2n - k + 8, 2k, n-k, 0, n-k, k \rangle$. Since $k > K$ we have $p(M) \prec \sigma$, a contradiction.

We have shown that $G$ has a vertex cover of size $\leq K$ if and only if $I$ has a stable matching $M$ such that $p(M) \succeq \sigma$. Since the reduction described above can be completed in polynomial time, {\sc rmsr-d} is $\NP$-hard. Finally, as {\sc rmsr-d} is in $\NP$, {\sc rmsr-d} is also $\NP$-complete.
\end{proof}

\begin{corollary}
	\label{pb_opt:thm_gensr}
    {\sc gensr-d} is $\NP$-complete.
\end{corollary}
\begin{proof}
	We use the same reduction and a similar argument  as Theorem \ref{pb_opt:thm_rmsr_gensr} to show that $G$ has a vertex cover of size $\leq K$ if and only if $I$ admits a stable matching $M$ such that $p_r(M) \preceq \sigma_r$.
\end{proof}

We are able to further extend these results looking at a two restricted decision problems. 

\begin{definition}
	We define {\sc rmsr-first-d} as follows. An instance of {\sc rmsr-first-d}, $(I, K)$, comprises an instance $I$ of \acrshort{sr} and an integer $K$. The problem is to decide whether there exists a stable matching $M$ in $I$ such that $p_1 \geq K$, where $p(M) = \langle p_1, p_2, ..., p_n \rangle$.
\end{definition}

\begin{definition}
	We define {\sc gensr-final-d} as follows. An instance of {\sc gensr-final-d}, $(I, K)$, comprises an instance $I$ of \acrshort{sr} and an integer $K$. The problem is to decide whether there exists a minimum regret stable matching $M$ in $I$ such that $p_i \leq K$, where $i=d(M)$ and $p(M) = \langle p_1, p_2, ..., p_n \rangle$.
\end{definition}

\begin{corollary}
	{\sc rmsr-first-d} is $\NP$-complete.
\end{corollary}
\begin{proof}
	Using identical constructions and logic as in Theorem \ref{pb_opt:thm_rmsr_gensr} above, we can see that by considering only the first element of the profiles we are able to prove {\sc rmsr-first-d} is $\NP$-complete.
\end{proof}

  \begin{corollary}
	{\sc gensr-final-d} is $\NP$-complete.
\end{corollary}
\begin{proof}
	We first show that any stable matching $M$ constructed as described in Theorem \ref{pb_opt:thm_rmsr_gensr} will have degree $6$. Suppose for contradiction that there exists a stable matching $M$ with $d(M) \leq 5$. Then $\{w_i,x_i\} \in M$ for all $i$ $(1 \leq i \leq n)$ which implies that $\{v_i,y_i\} \in M$ for all $i$ $(1 \leq i \leq n)$. If we pick any edge $\{v_i, v_j\} \in E$ then $\{v_i, v_j\}$ blocks $M$, a contradiction. 
	
	Now, using a similar proof to Theorem \ref{pb_opt:thm_rmsr_gensr}, we note that any stable matching $M$ constructed from a vertex cover $C$ of size $\leq K$ satisfies $d(M)=6$ and $p_6 \leq K$, where $p(M) = \langle p_1, p_2, ..., p_n \rangle$. Conversely, given a minimum regret stable matching $M$ such that $p_i \leq K$ where $i = d(M)$, it follows that $i=6$. We then proceed as before, obtaining a vertex cover $C$ of size $\leq K$.
\end{proof}

  \section{Experiments and evaluations}
  \label{sm_rm_sec_exps}
  \FloatBarrier

  \subsection{Methodology}
  \label{sm_prof_exp_sec_method}
 
  For our experiments we used randomly-generated data to compare rank-maximal, generous and median\footnote{Recall from Section \ref{smi_rm_sec_contribution} that although the median criterion is not profile-based, we are interested in determining whether the median stable matching more closely approximates a rank-maximal or a generous stable matching, in practice.} stable matchings over a range of measures (cost, sex-equal score, degree, number of agents obtaining their first choice and number of agents who obtain a partner in the lower $a\%$ of their preference list). This final measure (an example of which was given in Section \ref{smi_rm_sec_contribution}) may be more formally defined as the number of agents who obtain a partner between their $b$th choice and $n$th choice inclusive, where $b = (100-a) \frac{n}{100}+1$. We also investigated the effect of varying instance size (in terms of the number of men or women) on these properties. Our experiments explored $19$ separate instance types with the number of men (and women) taking the values of $\{10, 20, ..., 100, 200,..., 1000\}$ and with $1000$ instances tested in each case. All instances tested were complete with uniform distributions on preference lists, as preliminary experimentation showed that this, in general, produces a larger number of stable matchings than using incomplete lists or linear distributions.

The set of all stable matchings of an instance $I$ of \acrshort{smi} were found using \citeauthor{Gus87}'s $O(m + n|\mathcal{M}_S|)$ time Enumeration Algorithm \citep{Gus87}. The Enumeration Algorithm comprises two runs of the extended Gale-Shapley Algorithm \citep{GS62} (finding both the man-optimal and woman-optimal stable matchings), one run of
 the Minimal Differences Algorithm \cite{GI89} (finding all rotations of an instance), creation of the rotation digraph, and finally enumeration of all stable matchings using this digraph \cite{Gus87}.
Using the list of enumerated stable matchings we were then able to compute all the types of optimal stable matchings described above. 
A timeout of $1$ hour was used for this algorithm. Note that since the time complexities of the algorithms described in this paper do not beat the best known for these problems, we did not implement our algorithms to test their performance.
Experiments were carried out on a machine with 32 cores, 8$\times$64GB RAM and Dual Intel\textsuperscript{\textregistered} Xeon\textsuperscript{\textregistered} CPU E5-2697A v4 processors, running Ubuntu version $17.10$. Instance generation and statistics collection programs were written in Python and run on Python version $2.7.14$. The plot and table generation program was written in Python and run using Python version 3.6.1. All other code was written in Java and compiled using Java version $1.8.0$. All Java code was run on a single thread, with GNU Parallel \cite{Tan11} used to run multiple instances in parallel. Java garbage collection was run in serial and a maximum heap size of 1GB was distributed to each thread. Code and data repositories for these experiments can be found at \url{https://doi.org/10.5281/zenodo.2545798} and \url{https://doi.org/10.5281/zenodo.2542703} respectively.

In addition to the above experiments, we also compared the minimum amount of memory required to store edge capacities of the network (based on exponential weights) and associated vb-network (based on vector-based weights) for each instance. The instances described above were used along with five additional instances each for $n \in \{2000, 3000, 4000,$ $5000\}$. Unlike the instances described above for $n \leq 1000$, space requirement calculations for these larger instances were carried out on a machine running Ubuntu version $14.04$ with $4$ cores, $16$GB RAM and Intel\textsuperscript{\textregistered}, Core\textsuperscript{\texttrademark} i7-4790 processors, and compiled using Java version 1.7.0. This machine, compared to the one described previously, was able to calculate solutions to individual instances more quickly, although it could run fewer threads in parallel. In these experiments, we were solving a small number of more complex instances and were not interested in the time taken, but rather in the memory requirements, therefore this change in machine did not impact our results. A larger timeout of $24$ hours was used for each of two runs of the extended Gale-Shapley Algorithm and one run of the Minimal Differences Algorithm. Stable matchings were not enumerated for these instances. All other configurations were the same as before.

A description of how space requirements were calculated now follows. In these calculations we did not assume any particular implementation, but more generally estimated the minimum number of bits required theoretically to store the capacities in each case. As mentioned previously, the Minimal Differences Algorithm was used to find all rotations of an instance,
and only instances that did not timeout and had at least one rotation were used in these calculations. For each rotation, a rotation profile was easily computed, and vector-based weight and exponential weight edge capacities were then calculated directly from these rotation profiles.

Let $R$ be the set of rotations in an instance of \acrshort{smi} and let $\rho$ be a rotation with profile $p(\rho) = \{p_1, p_2, ..., p_n\}$. The \emph{degree} of $p(\rho)$ may be described as the maximum index $i$ for which there exists some $p_i$ such that $p_i \neq 0$. Let $d_t$ denote the maximum degree over all rotations in $R$. An exponential weight $w_e$ was calculated from $p(\rho)$ according to \citeauthor{ILG87}'s \cite{ILG87} original formula of a weight of $n^{n-i}$ for each agent assigned to their $i$th choice. In reality, we reduced this to $d_t^{d_t - i}$ which allowed a smaller number to be stored. Then, the number of bits required to store $w_e$ was calculated as $\left\lceil\log_2 w_e\right\rceil$ with an additional standard $32$-bit word used to store the length of this bit representation\footnote{For the case where $w_e = 0$, we calculated the number of bits required as $1$ with an additional standard $32$-bit word. For both the exponential weight and vector-based weight cases, if $d_t = 0$, we calculated the number of bits required as a standard $32$-bit word.}\textsuperscript{,}\footnote{We note that since $\left\lceil\log_2 w_e\right\rceil$ is in general far greater than $32$ (for large $n$), we chose to use a standard $32$-bit word to store the number of bits for each exponential number as opposed to assuming all exponential numbers are an equal size. As an example, consider an instance of size $n=100$, where there exists one rotation with profile $\langle 1, 0, ..., -1 \rangle$ of length $n$ (requiring $\left\lceil\log_2(100^{99} - 1)\right\rceil = 658$ bits to store the associated exponential number) and several rotations with profile $\langle 0, 0, ..., 0 \rangle$ (requiring $1$ bit to store the associated exponential number). Were we to require all exponential numbers to be stored using the same number of bits, far more space would be used than necessary.}. A vector-based weight $w_v$ was calculated from $p(\rho)$ using lossless vector compression, as described in Section \ref{sm_rm_motivation_sec}. Let $z$ be the number of non-zero elements of $p(\rho)$. Then, the number of bits required to hold indices of non-zero elements is given by $z\left\lceil\log_2 n\right\rceil$ (since the length of the profile is bounded by $n$), and the number of bits required to hold values of non-zero elements is given by $z (\left\lceil\log_2 2n \right\rceil + 1)$ (since each element is bounded below by $-2n$ and above by $2n$). The addition of a $32$-bit word then allowed the number $z$ to be stored. Finally, two additional $32$-bit words are required over the whole vb-network (for storing the numbers $n$ and $2n$).

Correctness testing was conducted in the following way. All stable matchings produced by all instances of size up to $n=1000$ were checked for \emph{(1) capacity}: each man (woman) may only be assigned to one woman (man) respectively; and \emph{(2) stability}: no blocking pair exists. Additional correctness testing was also conducted for all instances of size $n=10,...,60$. For these instances, in addition to the above testing, a process took place to determine whether the number of stable matchings found by the Enumeration Algorithm matched the number found by an IP model. This was developed in Python version 2.7.14 with the IP modelling framework PuLP (version 1.6.9) \cite{MOD11} using the CPLEX solver \cite{Cpl19}, version 12.8.0. Each instance was run on a single thread with a time limit of $10$ hours (all runs completed within this time), using the same machine that was used for instances of size $n \leq 1000$.
All correctness tests passed successfully.

 \subsection{Experimental results summary}


Table \ref{sm_rm_res_instanceParams} shows the $19$ instance types of size up to $n=1000$. Instance types are labelled according to $n$, e.g., S100 is the instance type containing instances where $n = 100$. In Columns $3$ and $4$ of this table we show the mean number of rotations $|\mathcal{R}|_{av}$ and the mean number of stable matchings $|\mathcal{M}_S|_{av}$ per instance type respectively. Column $5$ displays the number of instances that did not complete within the $1$ hour timeout and the mean time taken to run the Enumeration Algorithm over a completed instance is shown in Column $6$. Figure \ref{sm_rm_plot_av_stab_vs_nlogn} shows the mean number of stable matchings as $n$ increases. 

Figures \ref{sm_rm_plot_first_vs_n} and \ref{sm_rm_plot_degree_vs_n} show the mean number of first choices and mean degree respectively, for rank-maximal, median and generous stable matchings, as $n$ increases. Figures \ref{sm_rm_plot_egal_vs_n} and \ref{sm_rm_plot_se_vs_n} show the mean cost and sex-equal score respectively, of the above types of stable matchings with the addition of their respective mean optimal values. Note that for any given instance, the cost and sex-equal scores of all rank-maximal stable matchings are equal. This is also the case for generous stable matchings. Our definition of the median stable matching (given in Section \ref{sec_optimality}) ensures that the median stable matching is unique. Thus, we are not required to find an optimal matching with best cost or sex-equal score. Data for these plots may be seen as Tables \ref{sm_rm_res_generalStats}, \ref{sm_rm_res_RM}, \ref{sm_rm_res_GEN} and \ref{sm_rm_res_GM} in Appendix \ref{app_results_tables}, where Table \ref{sm_rm_res_generalStats} shows statistics for cost and sex-equal score, and Tables \ref{sm_rm_res_RM}, \ref{sm_rm_res_GEN} and \ref{sm_rm_res_GM} display statistics for rank-maximal, generous and median stable matchings. Additionally, these latter tables show the minimum, maximum and mean number of assignments $l_a$ in the last $a\%$ of all preference lists. 

Finally, Figure \ref{fig_sm_rm_mot_exp} (with associated Table \ref{sm_rm_space_table} in Appendix \ref{app_results_tables}) shows a plot comparing the mean number of bits required to store edge capacities of a network and vb-network using exponential weights and vector-based weights respectively. In this plot, circles represent the mean number of bits required for different values of $n$. The exact space requirements were calculated according to the process described in Section \ref{sm_prof_exp_sec_method}.   Solid circles represent data points $n\in \{100, 200, ..., 1000\}$ and these were used to calculate the best fit curves shown when assuming a second order polynomial model. $90\%$ confidence intervals using the $5$th and $95$th percentile measurements for each representation are also displayed. Above $n=1000$ we extrapolate up until $n=100,000$, showing the expected trend with an increasing $n$. Data points for instances of size $n \in \{10, 20, ..., 90, 2000, 3000, 4000, 5000\}$ are represented as unfilled circles. These data were not used to calculate the best fit curves, but are added to the figure to help determine the validity of our extrapolation mentioned above.

The main findings of these experiments are:

    \begin{itemize}
    	\item \emph{Number of stable matchings}: From Figure \ref{sm_rm_plot_av_stab_vs_nlogn} we can see that the mean number of stable matchings increases with instance size. \citet{LP09} showed that the number of expected stable matchings in an instance of size $n$ tends to the order of $n \log n$. Our experiments confirm this result and show a reasonably linear correlation between $n \log n$ and the mean number of stable matchings for instances with $n \geq 100$. 

    	\item \emph{Number of first choices}: As expected, rank-maximal stable matchings obtain the largest number of first choices by some margin, when compared to generous and median stable matchings. When looking at the mean number of first choices, this margin appears to increase from almost 1:1 for instance type S10 ($6.9$ for rank-maximal compared to $6.0$ and $6.1$ for generous and median respectively) to approximately 3:1 for instance type S1000 ($158.4$ for rank-maximal compared to $63.5$ and $71.5$ for generous and median respectively). Generous and median stable matchings are far more aligned, however generous is increasingly outperformed by median on the mean number of first choices with ratios starting at around 1:1 for S10, gradually increasing to 1.1:1 for S1000. This is summarised in the plot shown in Figure \ref{sm_rm_plot_first_vs_n}.
    	\item \emph{Number of last $a\%$ choices}: For rank-maximal stable matchings, the mean number of assignments in the final $10\%$ of preference lists was low, increasing from $0.4$ for instance type S10 to $1.4$ for instance type S1000. Note that this increase is far lower than the rate of instance size increase. The mean number of generous stable matching choices in the final $50\%$ of preference lists decreased from $2.4$ to $0.0$ over all instance types.  As the generous criteria minimises final choices, this is likely due to the number of stable matchings increasing with larger instance size. Finally, it is interesting to note that the mean number of median stable matching choices in the final $20\%$ of preference lists decreases from $0.5$ to $0.0$ despite the number of lower ranked choices not being directly minimised. Figure \ref{sm_rm_plot_degree_vs_n} shows how the mean degree changes with respect to $n$ for rank-maximal, median and generous stable matchings. We can see that on average the rank-maximal criteria performs badly, putting men or women very close to the end of their preference list. As above the generous criteria outperforms either of the other optimisations, with median somewhere in between.
    	
    	\item \emph{Cost and sex-equal score}: The range between minimum and maximum optimal costs and sex-equal scores over all instance types (Table \ref{sm_rm_res_generalStats}) is small when compared with results found for these measures in the rank-maximal, generous and median stable matching experiments. 
    	In Figure \ref{sm_rm_plot_egal_vs_n}, we can see that a generous stable matching is a close approximation of an egalitarian stable matching in practice. This is followed by the median and then the rank-maximal solution concepts. A similar, though less pronounced, result holds for the sex-equal stable matching, as can be seen in Figure \ref{sm_rm_plot_se_vs_n}.

    	\item \emph{Median stable matchings:} We can see from Figures \ref{sm_rm_plot_first_vs_n}, \ref{sm_rm_plot_degree_vs_n}, \ref{sm_rm_plot_egal_vs_n} and \ref{sm_rm_plot_se_vs_n} that over all measures, a median stable matching, on average, more closely approximates a generous stable matching. One possible reason for this is as follows. A rank-maximal stable matching may be seen as prioritising some agents obtaining high choice partners (out of their set of possible partners in any stable matching), possibly at the expense of some agents obtaining low choice partners (out of their set of possible partners in any stable matching). A generous stable matching on the other hand, ensures no-one is assigned to a partner of worse rank than the degree of a minimum-regret stable matching, and so it is less likely that agents achieve low choice partners and, consequently (by similar reasoning to the rank-maximal case), less likely that agents achieve high choice partners as well. Thus, a median stable matching, which assigns each agent with their middle-choice partner (out of the set of all the possible partners in any stable matching) is more likely to approximate a generous stable matching.


    	\item \emph{Space requirement:} From Figure \ref{fig_sm_rm_mot_exp}, we can see clearly that the exponential weights of the network require more space on average than the vector-based weights of the associated vb-network, and that this difference increases as $n$ grows large.  Note that the additional small number of data points (not used to calculate the curves) at $n\in \{2000, 3000, 4000, $ $5000\}$ fit this model well. At $n=1000$, the vector-based weights require approximately $10$ times less space than the exponential weights. Finally, we can see that at $n=100,000$, we expect the vector-based weights to be around $100$ times less costly in terms of space than the exponential weights, with the space requirement for the exponential weights nearing $1$GB.

    	Finally, we show that for a specific family of instances, the result above is even more pronounced. Consider the family of \acrshort{sm} instances represented by instance $I_1$ in Figure \ref{sm_rm_fig_space_example}, where $n$ is even. In $I_1$, the ``...'' symbol in each man's preference list indicates that all other women, not already specified, are listed in any order (after his first choice and before his last choice). Similarly in each woman's preference list, all other men, not already specified, are listed in any order after her second choice. For $n=100,000$, $I_1$ requires over $10$GB to store exponential weights of the network, and only $0.64$MB to store the equivalent vector-based weights of the vb-network. This shows that in certain circumstances, the vector-based weights can be over $100,000$ times less costly in terms of space than the exponential weights. 
    	
      \end{itemize}

 \begin{figure}[]
\centering
   \begin{subfigure}[t]{0.3\textwidth}
   Men's preferences:\\
   \begin{tabular}{p{0.8cm} p{0.6cm} p{0.6cm} p{0.6cm}}
$m_1$: & $w_1$ & ... & $w_2$\\
$m_2$: & $w_2$ & ... & $w_1$\\
...\\
$m_{2i - 1}$: & $w_{2i - 1}$ & ... & $w_{2i}$\\
$m_{2i}$: & $w_{2i}$ & ... & $w_{2i - 1}$\\
...\\
$m_{n-1}$: & $w_{n-1}$ & ... & $w_n$\\
$m_n$: & $w_n$ & ... & $w_{n-1}$\\
\end{tabular}
  \end{subfigure}
  \hspace{1cm}
   \begin{subfigure}[t]{0.3\textwidth}
  Women's preferences:\\
  \begin{tabular}{p{0.8cm} p{0.6cm} p{0.6cm} p{0.6cm}}
$w_1$: & $m_2$ & $m_1$ & ...\\
$w_2$: & $m_1$ & $m_2$ & ...\\
...\\
$w_{2i - 1}$: & $m_{2i}$ & $m_{2i - 1}$ & ...\\
$w_{2i}$: & $m_{2i - 1}$ & $m_{2i}$ & ...\\
...\\
$w_{n-1}$: & $m_n$ & $m_{n-1}$ & ...\\
$w_n$: & $m_{n-1}$ & $m_n$ & ...\\
\end{tabular}
  \end{subfigure}
  \caption[{\sc sm} instance $I_1$.]{\acrshort{sm} instance $I_1$, with $i$ satisfying $1\leq i\leq n/2$.}
  \label{sm_rm_fig_space_example}
    \end{figure}


\begin{table}[] \centerline{\begin{tabular}{ R{1.5cm} | R{1.5cm} R{1.5cm} R{1.5cm} R{1.5cm} R{2cm} }\hline\hline Case & $n$ & $|\mathcal{R}|_{av}$ & $|\mathcal{M}|_{av}$ & Timeout & Time (ms) \\ 
\hline S10 & $10$ & $1.8$ & $3.0$ & $0$ & $50.3$ \\ 
 S20 & $20$ & $4.2$ & $6.5$ & $0$ & $60.2$ \\ 
 S30 & $30$ & $6.5$ & $10.9$ & $0$ & $75.2$ \\ 
 S40 & $40$ & $8.9$ & $15.7$ & $0$ & $93.2$ \\ 
 S50 & $50$ & $11.2$ & $20.9$ & $0$ & $113.8$ \\ 
 S60 & $60$ & $13.4$ & $27.2$ & $0$ & $133.9$ \\ 
 S70 & $70$ & $15.9$ & $34.0$ & $0$ & $163.7$ \\ 
 S80 & $80$ & $18.2$ & $40.6$ & $0$ & $205.5$ \\ 
 S90 & $90$ & $20.0$ & $46.4$ & $0$ & $235.1$ \\ 
 S100 & $100$ & $22.4$ & $54.2$ & $0$ & $278.0$ \\ 
 S200 & $200$ & $41.9$ & $138.8$ & $0$ & $1084.5$ \\ 
 S300 & $300$ & $59.2$ & $231.0$ & $0$ & $2886.1$ \\ 
 S400 & $400$ & $76.2$ & $337.6$ & $0$ & $7972.7$ \\ 
 S500 & $500$ & $90.7$ & $442.0$ & $0$ & $15934.8$ \\ 
 S600 & $600$ & $105.8$ & $566.1$ & $0$ & $32925.3$ \\ 
 S700 & $700$ & $119.1$ & $675.5$ & $0$ & $50802.4$ \\ 
 S800 & $800$ & $131.4$ & $804.0$ & $1$ & $87169.2$ \\ 
 S900 & $900$ & $144.9$ & $937.6$ & $0$ & $128878.0$ \\ 
 S1000 & $1000$ & $157.6$ & $1115.2$ & $1$ & $196029.9$ \\ 
 \hline\hline \end{tabular}} \caption{General instance information and algorithm timeout results.} \label{sm_rm_res_instanceParams} \end{table} 

\begin{figure}
\centering
  		\includegraphics[scale=0.65]{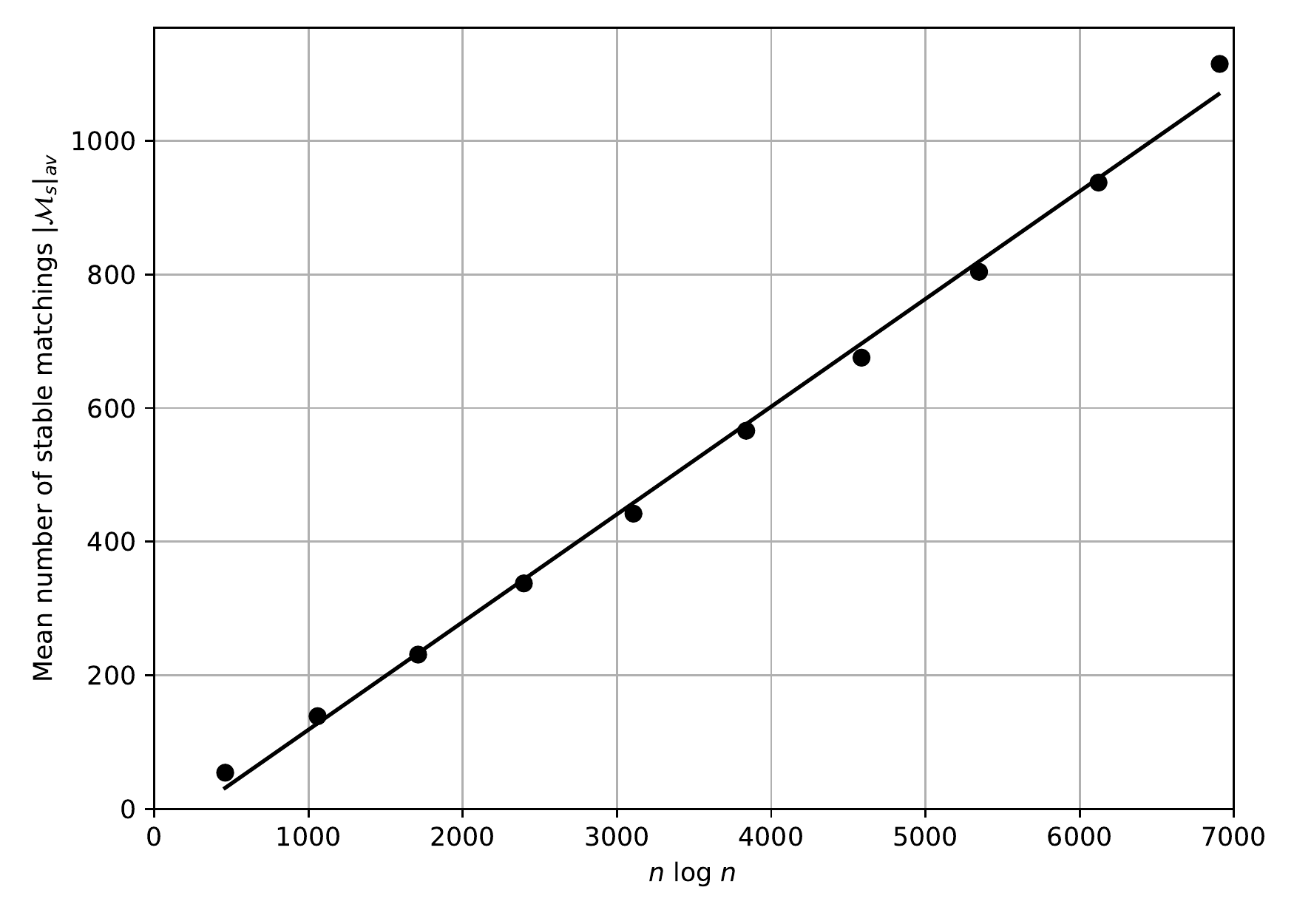}
  		\caption[Plot of the mean number of stable matchings.]{Plot of the mean number of stable matchings $|\mathcal{M}_S|_{av}$ with increasing $n$, where $n$ is the number of men or women. A linear model has been assumed for the best-fit line.}
		\label{sm_rm_plot_av_stab_vs_nlogn}
\end{figure}

\begin{figure}
\centering
   	\includegraphics[scale=0.65]{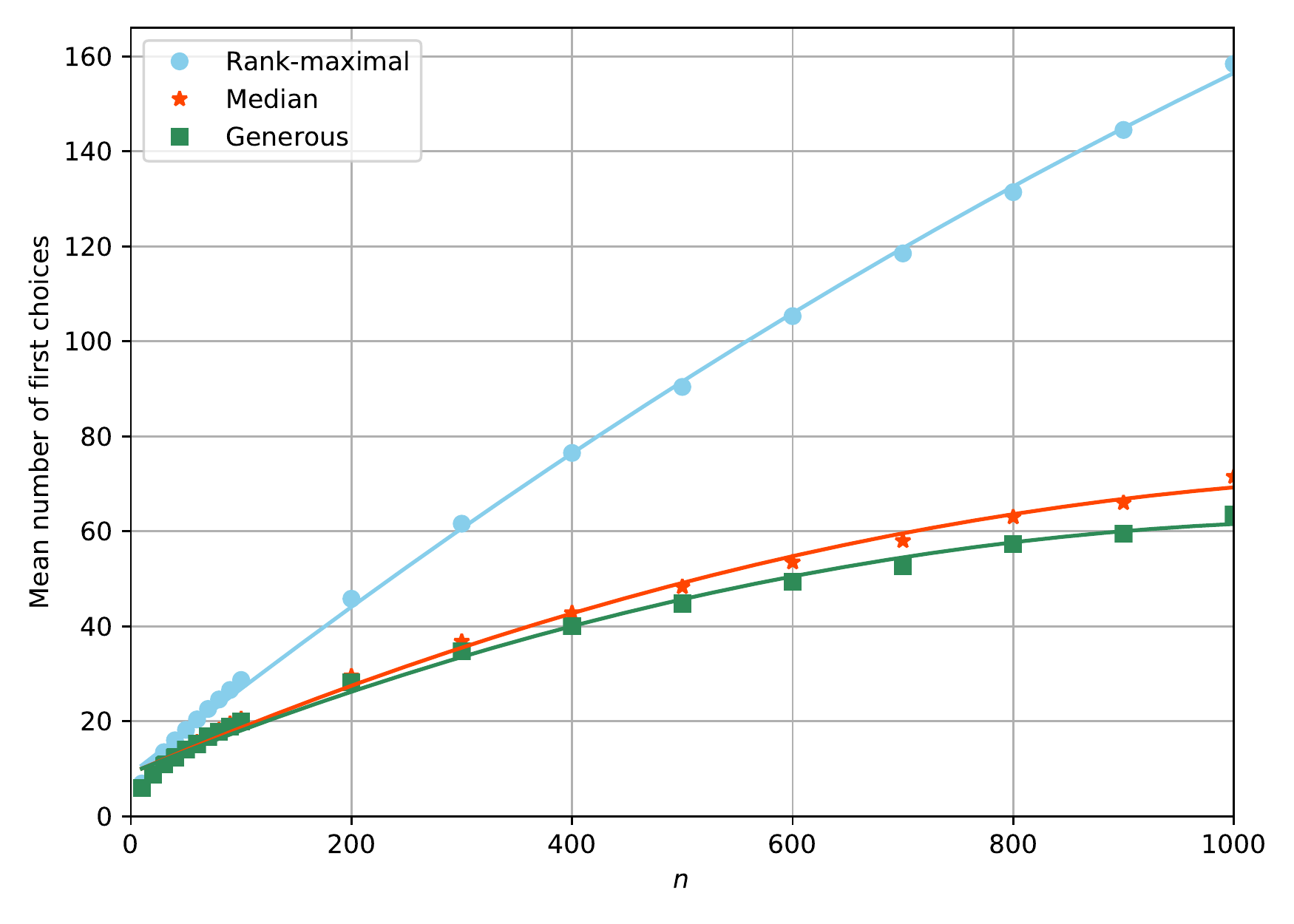}
  		\caption[Plot of the mean number of first choices for rank-maximal, median and generous stable matchings.]{Plot of the mean number of first choices for rank-maximal, median and generous stable matchings with increasing $n$, where $n$ is the number of men or women. A second order polynomial model has been assumed for all best-fit lines.}
  		\label{sm_rm_plot_first_vs_n}
\end{figure}

\begin{figure}
\centering
  		\includegraphics[scale=0.65]{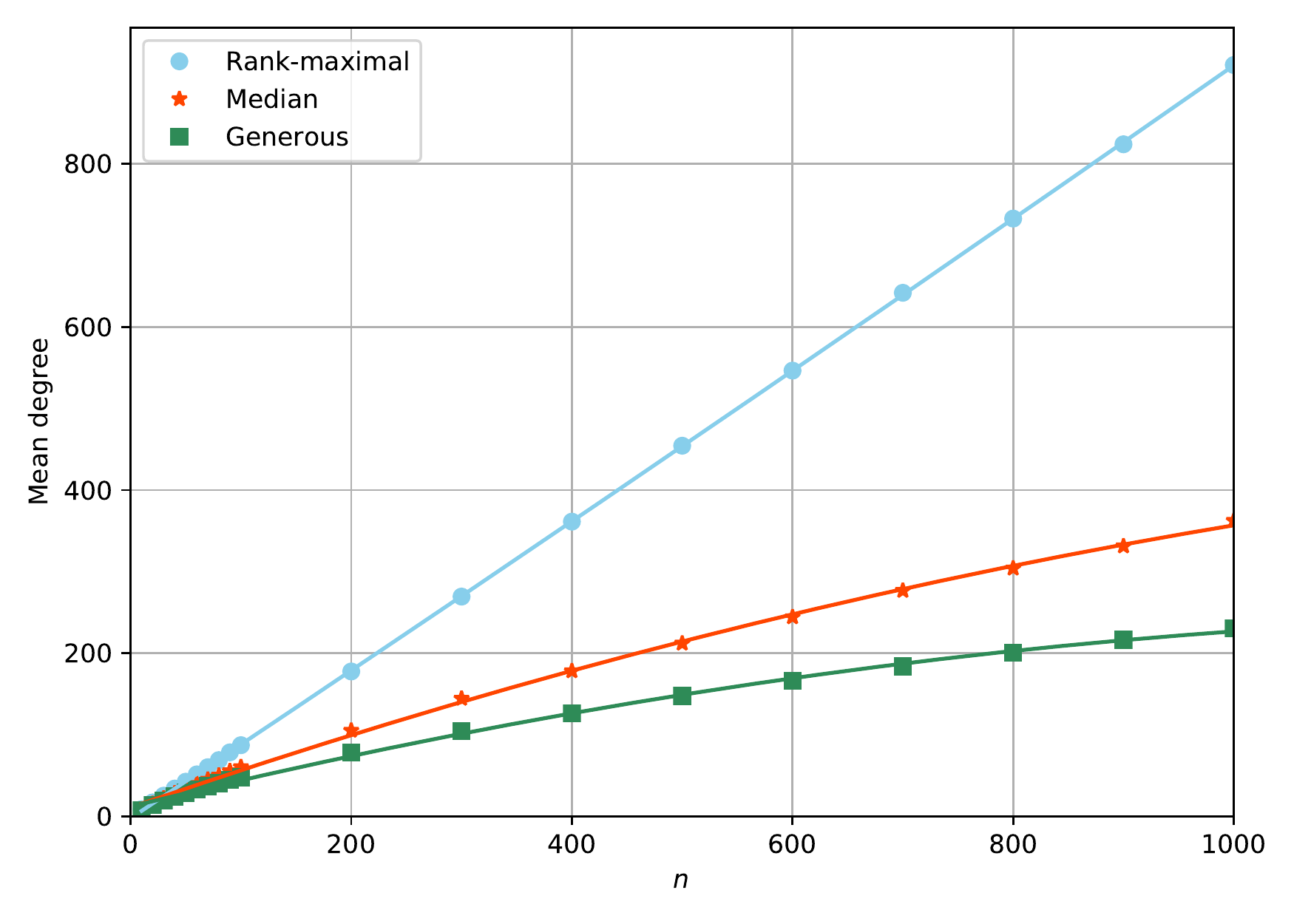}
  		\caption[Plot of the mean degree for rank-maximal, median and generous stable matchings.]{Plot of the mean degree for rank-maximal, median and generous stable matchings with increasing $n$, where $n$ is the number of men or women and the degree of a matching is the rank of a worst ranking man or woman. A second order polynomial model has been assumed for all best-fit lines.}
  		\label{sm_rm_plot_degree_vs_n}
\end{figure}

\begin{figure}
\centering
  		\includegraphics[scale=0.65]{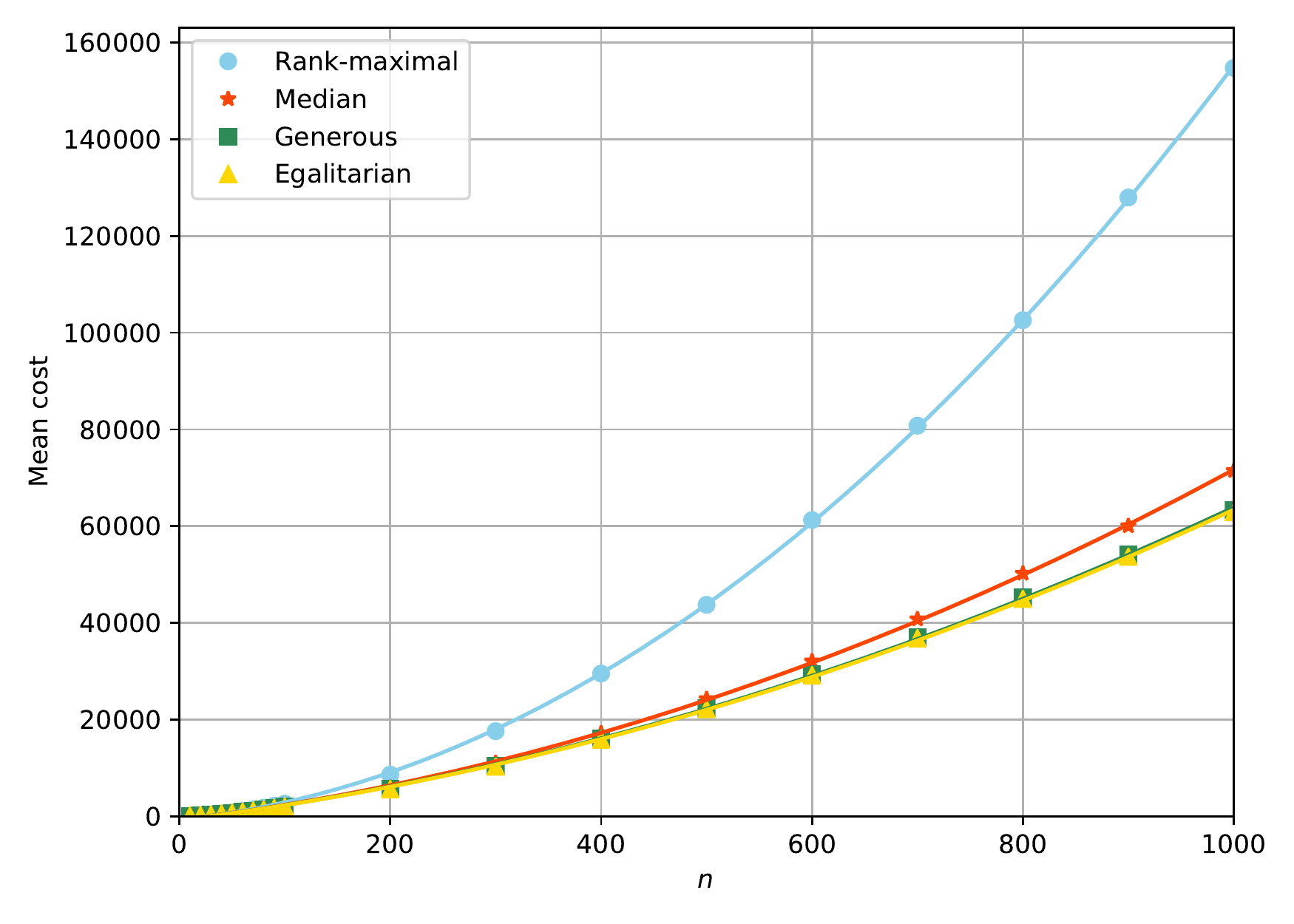}
  		\caption[Plot of the mean cost for rank-maximal, median, generous and egalitarian stable matchings.]{Plot of the mean cost for rank-maximal, median, generous and egalitarian stable matchings with increasing $n$, where $n$ is the number of men or women and the cost of a matching is the sum of ranks of all men and women. A second order polynomial model has been assumed for all best-fit lines.}
  		\label{sm_rm_plot_egal_vs_n}
\end{figure}

\begin{figure}
\centering
  		\includegraphics[scale=0.65]{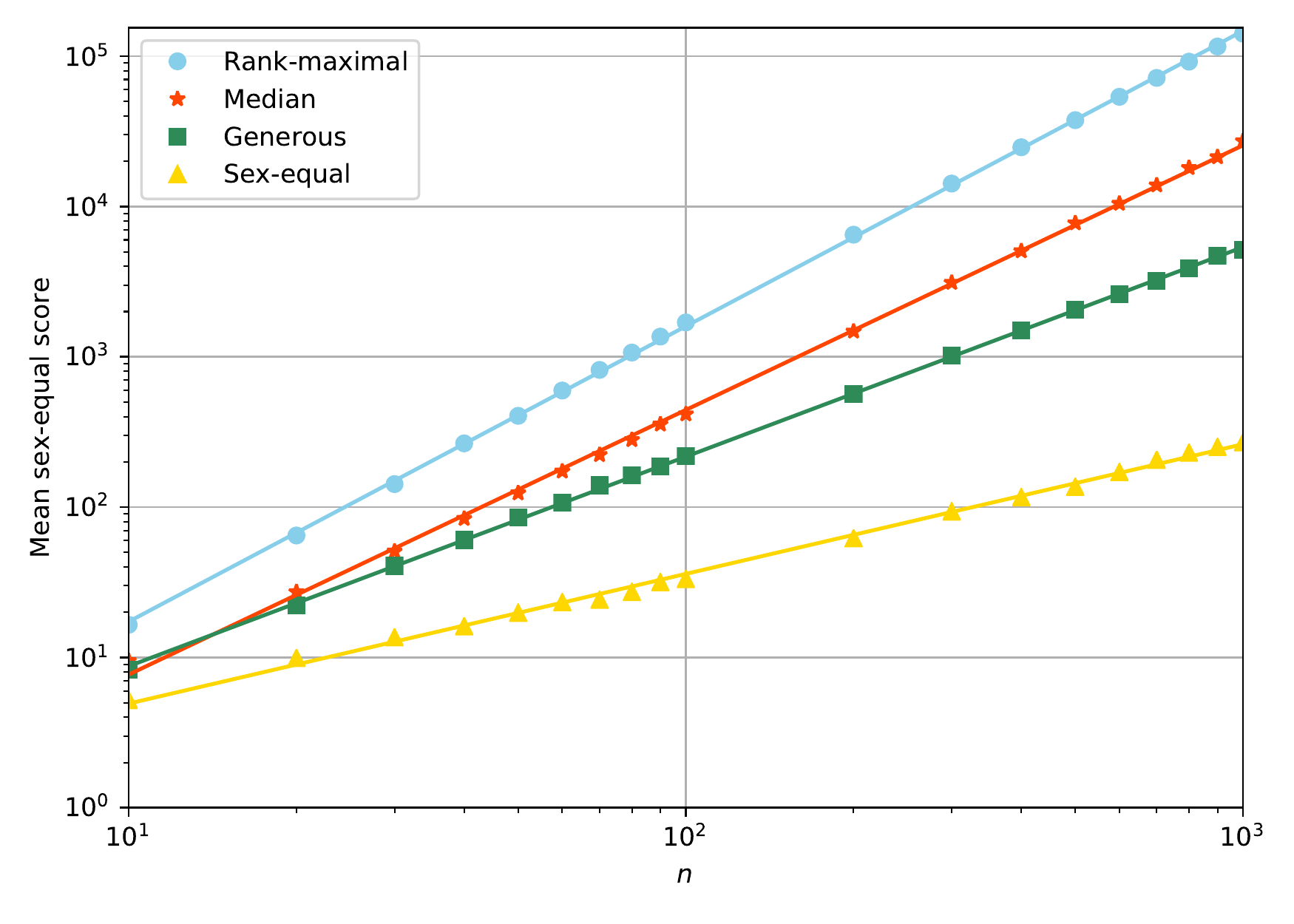}
  		\caption[A $\log$-$\log$ plot of the mean sex-equal score for rank-maximal, median, generous and sex-equal stable matchings.]{A $\log$-$\log$ plot of the mean sex-equal score for rank-maximal, median, generous and sex-equal stable matchings with increasing $n$, where $n$ is the number of men or women and the sex-equal score of a matching is the absolute difference in cost between the set of men and set of women. A first order polynomial model has been assumed for all best-fit lines.}
  		\label{sm_rm_plot_se_vs_n}	
\end{figure}

\begin{figure}
\centering
  		\includegraphics[scale=0.65]{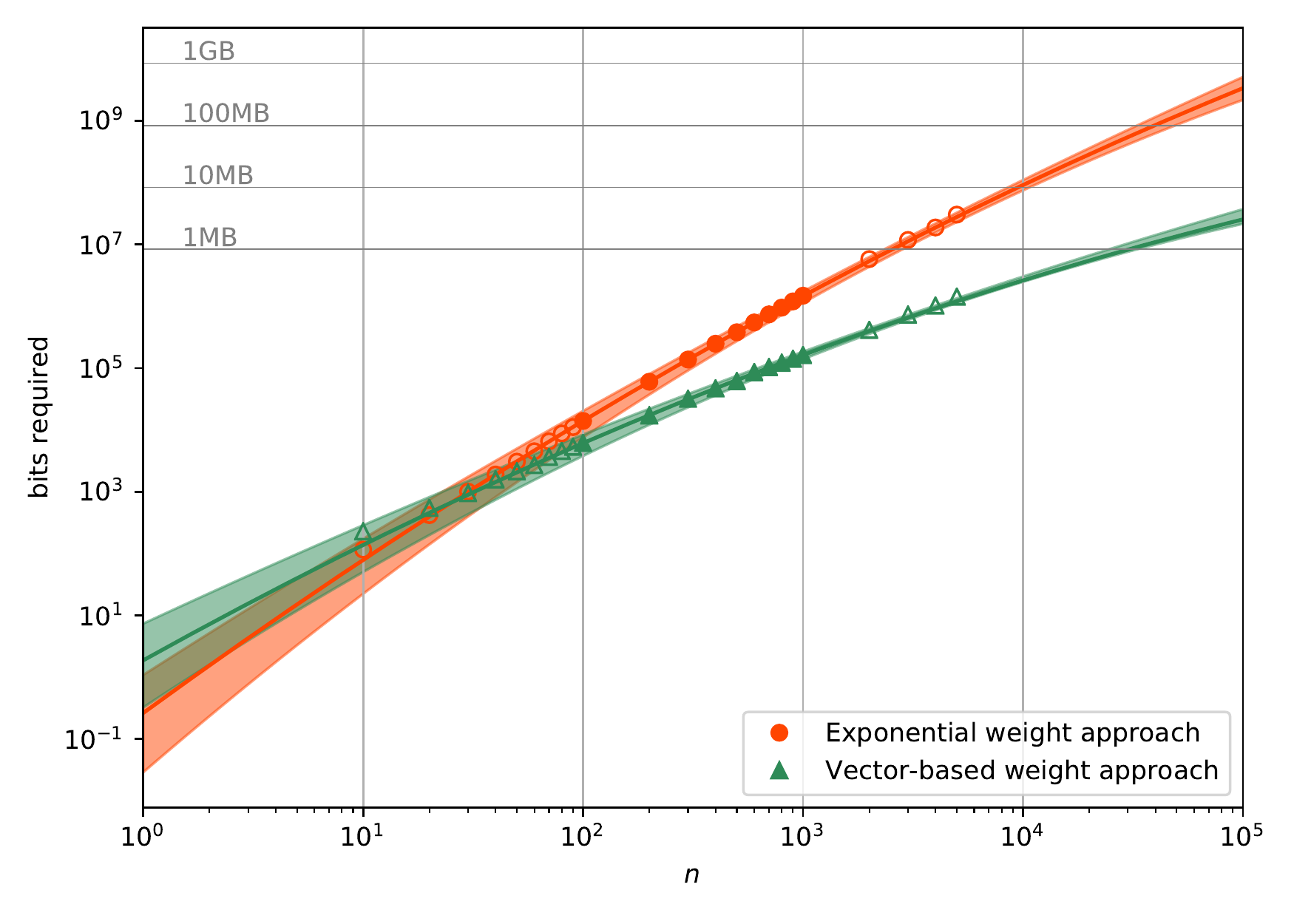}
  		\caption[A $\log$-$\log$ plot of the number of bits required to store a network and vb-network.]{A $\log$-$\log$ plot of the number of bits required to store a network and vb-network for varying the number of men or women $n$ up to $n=100,000$, comparing exponential weight and vector-based weight representations. A second order polynomial model has been assumed for all best-fit lines.}
  		\label{fig_sm_rm_mot_exp}	
\end{figure}

\FloatBarrier
\section{Conclusions and future work}
\label{sec:sm_rm_fw}
In this paper we have described a new method for computing rank-maximal and generous stable matchings for an instance of \acrshort{smi} using polynomially-bounded weight vectors that avoids the use of weights that can be exponential in the number of men. By using this approach we were able to avoid high-weight calculation problems such as overflow, inaccuracies and limitations in memory that may occur with some data types. We were also able to demonstrate an approximate factor of $10$ improvement when using polynomially-bounded weight vectors with vector compression, as opposed to exponential weights, in terms of the space required to store the network (used to find a maximum flow when computing the optimal matchings above) for instances where $n=1000$. We also showed that for a specific instance of size $n=100,000$, the space required to store exponential weights was over $10$GB, whereas the vector-based weights were over $100,000$ times less costly, requiring only $0.64$MB. This improvement is expected to increase further as $n$ grows large. An additional benefit to our new approach is that as all operations are conducted on rotation profiles (which denote the change in number of men or women with an $i$th choice partner), it is arguable that our algorithm is more transparent in terms of its execution. 

Finally, we also showed that the problem of finding rank-maximal and generous stable matchings in \acrshort{sr} is $\NP$-hard. This results still applies even in the restricted cases where we are either finding a stable matching that maximises the number of first choices or finding a minimum regret stable matching $M$ that minimises the number of $d$th choices where $d=d(M)$.

  In Section \ref{sm_rm_intro}, two potential improvements that could be made to the process of finding a rank-maximal stable matching in an instance of \acrshort{smi} were highlighted. First was the adaptation of \citeauthor{Orl13}'s \cite{Orl13} Max Flow algorithm to work in the vector-based setting. This adaptation would result in a time complexity of $O(nm^2)$ to find a rank-maximal stable matching, improving on the method outlined in this paper by a factor of $\log n$, however it is not clear that \citeauthor{Orl13}'s algorithm can be adapted to the vb-flow setting. Additionally, \citet{Fed92} used an entirely different technique based on weighted \acrshort{sat} for finding a rank-maximal stable matching in $O(n^{0.5}m^{1.5})$ time. It remains to be seen if this could be adapted to work in the vector-based setting.


\bibliographystyle{plainnat}
\bibliography{matching}

\newpage
\appendix
\noindent
\LARGE {\bf Appendix}
\normalsize
\section{Experimental results tables}
\label{app_results_tables}
\FloatBarrier

This section contains tables of results for the experiments conducted in Section \ref{sm_rm_sec_exps}.

\begin{table}[h] \centerline{\begin{tabular}{ R{1.5cm} | R{1.5cm} R{1.5cm} R{1.5cm} R{2cm} R{1.5cm} R{1.5cm} }\hline\hline & \multicolumn{3}{c}{Cost} & \multicolumn{3}{c}{Sex-equal score} \\ 
Case & Min & Max & Mean & Min & Max & Mean \\ 
\hline S10 & $40$ & $78$ & $58.0$ & $0$ & $41$ & $5.2$ \\ 
 S20 & $113$ & $215$ & $167.2$ & $0$ & $71$ & $9.9$ \\ 
 S30 & $243$ & $371$ & $311.2$ & $0$ & $137$ & $13.6$ \\ 
 S40 & $396$ & $572$ & $484.3$ & $0$ & $130$ & $16.1$ \\ 
 S50 & $567$ & $821$ & $679.5$ & $0$ & $218$ & $19.8$ \\ 
 S60 & $730$ & $1057$ & $896.2$ & $0$ & $310$ & $23.3$ \\ 
 S70 & $955$ & $1299$ & $1133.2$ & $0$ & $255$ & $24.2$ \\ 
 S80 & $1164$ & $1609$ & $1390.0$ & $0$ & $217$ & $27.2$ \\ 
 S90 & $1447$ & $1909$ & $1660.9$ & $0$ & $178$ & $31.6$ \\ 
 S100 & $1663$ & $2179$ & $1947.0$ & $0$ & $314$ & $33.2$ \\ 
 S200 & $4865$ & $6257$ & $5554.9$ & $0$ & $411$ & $62.0$ \\ 
 S300 & $9444$ & $11301$ & $10250.0$ & $0$ & $528$ & $93.4$ \\ 
 S400 & $14755$ & $16999$ & $15847.1$ & $0$ & $885$ & $116.3$ \\ 
 S500 & $20610$ & $23767$ & $22137.6$ & $0$ & $1158$ & $136.3$ \\ 
 S600 & $27502$ & $31147$ & $29147.8$ & $0$ & $1311$ & $170.8$ \\ 
 S700 & $35117$ & $38975$ & $36772.6$ & $0$ & $1352$ & $206.2$ \\ 
 S800 & $42372$ & $47962$ & $44930.2$ & $0$ & $1809$ & $230.0$ \\ 
 S900 & $50650$ & $56289$ & $53658.2$ & $0$ & $1812$ & $251.3$ \\ 
 S1000 & $59776$ & $65571$ & $62875.9$ & $0$ & $2295$ & $269.4$ \\ 
 \hline\hline \end{tabular}} \caption{Optimal costs and sex-equal scores.} \label{sm_rm_res_generalStats} \end{table}

 \begin{landscape}
 \thispagestyle{empty}
\begin{table}[] \centerline{\begin{tabular}{ R{1cm} | R{0.9cm} R{0.9cm} R{0.9cm} R{1.5cm} R{0.9cm} R{0.9cm} R{1.5cm} R{0.9cm} R{0.9cm} R{1.8cm} R{1.4cm} R{1.4cm} R{1.7cm} R{1.4cm} R{1.4cm} }\hline\hline & \multicolumn{3}{c}{$f$} & \multicolumn{3}{c}{$l_{10}$} & \multicolumn{3}{c}{Degree} & \multicolumn{3}{c}{Cost} & \multicolumn{3}{c}{Sex-equal score} \\ 
Case & Min & Max & Mean & Min & Max & Mean & Min & Max & Mean & Min & Max & Mean & Min & Max & Mean \\ 
\hline S10 & $0$ & $13$ & $6.9$ & $0.0$ & $3.0$ & $0.4$ & $4$ & $10$ & $8.6$ & $40$ & $87$ & $60.4$ & $0$ & $63$ & $16.5$ \\ 
 S20 & $2$ & $19$ & $10.6$ & $0.0$ & $4.0$ & $0.4$ & $9$ & $20$ & $17.0$ & $124$ & $275$ & $182.3$ & $0$ & $205$ & $64.8$ \\ 
 S30 & $6$ & $22$ & $13.5$ & $0.0$ & $4.0$ & $0.5$ & $12$ & $30$ & $25.3$ & $247$ & $496$ & $349.3$ & $0$ & $374$ & $142.4$ \\ 
 S40 & $6$ & $26$ & $16.0$ & $0.0$ & $4.0$ & $0.6$ & $18$ & $40$ & $34.0$ & $407$ & $810$ & $564.2$ & $1$ & $650$ & $265.6$ \\ 
 S50 & $7$ & $32$ & $18.2$ & $0.0$ & $6.0$ & $0.6$ & $22$ & $50$ & $42.4$ & $604$ & $1150$ & $809.7$ & $1$ & $902$ & $405.0$ \\ 
 S60 & $11$ & $35$ & $20.4$ & $0.0$ & $4.0$ & $0.6$ & $25$ & $60$ & $51.4$ & $742$ & $1617$ & $1095.5$ & $13$ & $1332$ & $596.8$ \\ 
 S70 & $11$ & $36$ & $22.6$ & $0.0$ & $7.0$ & $0.6$ & $29$ & $70$ & $60.2$ & $1040$ & $2054$ & $1421.3$ & $2$ & $1730$ & $818.4$ \\ 
 S80 & $11$ & $42$ & $24.6$ & $0.0$ & $4.0$ & $0.7$ & $34$ & $80$ & $68.9$ & $1328$ & $2586$ & $1780.0$ & $3$ & $2106$ & $1067.6$ \\ 
 S90 & $14$ & $43$ & $26.6$ & $0.0$ & $5.0$ & $0.7$ & $41$ & $90$ & $78.4$ & $1501$ & $3156$ & $2186.0$ & $7$ & $2644$ & $1363.4$ \\ 
 S100 & $13$ & $45$ & $28.7$ & $0.0$ & $5.0$ & $0.7$ & $40$ & $100$ & $87.2$ & $1807$ & $3609$ & $2617.4$ & $32$ & $2983$ & $1693.5$ \\ 
 S200 & $24$ & $73$ & $45.8$ & $0.0$ & $7.0$ & $0.9$ & $98$ & $200$ & $177.6$ & $5617$ & $12532$ & $8616.7$ & $1277$ & $11096$ & $6494.2$ \\ 
 S300 & $33$ & $96$ & $61.6$ & $0.0$ & $6.0$ & $1.0$ & $110$ & $300$ & $269.4$ & $10333$ & $24028$ & $17608.0$ & $1677$ & $21562$ & $14213.9$ \\ 
 S400 & $46$ & $110$ & $76.5$ & $0.0$ & $8.0$ & $1.1$ & $230$ & $400$ & $361.3$ & $19149$ & $39409$ & $29521.9$ & $10331$ & $35863$ & $24778.8$ \\ 
 S500 & $55$ & $132$ & $90.4$ & $0.0$ & $7.0$ & $1.1$ & $271$ & $500$ & $454.4$ & $26729$ & $60640$ & $43725.3$ & $15185$ & $56490$ & $37529.6$ \\ 
 S600 & $67$ & $155$ & $105.3$ & $0.0$ & $9.0$ & $1.2$ & $365$ & $600$ & $546.5$ & $41102$ & $81496$ & $61248.3$ & $28874$ & $76436$ & $53676.6$ \\ 
 S700 & $65$ & $162$ & $118.5$ & $0.0$ & $9.0$ & $1.3$ & $396$ & $700$ & $641.9$ & $52098$ & $108000$ & $80778.1$ & $37230$ & $101646$ & $71714.5$ \\ 
 S800 & $77$ & $178$ & $131.4$ & $0.0$ & $8.0$ & $1.3$ & $402$ & $800$ & $732.9$ & $60915$ & $134559$ & $102579.6$ & $40623$ & $126963$ & $91993.6$ \\ 
 S900 & $94$ & $198$ & $144.5$ & $0.0$ & $8.0$ & $1.3$ & $491$ & $900$ & $824.0$ & $86785$ & $183870$ & $127944.3$ & $69419$ & $175936$ & $115909.3$ \\ 
 S1000 & $104$ & $208$ & $158.4$ & $0.0$ & $8.0$ & $1.4$ & $652$ & $1000$ & $921.2$ & $104836$ & $205341$ & $154730.8$ & $83732$ & $195439$ & $141113.6$ \\ 
 \hline\hline \end{tabular}} \caption{Results for rank-maximal stable matchings over various measures.} \label{sm_rm_res_RM} \end{table} 
 \end{landscape}

  \begin{landscape}
 \thispagestyle{empty}
\begin{table}[] \centerline{\begin{tabular}{ R{1cm} | R{0.9cm} R{0.9cm} R{0.9cm} R{1.5cm} R{0.9cm} R{0.9cm} R{1.5cm} R{0.9cm} R{0.9cm} R{1.8cm} R{1.4cm} R{1.4cm} R{1.7cm} R{1.4cm} R{1.4cm} }\hline\hline & \multicolumn{3}{c}{$f$} & \multicolumn{3}{c}{$l_{50}$} & \multicolumn{3}{c}{Degree} & \multicolumn{3}{c}{Cost} & \multicolumn{3}{c}{Sex-equal score} \\ 
Case & Min & Max & Mean & Min & Max & Mean & Min & Max & Mean & Min & Max & Mean & Min & Max & Mean \\ 
\hline S10 & $0$ & $12$ & $6.0$ & $0.0$ & $6.0$ & $2.4$ & $4$ & $10$ & $7.6$ & $40$ & $81$ & $58.8$ & $0$ & $41$ & $8.3$ \\ 
 S20 & $2$ & $17$ & $8.8$ & $0.0$ & $7.0$ & $2.4$ & $8$ & $20$ & $13.8$ & $113$ & $225$ & $170.0$ & $0$ & $93$ & $22.2$ \\ 
 S30 & $3$ & $20$ & $10.9$ & $0.0$ & $7.0$ & $2.2$ & $12$ & $30$ & $19.3$ & $243$ & $396$ & $317.2$ & $0$ & $161$ & $40.7$ \\ 
 S40 & $3$ & $21$ & $12.4$ & $0.0$ & $8.0$ & $1.8$ & $14$ & $40$ & $24.2$ & $397$ & $626$ & $492.8$ & $0$ & $294$ & $60.5$ \\ 
 S50 & $5$ & $28$ & $14.0$ & $0.0$ & $7.0$ & $1.4$ & $18$ & $49$ & $28.7$ & $567$ & $875$ & $691.0$ & $0$ & $398$ & $85.2$ \\ 
 S60 & $6$ & $28$ & $15.2$ & $0.0$ & $7.0$ & $1.2$ & $21$ & $56$ & $33.0$ & $730$ & $1105$ & $910.5$ & $0$ & $470$ & $106.9$ \\ 
 S70 & $6$ & $29$ & $16.8$ & $0.0$ & $6.0$ & $0.9$ & $24$ & $68$ & $37.0$ & $955$ & $1333$ & $1151.6$ & $0$ & $602$ & $140.0$ \\ 
 S80 & $7$ & $28$ & $17.8$ & $0.0$ & $5.0$ & $0.7$ & $24$ & $76$ & $40.6$ & $1164$ & $1683$ & $1411.4$ & $0$ & $670$ & $163.2$ \\ 
 S90 & $9$ & $32$ & $18.9$ & $0.0$ & $7.0$ & $0.5$ & $26$ & $75$ & $44.4$ & $1455$ & $1981$ & $1683.9$ & $0$ & $825$ & $186.4$ \\ 
 S100 & $8$ & $34$ & $20.0$ & $0.0$ & $5.0$ & $0.4$ & $30$ & $74$ & $47.8$ & $1704$ & $2276$ & $1974.6$ & $1$ & $951$ & $219.4$ \\ 
 S200 & $13$ & $55$ & $28.2$ & $0.0$ & $2.0$ & $0.0$ & $51$ & $119$ & $78.5$ & $4865$ & $6375$ & $5622.4$ & $1$ & $2428$ & $566.6$ \\ 
 S300 & $19$ & $55$ & $34.7$ & $0.0$ & $1.0$ & $0.0$ & $70$ & $165$ & $104.3$ & $9460$ & $12079$ & $10366.3$ & $2$ & $6095$ & $1015.7$ \\ 
 S400 & $22$ & $62$ & $40.0$ & $0.0$ & $1.0$ & $0.0$ & $86$ & $203$ & $126.3$ & $14802$ & $17981$ & $16005.4$ & $0$ & $7130$ & $1503.9$ \\ 
 S500 & $24$ & $76$ & $44.8$ & $0.0$ & $0.0$ & $0.0$ & $107$ & $223$ & $147.7$ & $20625$ & $24592$ & $22358.6$ & $2$ & $9214$ & $2056.5$ \\ 
 S600 & $30$ & $71$ & $49.4$ & $0.0$ & $0.0$ & $0.0$ & $124$ & $292$ & $165.9$ & $27514$ & $32564$ & $29406.9$ & $2$ & $11850$ & $2603.3$ \\ 
 S700 & $33$ & $80$ & $52.7$ & $0.0$ & $0.0$ & $0.0$ & $135$ & $264$ & $183.7$ & $35117$ & $39710$ & $37086.4$ & $2$ & $14023$ & $3201.3$ \\ 
 S800 & $35$ & $83$ & $57.3$ & $0.0$ & $0.0$ & $0.0$ & $145$ & $291$ & $200.3$ & $42579$ & $49508$ & $45308.3$ & $0$ & $19508$ & $3883.3$ \\ 
 S900 & $30$ & $81$ & $59.5$ & $0.0$ & $0.0$ & $0.0$ & $160$ & $305$ & $216.4$ & $50957$ & $59447$ & $54104.8$ & $2$ & $22087$ & $4693.8$ \\ 
 S1000 & $40$ & $89$ & $63.5$ & $0.0$ & $0.0$ & $0.0$ & $176$ & $350$ & $230.6$ & $60181$ & $69426$ & $63364.8$ & $2$ & $26076$ & $5163.1$ \\ 
 \hline\hline \end{tabular}} \caption{Results for generous stable matchings over various measures.} \label{sm_rm_res_GEN} \end{table} 
 \end{landscape}

   \begin{landscape}
 \thispagestyle{empty}
\begin{table}[] \centerline{\begin{tabular}{ R{1cm} | R{0.9cm} R{0.9cm} R{0.9cm} R{1.5cm} R{0.9cm} R{0.9cm} R{1.5cm} R{0.9cm} R{0.9cm} R{1.8cm} R{1.4cm} R{1.4cm} R{1.7cm} R{1.4cm} R{1.4cm} }\hline\hline & \multicolumn{3}{c}{$f$} & \multicolumn{3}{c}{$l_{20}$} & \multicolumn{3}{c}{Degree} & \multicolumn{3}{c}{Cost} & \multicolumn{3}{c}{Sex-equal score} \\ 
Case & Min & Max & Mean & Min & Max & Mean & Min & Max & Mean & Min & Max & Mean & Min & Max & Mean \\ 
\hline S10 & $0$ & $12$ & $6.1$ & $0.0$ & $3.0$ & $0.5$ & $4$ & $10$ & $8.2$ & $40$ & $79$ & $59.9$ & $0$ & $41$ & $9.5$ \\ 
 S20 & $2$ & $18$ & $8.9$ & $0.0$ & $4.0$ & $0.4$ & $8$ & $20$ & $15.3$ & $113$ & $224$ & $173.7$ & $0$ & $121$ & $27.3$ \\ 
 S30 & $4$ & $20$ & $11.0$ & $0.0$ & $5.0$ & $0.3$ & $12$ & $30$ & $22.0$ & $243$ & $416$ & $323.9$ & $0$ & $228$ & $50.8$ \\ 
 S40 & $3$ & $22$ & $12.6$ & $0.0$ & $3.0$ & $0.3$ & $17$ & $40$ & $28.2$ & $402$ & $693$ & $504.4$ & $0$ & $403$ & $83.8$ \\ 
 S50 & $5$ & $25$ & $14.2$ & $0.0$ & $5.0$ & $0.2$ & $19$ & $50$ & $33.9$ & $570$ & $924$ & $709.6$ & $0$ & $537$ & $123.9$ \\ 
 S60 & $6$ & $30$ & $15.6$ & $0.0$ & $3.0$ & $0.2$ & $23$ & $60$ & $39.8$ & $756$ & $1232$ & $938.0$ & $0$ & $736$ & $173.1$ \\ 
 S70 & $8$ & $30$ & $17.0$ & $0.0$ & $2.0$ & $0.2$ & $25$ & $70$ & $45.0$ & $970$ & $1487$ & $1186.7$ & $0$ & $973$ & $222.8$ \\ 
 S80 & $7$ & $30$ & $18.2$ & $0.0$ & $2.0$ & $0.1$ & $28$ & $80$ & $50.0$ & $1164$ & $1858$ & $1457.6$ & $0$ & $1171$ & $280.2$ \\ 
 S90 & $9$ & $32$ & $19.5$ & $0.0$ & $4.0$ & $0.1$ & $31$ & $90$ & $55.4$ & $1447$ & $2484$ & $1744.8$ & $1$ & $1850$ & $355.7$ \\ 
 S100 & $7$ & $34$ & $20.5$ & $0.0$ & $2.0$ & $0.1$ & $34$ & $100$ & $60.5$ & $1663$ & $2604$ & $2045.8$ & $1$ & $1553$ & $415.6$ \\ 
 S200 & $14$ & $52$ & $29.5$ & $0.0$ & $2.0$ & $0.1$ & $57$ & $199$ & $105.2$ & $5006$ & $8226$ & $5917.3$ & $1$ & $6108$ & $1477.0$ \\ 
 S300 & $21$ & $69$ & $36.8$ & $0.0$ & $3.0$ & $0.0$ & $79$ & $294$ & $144.4$ & $9541$ & $17224$ & $11046.9$ & $4$ & $13374$ & $3115.2$ \\ 
 S400 & $21$ & $72$ & $42.8$ & $0.0$ & $3.0$ & $0.0$ & $99$ & $393$ & $178.1$ & $14934$ & $26115$ & $17149.7$ & $17$ & $20951$ & $5052.7$ \\ 
 S500 & $25$ & $99$ & $48.3$ & $0.0$ & $9.0$ & $0.0$ & $110$ & $496$ & $212.0$ & $20725$ & $48487$ & $24257.5$ & $15$ & $43025$ & $7752.5$ \\ 
 S600 & $31$ & $102$ & $53.5$ & $0.0$ & $4.0$ & $0.0$ & $127$ & $595$ & $244.3$ & $28028$ & $58399$ & $32052.6$ & $14$ & $50581$ & $10459.2$ \\ 
 S700 & $32$ & $113$ & $58.0$ & $0.0$ & $2.0$ & $0.0$ & $141$ & $665$ & $276.7$ & $35634$ & $71039$ & $40774.9$ & $34$ & $60999$ & $13863.8$ \\ 
 S800 & $37$ & $129$ & $63.0$ & $0.0$ & $3.0$ & $0.0$ & $162$ & $797$ & $304.0$ & $42713$ & $105477$ & $50215.3$ & $10$ & $95243$ & $18093.7$ \\ 
 S900 & $36$ & $136$ & $66.0$ & $0.0$ & $2.0$ & $0.0$ & $174$ & $833$ & $331.2$ & $51568$ & $117599$ & $60037.6$ & $46$ & $104663$ & $21334.7$ \\ 
 S1000 & $44$ & $163$ & $71.5$ & $0.0$ & $3.0$ & $0.0$ & $191$ & $975$ & $362.5$ & $60270$ & $155476$ & $71456.3$ & $11$ & $142570$ & $27270.2$ \\ 
 \hline\hline \end{tabular}} \caption{Results for median stable matchings over various measures.} \label{sm_rm_res_GM} \end{table} 
 \end{landscape}

  \begin{landscape}
  \thispagestyle{empty}
\begin{table}[] \centerline{\begin{tabular}{ R{1.2cm} | R{1.5cm} R{2.5cm} R{2cm} R{2cm} R{2cm} R{2.5cm} R{2cm} R{2cm} R{2cm} }\hline\hline && \multicolumn{4}{c}{Exponential weight} & \multicolumn{4}{c}{Vector-based weight} \\ 
Case & $N_I$ & Mean & Median & $5$th & $95$th & Mean & Median & $5$th & $95$th \\ 
\hline S10 & $821$ & $124.9$ & $116.0$ & $43.0$ & $249.0$ & $243.1$ & $228.0$ & $126.0$ & $424.0$ \\ 
 S20 & $970$ & $447.4$ & $420.0$ & $99.0$ & $870.6$ & $573.2$ & $548.0$ & $204.0$ & $1028.0$ \\ 
 S30 & $992$ & $1039.3$ & $1007.0$ & $303.3$ & $1853.3$ & $964.1$ & $952.0$ & $380.0$ & $1589.8$ \\ 
 S40 & $999$ & $1945.7$ & $1905.0$ & $641.0$ & $3294.2$ & $1582.7$ & $1580.0$ & $634.8$ & $2514.4$ \\ 
 S50 & $1000$ & $3107.8$ & $3071.5$ & $1241.8$ & $4959.9$ & $2145.9$ & $2138.0$ & $1035.4$ & $3248.3$ \\ 
 S60 & $1000$ & $4575.3$ & $4534.0$ & $2021.5$ & $7405.2$ & $2733.3$ & $2708.0$ & $1464.0$ & $4168.1$ \\ 
 S70 & $1000$ & $6578.8$ & $6552.5$ & $3251.9$ & $9954.4$ & $3693.3$ & $3662.0$ & $1978.6$ & $5504.5$ \\ 
 S80 & $1000$ & $8757.7$ & $8792.0$ & $4606.9$ & $13191.3$ & $4588.7$ & $4560.0$ & $2639.4$ & $6512.8$ \\ 
 S90 & $1000$ & $11142.1$ & $11133.0$ & $5852.9$ & $16843.2$ & $5332.8$ & $5336.0$ & $3085.6$ & $7521.6$ \\ 
 S100 & $1000$ & $14017.8$ & $14017.0$ & $7619.8$ & $20304.6$ & $6178.5$ & $6160.0$ & $3824.0$ & $8480.8$ \\ 
 S200 & $1000$ & $59602.8$ & $60321.5$ & $36976.8$ & $80033.2$ & $17044.9$ & $17140.0$ & $12061.4$ & $21662.1$ \\ 
 S300 & $1000$ & $136097.0$ & $137972.5$ & $90976.8$ & $177193.6$ & $31607.3$ & $31892.0$ & $23001.8$ & $39012.4$ \\ 
 S400 & $1000$ & $244864.3$ & $248596.0$ & $173835.2$ & $307019.8$ & $46661.3$ & $46880.0$ & $37165.4$ & $54896.6$ \\ 
 S500 & $1000$ & $378996.3$ & $381381.5$ & $274100.7$ & $472220.2$ & $61405.3$ & $61654.0$ & $49173.8$ & $72612.8$ \\ 
 S600 & $1000$ & $546143.5$ & $549124.5$ & $397475.9$ & $670184.7$ & $84767.9$ & $85786.0$ & $67468.6$ & $99364.5$ \\ 
 S700 & $1000$ & $737304.2$ & $742278.5$ & $565928.8$ & $892021.2$ & $103445.1$ & $103565.0$ & $86707.3$ & $119857.8$ \\ 
 S800 & $999$ & $945159.5$ & $954683.0$ & $704833.1$ & $1149060.4$ & $121385.7$ & $121794.0$ & $101562.0$ & $139396.8$ \\ 
 S900 & $1000$ & $1194184.0$ & $1202831.0$ & $935470.5$ & $1432438.2$ & $141469.8$ & $141310.0$ & $121259.1$ & $160431.9$ \\ 
 S1000 & $999$ & $1472310.9$ & $1488288.0$ & $1126992.0$ & $1749723.5$ & $161565.1$ & $162486.0$ & $136646.2$ & $181580.0$ \\ 
 S2000 & $5$ & $5634355.8$ & $5817101.0$ & $4822710.8$ & $6399964.2$ & $417891.2$ & $414424.0$ & $399376.0$ & $443734.4$ \\ 
 S3000 & $5$ & $11823955.4$ & $11868313.0$ & $11037121.2$ & $12496744.2$ & $738962.0$ & $735218.0$ & $696282.8$ & $776654.0$ \\ 
 S4000 & $5$ & $19613652.4$ & $18888936.0$ & $18401424.6$ & $21468356.2$ & $1029854.0$ & $1029180.0$ & $980168.8$ & $1071605.2$ \\ 
 S5000 & $3$ & $30767743.3$ & $30472468.0$ & $30221465.2$ & $31520714.2$ & $1434044.0$ & $1430948.0$ & $1424421.2$ & $1445834.0$ \\ 
 \hline\hline \end{tabular}} \caption[Comparison of the minimum number of bits required to store edge capacities of a network and vb-network.]{Comparison of the minimum number of bits required to store edge capacities of a network (exponential weight edge capacities) and vb-network (vector-based weight edge capacities). In this table, $5$th and $95$th refer to the $5$th and $95$th percentiles respectively, and $N_I$ denotes the number of instances that did not timeout and had at least one rotation, and were thus used in space requirement calculations.} \label{sm_rm_space_table} \end{table} 
 \end{landscape}

\end{document}